\PassOptionsToPackage{unicode}{hyperref}
\PassOptionsToPackage{naturalnames}{hyperref}

\documentclass[a4paper,11pt]{article}

\usepackage{etex}
\usepackage{graphicx}
\usepackage{amsmath}
\usepackage{amssymb}
\usepackage{amsthm}
\usepackage{empheq}
\usepackage{multirow}
\usepackage{makecell}
\usepackage{pict2e}
\usepackage{booktabs}
\usepackage[figuresright]{rotating}
\usepackage{bbm}
\usepackage{tikz}

\usetikzlibrary{calc,trees,positioning,arrows,chains,shapes.geometric,%
    decorations.pathreplacing,decorations.pathmorphing,shapes,%
    matrix,shapes.symbols,external,fadings,patterns}

\usepackage{algorithm}      
\usepackage{pgfplots}
\usepackage{listings}		
\usepackage[noend]{algpseudocode}
\usepackage{changepage}     
\usepackage{color,colortbl}              
\definecolor{myred}{gray}{0} 
\definecolor{light-gray}{gray}{0.7}

\usepackage[caption=false]{subfig}
\usepackage{epstopdf}
\usepackage{graphics} 
\usepackage{epsfig} 
\usepackage[ plainpages = false, pdfpagelabels,
	     bookmarks,
	     bookmarksopen = true,
	     bookmarksnumbered = true,
	     breaklinks = true,
	     linktocpage,
	     pagebackref,         
	     colorlinks = false,  
         hypertexnames=false,   
	     linkcolor = blue,
	     urlcolor  = blue,
	     citecolor = red,
	     anchorcolor = green,
	     hyperindex = true,
	     hyperfigures
	     ]{hyperref}
\usepackage{pdflscape}
\usepackage{todonotes}\reversemarginpar
\usepackage{overpic}
\makeatletter
\newenvironment{subfigures}
 {\begin{minipage}{\columnwidth}\def\@captype{figure}\centering}
 {\end{minipage}}
\makeatother

\newtheorem{theorem}{Theorem}[section]

\theoremstyle{remark}
\newtheorem{remark}[theorem]{Remark}

\theoremstyle{definition}

\newcommand{\vect}[1]{\textrm{\boldmath${#1}$}} 
\newcommand{\Vect}[1]{{\bf #1}} 

\newcommand\pd[2]{\dfrac{\partial {#1}}{\partial {#2}}} 

\newcommand\Vel{\vect{u}}   
\newcommand\vel{u}          

\newcommand\Var{\vect{w}}   

\let\originalleft\left
\let\originalright\right
\renewcommand{\left}{\mathopen{}\mathclose\bgroup\originalleft}
\renewcommand{\right}{\aftergroup\egroup\originalright}

\binoppenalty=\maxdimen 
\relpenalty=\maxdimen   

\numberwithin{equation}{section}    
\date{\today}             

\title{Projective Integration Schemes \\ for Hyperbolic Moment Equations}
\author{
Julian Koellermeier\footnote{Corresponding author, email address {\tt julian.koellermeier@kuleuven.be}} \footnote{Department of Computer Science, KU Leuven},
Giovanni Samaey\footnotemark[\value{footnote}]
}

\begin{document}

\maketitle


\begin{abstract}
In this paper, we apply projective integration methods to hyperbolic moment models of the Boltzmann equation and the BGK equation, and investigate the numerical properties of the resulting scheme. Projective integration is an explicit scheme that is tailored to problems with large spectral gaps between slow and (one or many) fast eigenvalue clusters of the model. The spectral analysis of a linearized moment model clearly shows spectral gaps and reveals the multi-scale nature of the model for which projective integration is a matching choice. The combination of the non-intrusive projective integration method with moment models allows for accurate, but efficient simulations with significant speedup, as demonstrated using several 1D and 2D test cases with different collision terms, collision frequencies and relaxation times.

\end{abstract}

\section{Introduction}
\label{sec:intro}

Kinetic equations are widely used for applications in science and engineering \cite{Aoki2007,Bouffanais2016}, e.g., for the description of flows under rarefied conditions \cite{Au2001,Boyd2003,Westerkamp2012}.
The collision term on the right-hand side of the kinetic equation can be a full Boltzmann collision operator or some simplified version, e.g., the BGK collision operator. The collision operator is multiplied by a potentially large collision frequency that leads to models that are stiff close to equilibrium, where the model should converge to the hydrodynamic limit. This stiffness is the reason for a severe time step constraint that needs to be overcome, especially as the limiting Euler equations do not exhibit such a time step constraint.
It is thus necessary to employ an asymptotic-preserving numerical scheme \cite{Jin2010} for which the computational complexity is bounded when approaching the typically stiff hydrodynamic limit.
Some numerical schemes mitigate this stiffness by using a splitting algorithm \cite{Cai2013,Tcheremissine2001,Xu2010}. This is useful when considering semi-lagrangian schemes, see \cite{Dimarco2018}. In the Eulerian setting
however, splitting is difficult for higher-order schemes and an exact solution of the split collision term is only possible for certain (linear, or simplified) collision operators. It is thus not feasible for the full Boltzmann collision operator, which is non-linear.
We do not want to use implicit asymptotic preserving schemes such as IMEX \cite{Pareschi2005} or the implicit Galerkin method in \cite{Grohs2017} as they increase the runtime and do not correspond to the hyperbolic nature of our system of equations.
It is furthermore important for the method to be as little intrusive as possible to allow for broad applicability.

Projective integration (PI) was successfully demonstrated as a stable scheme for discrete velocity models (DVM) \cite{Lafitte2016,Lafitte2010,Melis2019,Melis2016}, in which the velocity variable was discretized based on point values.
Note that, in part of the literature, the acronym DVM is only used for models that mimic the basic properties of the kinetic equation including convergence to the fluid limit. 
We refer to the lecture notes \cite{cabannes1980} for more background information.
Projective integration takes a few small time steps with a time step size corresponding to the (stiff) small scale of the model to damp the fast modes, before extrapolation using a large CFL-type time step size. The slow modes then are treated with sufficient accuracy while the fast modes do not spoil the stability. PI schemes were applied to this class of DVM kinetic models successfully, which points towards the potential to be used for other models as well. The PI method is non-intrusive in the sense that it requires minimal changes in the implementation of the numerical time-stepping method and only needs some necessary information about the spectrum of the model to deal with the present slow and fast modes. 

The DVM needs a fine discretization of velocity space \cite{Baranger2012,Mieussens2000,Dubroca1999}.
One reason for that is the poor approximation quality of the distribution function for small number of discrete velocities. Another reason is the fact that the mean velocity as well as the temperature of the fluid vary throughout the simulation and the global velocity grid needs to cover all cases.
An alternative is moment models that result in hierarchical systems of equations with only few variables \cite{Grad1949,Struchtrup2006,Torrilhon2016}. The general idea of a moment model is not to discretize the distribution function of the kinetic equation using point values of the microscopic velocity space but using higher order moments of the distribution function. This allows for the reduction to a small set of explicit equations with direct physical insight.
For a long time, lack of hyperbolicity was a major disadvantage and many alternative models where not computationally efficient, e.g., the maximum entropy models \cite{Levermore1996,McDonald2013}. The recent development of hyperbolic moment models, however, solved this problem and led to many promising models and applications \cite{Cai2013,Cai2014a,Fan2016,Koellermeier2014}. This made moment models accessible for further improvement by combination with other high-fidelity numerical methods. The final goal is to develop a numerical solver that uses moment models bridging from the continuum to the rarefied regime with the help of appropriate switching criteria, e.g., as suggested in \cite{Lockerby2009}.

In this paper, we will use projective integration schemes for different moment models to demonstrate the capability of overcoming stiffness near equilibrium and achieve significant speedup in comparison to the forward Euler time-stepping scheme. This constitutes the first application of higher-order, explicit PI schemes for moment models.

The rest of the paper is organized as follows: In section \ref{sec:models} we introduce the hyperbolic moment models together with the Boltzmann and BGK collision operators. The different PI schemes used in this paper are described in section \ref{sec:Num}. We analyze the stability of the involved moment models using a linearized moment model and a variety of collision terms that vary in their collision frequency and relaxation time in section \ref{sec:LSA}. That way the application of the proper PI scheme will be made possible for a wide range of test cases. We give explicit guidelines how to choose the parameters and apply the different schemes to a 1D shock tube test case, a 1D two-beam test case and a 2D forward facing step test case in section \ref{sec:NumEx}. 
The stability and speedup of the PI schemes during the simulations will be shown using our results and they open up many possibilities for further work. This work is the necessary step towards taking full advantage of moment models and accelerating the runtime of solution schemes near equilibrium. The long-term goal for future work is to exploit the hierarchical structure of moment models by means of an adaptive moment method that covers a range of moment models from equilibrium to the kinetic regime.



\section{Model equations}
\label{sec:models}
In this section, we introduce the kinetic equation and the moment model as an efficient discretization in velocity space. As the focus of this paper is not the precise form of the model equations but the applicability, parameter choice, and speedup of projective integration, we follow \cite{Koellermeier2017,Melis2019} and focus on a concise presentation of the models with only necessary details. We will describe the general multi-dimensional case where possible while showing several 1D test results and a 2D test case at the end of the paper.

\subsection{Kinetic equation}
We consider the evolution of the mass density distribution function $f(t,\vect{x},\vect{c})$ given by the kinetic transport equation \cite{Cercignani1994}
\begin{equation}
\label{e:BTE}
    \frac{\partial}{\partial t} f(t,\vect{x},\vect{c}) + \displaystyle \sum_{i=1}^D c_i \frac{\partial}{\partial x_i} f(t,\vect{x},\vect{c}) = \frac{1}{\tau}S(f),
\end{equation}
where $\vect{x} \in \mathbb{R}^D$ and $\vect{c} \in \mathbb{R}^D$ denote position and microscopic velocity, respectively. Two right-hand side collision operators $S(f)$ are specified in the next sections. The collision operator typically drives the distribution function closer towards the equilibrium Maxwellian $f_{\text{Maxwell}}$ given by
\begin{equation}
\label{e:Maxwellian}
    f_{\text{Maxwell}}(t,\vect{x},\vect{c})=\frac{\rho(t,\vect{x})}{\sqrt{2\pi\theta(t,\vect{x})}^D}\exp\left(
    -\frac{|\vect{c}-\Vel(t,\vect{x})|^2}{2\theta(t,\vect{x})} \right).
\end{equation}

The macroscopic quantities density $\rho(t,\vect{x})$, velocity $\Vel(t,\vect{x})$, and temperature $\theta(t,\vect{x})$ can be obtained by integration of $f(t, \vect{x}, \vect{c})$ over velocity space as follows:
\begin{align}
\label{e:density}
    \rho(t,\vect{x}) &= \int_{\mathbb{R}^D}f(t, \vect{x}, \vect{c}) \, d\vect{c}, \\
\label{e:velocity}
    \rho(t,\vect{x})\Vel(t,\vect{x}) &= \int_{\mathbb{R}^D}\vect{c} f(t, \vect{x}, \vect{c}) \, d\vect{c}, \\
\label{e:temperature}
    \frac{D}{2}\rho(t,\vect{x})\theta(t,\vect{x}) + \frac{1}{2}\rho(t,\vect{x})|\Vel(t,\vect{x})|^2 &= \int_{\mathbb{R}^D}\frac{1}{2}|\vect{c}|^2f(t, \vect{x}, \vect{c}) \, d\vect{c}.
\end{align}

The parameter $\tau > 0$ can be seen as the dimensionless Knudsen number, i.e. the ratio of the particles' mean free path length and a reference length. The dimensionless Knudsen number $\tau$ is a measure for the relaxation time towards equilibrium and therefore defines the regime of the flow. From the kinetic regime ($\tau \geq 10^{-1}$) via the transitional regime ($\tau \in [10^{-4},10^{-1}]$) to the hydrodynamic regime ($\tau \leq 10^{-4}$) \cite{Melis2019}, the kinetic equation converges to the well-known Euler equations for ideal gases in the limit of infinitely small relaxation time $\tau \rightarrow 0$.

Macroscopic equations are already contained in \eqref{e:BTE}, which can be seen by multiplying with monomials $(1,\vect{c},|\vect{c}|^2/2)^T$ and integration of both sides over velocity space \cite{Struchtrup2006}. This leads to the well-known macroscopic conservation laws of mass, momentum and energy. In primitive variables and in non-conservative form the limiting macroscopic equations can be written as
\begin{align}
\label{e:cons_mass}
    \pd{\rho}{t} + \sum_{d=1}^D\pd{\rho \vel_d}{x_d} &=0, \\
\label{e:cons_momentum}
    \rho\pd{\vel_i}{t} + \sum_{d=1}^D\left( \rho \vel_d\pd{\vel_i}{x_d}+\pd{p_{i,d}}{x_d} \right) &= 0, \quad i = 1, \dots, D, \\
\label{e:cons_energy}
    \frac{D\rho}{2}\pd{\theta}{t} + \sum_{d=1}^D\left( \frac{D}{2}\rho \vel_d\pd{\theta}{x_d} + \pd{q_d}{x_d} \right) + \sum_{d=1}^D\sum_{k=1}^D p_{k,d}\pd{\vel_k}{x_d} &=0,
\end{align}
where the pressure is denoted as $p_{i,j}$, and the heat flux as $q_i$ for $i,j = 1, \dots, D$, using the definitions
\begin{align}
\label{e:pressure}
    p_{i,j} &= \int_{\mathbb{R}^D}f(t, \vect{x}, \vect{c})(c_i-\vel_i)(c_j-\vel_j) \, d\vect{c}, \\
\label{e:heat_flux_2}
    q_i &= \int_{\mathbb{R}^D}f(t, \vect{x}, \vect{c})|\vect{c}-\Vel|^2(c_i-\vel_i) \, d\vect{c}.
\end{align}
Note that the equations are not closed because pressure and heat flux require full knowledge of the distribution function. Closing the system by assuming an ideal gas law for the pressure and zero heat flux then results in the Euler equations.

Notice that the respective right-hand sides of equations \eqref{e:cons_mass}-\eqref{e:cons_energy} equal zero as the monomials $(1,\vect{c},|\vect{c}|^2/2)^T$ are so-called \emph{collision invariants}, for which the respective integrals of the collision operator vanish. Equations \eqref{e:cons_mass}-\eqref{e:cons_energy} describe the slow modes in our models, which propagate according to the macroscopic variables. When integrating the collision operator multiplied with higher order monomials, the right-hand side does not vanish and the higher order equations that are used to describe deviations from the equilibrium state then contain fast relaxing modes, as will be explained after a more detailed description of the collision operators.

\subsection{Boltzmann collision operator}
\label{sec:BTE}
For the Boltzmann collision operator, we only describe the 2D version, which will later be used in the numerical tests.
The Boltzmann collision operator models elastic binary collisions between particles with pre-collision velocities $\left(c', c'_1\right)$ and post-collision velocities $\left(c, c_1\right)$ \cite{Cercignani1994}. In the 2D setting, they can be related by \cite{Melis2019}
\begin{equation}
    \vect{c}' = \frac{\vect{c} + \vect{c}_1}{2} + \frac{|\vect{c} + \vect{c}_1|}{2}\vect{\sigma}, \quad\quad \vect{c}'_1 = \frac{\vect{c} + \vect{c}_1}{2} - \frac{|\vect{c} + \vect{c}_1|}{2}\vect{\sigma},
\end{equation}
with two-dimensional unit vector $\vect{\sigma}$, that points into the direction of the pre-collisional relative velocity $\vect{c}'_r = \vect{c}' - \vect{c}'_1$, such that
\begin{equation}
    \vect{\sigma} = \frac{\vect{c}'_r}{|\vect{c}'_r|}.
\end{equation}

The Boltzmann collision operator in 2D is then given by
\begin{equation}
\label{e:Boltzmann_collision}
    S(f) = \int_{\mathbb{R}^2} \int_{0}^{2\pi} \left(f' f'_1 - f f_1\right) B\left(|\vect{c}-\vect{c}_1|,\theta_{\sigma}\right) \, d\theta_{\sigma}d\vect{c}_1,
\end{equation}
where $f', f'_1$ are post-collision distribution functions and $f, f_1$ represent the pre-collision distribution functions, $\theta_{\sigma}$ is the angle between $c'_r$ and $\vect{\sigma}$, and $B\left(|\vect{c}_r|,\theta_{\sigma}\right)$ is the collision kernel. The numerical method in this paper can be used for different collision kernels. However, we assume pseudo-Maxwellian particles in this paper. The kernel then simplifies to $B\left(|\vect{c}-\vect{c}_1|,\theta_{\sigma}\right) = b_0$. Note that the choice of the collision kernel does influence the form of the Boltzmann collision operator but not its separation of fast and slow scales as the collision invariants are still the same. For more details on how to choose the collision kernel and how this would influences the choice of the numerical method later, we refer to \cite{Melis2019}.

The high-dimensional integral in equation \eqref{e:Boltzmann_collision} is expensive to evaluate computationally. This is especially problematic if point evaluations are needed, for example, in a DVM method \cite{Baranger2012,Mieussens2000}. For a moment model, however, the projected integrals can be evaluated offline beforehand leading to a speed-up of the collision term computation. More details about further speedup of the collision term can be found in \cite{Cai2014c,Cai2015a}.

When splitting equation \eqref{e:Boltzmann_collision} into a gain term and a loss term, the loss term includes as proportionality factor a collision frequency $\nu$, which can be computed explicitly for the pseudo-Maxwellian collision kernel as
\begin{equation}\label{e:collision_frequency}
  \nu = 2 \pi b_0 \rho,
\end{equation}
and will appear again in the simplified model described in the next section.

\subsection{BGK collision operator}
\label{sec:BGK}
A simpler collision model is the so-called BGK model \cite{Bhatnagar1954}, describing a relaxation towards the equilibrium Maxwellian \eqref{e:Maxwellian} as follows
\begin{equation}
\label{e:BGK}
    \frac{1}{\tau} S(f) = - \frac{\nu}{\tau} \left( f - f_{\text{Maxwell}} \right).
\end{equation}
Modifications of the model are possible, leading to the so-called ES-BGK or Shakov model \cite{Andries2000}. The collision frequency $\nu$ can be chosen in accordance with the collision frequency of the Boltzmann collision operator. When chosing $\nu = \rho$, the BGK model matches the loss term of the Boltzmann collision operator from above. A constant collision frequency $\nu=const$ leads to a simpler model. However, the model is not linear as the Maxwellian on the right-hand side contains the macroscopic moments of $f$.

Point evaluations of the BGK operator \eqref{e:BGK} are simpler than for the Boltzmann equation, but discrete values still need to ensure conservation of mass, momentum, and energy throughout the simulation by a special projection procedure. This will be much simpler for moment models, where the BGK operator become a linear, diagonal operator and can be explicitly derived beforehand.

\subsection{Hyperbolic moment models}
Moment models have a clear advantage over DVM models when it comes to the necessary number of variables, the evaluation of the collision operator, and the approximation quality. The reason is that standard DVM models need many variables and moment models can reduce the number of necessary variables drastically to the expense of a more complex, possibly non-linear model \cite{Torrilhon2016}. For more results on the accuracy and convergence of the moment models, the interested reader is referred to the literature, e.g., \cite{Bourgault2015,Cai2019,Sarna2020}.

To derive the additional equations for deviations from equilibrium, the distribution function is expanded around the local equilibrium using a sum of basis functions \cite{Grad1949} $\phi^{[\vect{\vel}(t,\vect{x}),\theta(t,\vect{x})]}_{\vect{\alpha}}$
\begin{equation}
\label{e:vars_expansion}
    f(t,\vect{x},\vect{c}) = \sum_{\vect{\alpha} \in \mathbb{M}} f_{\vect{\alpha}}(t,\vect{x}) \phi^{[\vect{\vel},\theta]}_{\vect{\alpha}}\left(\vect{\xi}\right),
\end{equation}
with coefficients $f_{\vect{\alpha}}(t,\vect{x})$, which are also called \emph{moments}, and weighted Hermite basis functions \cite{Cai2014a,Koellermeier2017} defined as
\begin{equation}
\label{e:vars_basis}
    \phi^{[\vect{\vel},\theta]}_{\vect{\alpha}}\left(\vect{\xi}\right) = \prod_{d=1}^{D}\frac{1}{\sqrt{2\pi \theta^{\alpha_d+1}}}  He_{\alpha_d}(\xi_d) \exp\left(-\frac{-\xi_d^2}{2}\right)
\end{equation}
for one-dimensional Hermite polynomials
\begin{equation}
\label{e:vars_Hermite}
    He_{\alpha_d}(\xi_d) = (-1)^k \exp \left(\frac{\xi_d^2}{2}\right)\frac{d^k}{dx^k} \exp \left(-\frac{\xi_d^2}{2}\right).
\end{equation}

The coefficients $f_{\vect{\alpha}}(t,\vect{x})$ in the ansatz \eqref{e:vars_expansion} use a multi-index $\vect{\alpha} \in \mathbb{M}$ from an index set $\mathbb{M} \in \mathbb{N}^ D$ that defines the used moment theory. According to \cite{Koellermeier2018,Torrilhon2015}, different moment theories are possible. We use the so-called \emph{full moments}, corresponding to using full tensors in a spherical harmonics expansion
\begin{equation}
    \label{e:ansatz_full}
    \mathbb{M} = \left\{ \vect{\alpha} \in \mathbb{N}^D, |\vect{\alpha}| \leq M \right\},
\end{equation}
which have the benefit to be rotationally invariant in a multi-dimensional setting.
The transformed velocity $\vect{\xi}$ allows for an efficient discretization in velocity space \cite{Kauf2011} and is denoted as
\begin{equation}
\label{e:vars_transformed_velocity}
    \vect{\xi} = \frac{\vect{c}-\vect{\vel}}{\sqrt{\theta}}.
\end{equation}

The basis coefficients $f_{\vect{\alpha}}(t,\vect{x})$ in expansion \eqref{e:vars_expansion} depend only on $t$ and $x$, but no longer on the transformed microscopic velocity $\vect{\xi}$, which is solely encoded in the basis function. In the following we outline the derivation of evolution equations for the coefficients $f_{\vect{\alpha}}(t,\vect{x})$.

%

By ensuring that the expanded distribution function \eqref{e:vars_expansion} fulfills \eqref{e:density}-\eqref{e:temperature}, we get $D+2$ additional algebraic equations, the so-called compatibility conditions. These conditions ensure that the distribution function has the correct density, momentum and energy. With the chosen basis functions \eqref{e:vars_basis}, the compatibility conditions can be simplified according to \cite{Koellermeier2014a} and read
\begin{equation}
    \label{e:vars_comp}
    f_{\vect{0}} = \rho, \quad  f_{\vect{e}_j} = 0, \quad j = 1,\ldots,D, \quad \sum_{d=1}^{D} f_{2\vect{e}_d} = 0,
\end{equation}
for the $j$-th unit vector $\vect{e}_j \in \mathbb{N}^D, (\vect{e}_j)_i = \delta_{i,j}, ~i,j = 1,\ldots,D$.

We directly set $f_{\vect{0}} = \rho$ and $f_{\vect{e}_j}=0, ~j = 1,\ldots,D$. The last equation $\sum_{d=1}^{D} f_{2\vect{e}_d} = 0$ is automatically fulfilled by considering the pressure tensor $\vect{p}$ computed analogously to \eqref{e:pressure} in the form
\begin{equation}
\label{e:vars_pressure}
    p_{\vect{e}_i+\vect{e}_j} = \delta_{i,j} \theta + (1+\delta_{i,j}) f_{\vect{e}_i+\vect{e}_j}
\end{equation}
and then using the variables
\begin{align}\label{e:vars_newvars}
    \frac{p_{2\vect{e}_i}}{2} & \text{ instead of } f_{2\vect{e}_i}, \quad \text{ for } i = 1,\ldots,D, \\
    p_{\vect{e}_i+\vect{e}_j} & \text{ instead of } f_{\vect{e}_i+\vect{e}_j}, \quad \text{ for } i,j = 1,\ldots,D, i \neq j.
\end{align}
For more details, we refer to \cite{Koellermeier2018}.


Below we exemplify the two-dimensional and the one-dimensional cases, which will be used in the simulations in section \ref{sec:NumEx}.

%

\subsubsection{Two-dimensional moment model}
In the two-dimensional case, the full moment ansatz result in the following variable vector $\Var_M=\Var_3$ for $M=3$, which was used, e.g., in the simulations in \cite{Koellermeier2018}
\begin{equation}
\label{e:vars_full}
    \Var_{3} = \left( \rho, \vel_x, \vel_y, \frac{p_1}{2}, f_{1,1}, \frac{p_2}{2}, f_{3,0}, f_{2,1}, f_{1,2}, f_{0,3}\right)^T.
\end{equation}
for $\frac{p_1}{2} = \frac{\rho \theta}{2}+f_{2,0}$, $\frac{p_2}{2} = \frac{\rho \theta}{2}+f_{0,2}$ and $f_{i,j} = f_{i\vect{e}_1+j\vect{e}_2}$.

A closed set of equations is then derived by inserting the ansatz \eqref{e:vars_expansion} into the kinetic equation \eqref{e:BTE} and projecting onto the proper Hermite test functions
\begin{equation}
\label{e:vars_system}
    \frac{\partial \Var_M}{\partial t}  + \Vect{A}_{x} \frac{\partial \Var_M}{\partial {x}} + \Vect{A}_{y} \frac{\partial \Var_M}{\partial {y}} = \Vect{S}(\Var_M),
\end{equation}
where $\Vect{S}(\Var_M)$ results from the projection of the right-hand side collision operators from sections \ref{sec:BTE} and \ref{sec:BGK}, and $\Var_M \in \mathbb{R}^{|\mathbb{M}|}$ is the vector of unknown variables depending on the specific moment theory. The terms in equation \eqref{e:vars_system} can be found in Appendix \ref{app:2D_QBME}. We refer to \cite{Koellermeier2018} and the implementation in \cite{Koellermeier2020b} for more details.

\subsubsection{One-dimensional moment model}
In the one-dimensional case, the compatibility conditions reduce to
\begin{equation}
    \label{e:vars_comp1D}
    f_0 = \rho, \quad  f_1 = 0,  \quad f_2 = 0,
\end{equation}
which leads to the following vector of unknown variables
\begin{equation}
    \label{e:vars_full_1D}
    \Var_M = \left( \rho, \vel, \theta, f_3, \ldots, f_M \right)^T.
\end{equation}
The moment equations read
\begin{equation}
\label{e:vars_system1D}
    \frac{\partial \Var_M}{\partial t}  + \Vect{A} \frac{\partial \Var_M}{\partial x}  = \Vect{S}(\Var_M),
\end{equation}
where the collision term $\Vect{S}(\Var_M)$ for the simple BGK model is given by \cite{Koellermeier2017}
\begin{equation}
\label{e:BGK_term}
    \Vect{S}(\Var_M)  = - \frac{1}{\tau} \left( 0,0,0, f_3, \ldots, f_M \right)^T,
\end{equation}
and the system matrix by
\begin{equation}
\label{e:QBME_A}
\Vect{A} = \setlength{\arraycolsep}{1pt}
\left(
  \begin{array}{cccccccc}
    \vel & \rho &  &  &   &   &   &  \\
    \frac{\theta}{\rho} & \vel & 1 &  &  &   &   &  \\
    & 2 \theta & \vel & \frac{6}{\rho} &  &  &   &   \\
    & 4 f_3 & \frac{\rho \theta}{2} & \vel & 4 &  &  &  \\
    -\frac{\theta f_3}{\rho} & 5f_4 & \frac{3f_3}{2} & \theta & \vel & 5 &  &    \\
    \vdots & \vdots & \vdots & \vdots &  \ddots & \ddots & \ddots &  \\
    -\frac{\theta f_{M-2}}{\rho} & Mf_{M-1} & \frac{\left(M-2\right)f_{M-2}+ \theta f_{M-4}}{2} \textcolor{myred}{-\frac{M(M+1)f_{M}}{2 \theta}}& -\frac{3f_{M-3}}{\rho}  &  & \theta & \vel & M \\
    -\frac{\theta f_{M-1}}{\rho} & (M\hspace{-0.1cm}+\hspace{-0.1cm}1) f_{M} & \textcolor{myred}{-f_{M-1}} \hspace{-0.1cm} + \hspace{-0.1cm} \frac{\theta f_{M-3}}{2} & \textcolor{myred}{\frac{3(M+1)f_{M}}{\rho \theta}} \hspace{-0.1cm} - \hspace{-0.1cm} \frac{3f_{M-2}}{\rho} &  &  & \theta & \vel \\
  \end{array}
\right),
\setlength{\arraycolsep}{6pt}
\end{equation}

Note that we use the hyperbolic regularization called QBME, which was developed in \cite{Koellermeier2017,Koellermeier2014}, to obtain global hyperbolicity. The standard model \cite{Grad1949} does not yield hyperbolic equations.
As a result of the hyperbolic fix, we can explicitly evaluate the real eigenvalues of the system, see also \cite{Fan2016}. In the 1D example, they are given by the shifted and scaled roots of the Hermite polynomials of degree $M+1$, with $M$ being the highest degree within the expansion \eqref{e:vars_expansion}
\begin{equation}\label{e:HME_EV}
    \lambda_{i} = \vel + \sqrt{\theta} \, c_{i},~ i=1, \ldots, M+1,
\end{equation}
where the $c_{i}$ are the Hermite roots $\textrm{He}_{M+1}(c_{i}) = 0$.
\section{Numerical Method}
\label{sec:Num}
The moment models introduced in the previous section are characterized by a hyperbolic transport part with bounded propagation of information given by the eigenvalues and a possibly stiff right-hand side collision term, depending on the collision frequency $\nu$ and the relaxation time $\tau$ of the collision operator. We are interested in stable solutions of the model equations \eqref{e:vars_system},\eqref{e:vars_system1D} for small values of the relaxation time $\tau$ \eqref{e:BGK}. In this section, we will briefly discuss the spatial discretization and then describe the PI schemes used to overcome the stiffness of the collision operators.

\subsection{Path-conservative spatial discretization}
Due to the hyperbolic regularization of the equations, the resulting moment model \eqref{e:vars_system} contains specific terms that are added to the higher-order equations, see \cite{Koellermeier2017} for details. In turn, the left hand side of the system can no longer be written in conservative form. This gives rise to a partially-conservative system, where the first equations can be written using a flux function and the last equations are given in non-conservative form only. The non-conservative terms in this paper are discretized using a path-conservative scheme, which computes the occurring generalized Roe matrix based on a linear path connecting the left and right states of the computational cell \cite{DalMaso1995}. The method has been used in many applications \cite{Castro2017,Castro2008,Castro2012} and especially for moment models in \cite{Cai2013,Koellermeier2017a,Koellermeier2018}. It was found that the non-conservative terms do not spoil stability or accuracy of the model equations, despite the problems occurring for other non-conservative models \cite{Abgrall2010}.

After the non-conservative terms are discretized, we decide on the numerical flux. Due to the different test cases and implementations, we use two different numerical fluxes depending on the model:

The 1D test cases are performed on a equidistant grid with constant cell size $\Delta x$ and use a dedicated high-order CWENO reconstruction up to third order in space and the FORCE scheme \cite{Toro2000,Castro2006,Cravero2018,Koellermeier2020}.

The 2D test cases are performed on a non-uniform quadrilateral grid and use the first-order PRICE scheme \cite{Canestrelli2008,Canestrelli2010}. For more details we refer to the respective references and the implementation \cite{Koellermeier2020b}, which is based on the developments for \cite{Koellermeier2020}. Here we will only assume that the spatial discretization leads to a semi-discrete time integration problem of the following form
\begin{equation}\label{e:semidiscrete}
  \frac{\partial \Var_M}{\partial t} = D_t\left(\Var_M\right), \quad D_t\left(\Var_M\right) = - D_x\left(\Var_M\right) - D_y\left(\Var_M\right) + \frac{1}{\tau} S\left(\Var_M\right),
\end{equation}
where the two terms $D_x, D_y$ are the result of the spatial discretization in the respective direction and the last term represents the point evaluation or integral of the collision operator within the respective cells. 

\begin{remark}
    The two numerical methods for the 1D and the 2D test cases are not the same. One advantage of the Projective Integration method mentioned in the next section is that it can readily be applied to any existing spatial discretization with minor additional modifications, despite, e.g., non-uniform grids or higher-order reconstructions. It is therefore in the spirit of this work to highlight the applicability of the Projective Integration method and consider different numerical schemes.
\end{remark}

\begin{remark}
    It is important to emphasize again that the QBME are non-linear and non-conservative. As a consequence, their relevance can be questioned for discontinuous solutions since Rankine-Hugoniot jump relations cannot be derived, except by the use of vanishing regularisation, see \cite{LeFloch1996}. As mentioned in \cite{Abgrall2010} their numerical approximation is still an open question, even with path-conservative methods like PRICE.
    Discontinuous solutions typically arise if the number of equations is too small and the flow conditions model strong non-equilibrium, as mentioned in \cite{Cai2013}. In this paper, we are interested in the solutions relatively close to the fluid dynamic limit in equilibrium. In these situations, the solutions will typically be smooth in the non-conservative variables as those are the non-equilibrium variables and tend to zero in equilibrium. For numerical tests in the non-equilibrium regime, we refer to \cite{Cai2013,Koellermeier2017a}.
\end{remark}

\subsection{Projective integration}
For small values of the relaxation time $\tau$, the semi-discrete system \eqref{e:semidiscrete} is characterized by slow macroscopic scales related to the macroscopic transport properties and (one or more) fast microscopic scales related to the relaxation of higher order moments. The dynamics of the higher order moments pose a severe time step restriction of $\Delta t \sim \tau$, while the macroscopic variables would be efficiently integrated using a CFL-type time step size $\Delta t \sim \Delta x$, which is independent of $\tau$. Especially towards the limit $\tau \rightarrow 0$, a standard forward Euler discretization of equation \eqref{e:semidiscrete} would require more and more time steps and practically become infeasible.

Projective Integration (PI) is a time stepping scheme consisting of an inner integrator and an extrapolation step \cite{Lafitte2010,Melis2017}. It is typically used for stiff problems and overcomes the stiff time step restriction by first iterating a few small time steps with time step size $\delta t \sim \tau$ corresponding to the fast relaxation speed of the fast modes and then extrapolating the result over a large time step corresponding to a CFL type time step of the slow modes. Note that according to \cite{Melis2019}, PI is not asymptotic-preserving in the sense of \cite{Jin2010}, as the limit $\epsilon \to 0$ leads to vanishing $\delta t \to 0$ so that the inner time step will never advance. However, simulations with arbitrarily small $\epsilon$ can be performed and in most cases the cost of the scheme does not depend on the stiffness of the problem. This is closely related to the asymptotic-preserving property. The interested reader is referred to \cite{Melis2019} for more details.

The first order PI scheme uses the standard forward Euler method as inner integrator with time step size $\delta t$ for $K+1$ steps
\begin{equation}
  \Var_M^{n,k+1} = \Var_M^{n,k} + \delta t D_t\left(\Var_M^{n,k}\right), \quad k=0,1,\ldots, K.
\end{equation}

After the $K+1$ inner steps a discrete time derivative using the last two values is obtained and used in an outer step to compute the value at the new time step $\Var_M^{n+1}$ via extrapolation in time
\begin{equation}
  \Var_M^{n+1} = \Var_M^{n,K+1} + \left(\Delta t - (K+1) \delta t \right) \frac{\Var_M^{n,K+1} - \Var_M^{n,K}}{\delta t}.
\end{equation}

This method is called the Projective Forward Euler (PFE) method. The parameters of the PFE method are the inner time step size $\delta t$ and the number of inner time steps $K+1$ in addition to the outer time step size $\Delta t$.

According to \cite{Melis2017} the stability domain of the PFE method for the standard model equation $\partial_t w = \lambda w, \lambda \in \mathbb{C}$ is characterized by
\begin{equation}
\label{e:stability_conditionPFE}
  \left| \left( 1 + \left( \frac{\Delta t}{\delta t} - K \right) \lambda \delta t \right) \left( 1 + \lambda \delta t  \right)^K \right| \leq 1.
\end{equation}
The stability condition \eqref{e:stability_conditionPFE} is fulfilled for eigenvalues $\lambda$ within the union of the two discs
\begin{equation}
\label{e:stability_conditionPFE_lambda}
  \lambda \in \mathcal{D}\left( -\frac{1}{\Delta t},\frac{1}{\Delta t} \right) \cup \mathcal{D}\left( -\frac{1}{\delta t},\frac{1}{\delta t}\frac{\delta t}{\Delta t}^{\frac{1}{K}} \right),
\end{equation}
where $\mathcal{D}\left( c, r \right) \in \mathbb{C}$ denotes the disc with center $(c,0)$ and radius $r$ in the complex plane.

The stability domain containing the two discs from equation \eqref{e:stability_conditionPFE_lambda} is shown exemplarily for $\delta t =10^{-4}$, $\Delta t = 10^{-3}$, and $K=1$ in figure \ref{fig:stabilitydomain}. The PFE method is ideally suited for stable integration of models including a scale separation into one fast and one slow cluster as indicated.
\begin{figure}[htb!]
    \centering
    \includegraphics[width=0.85\linewidth]{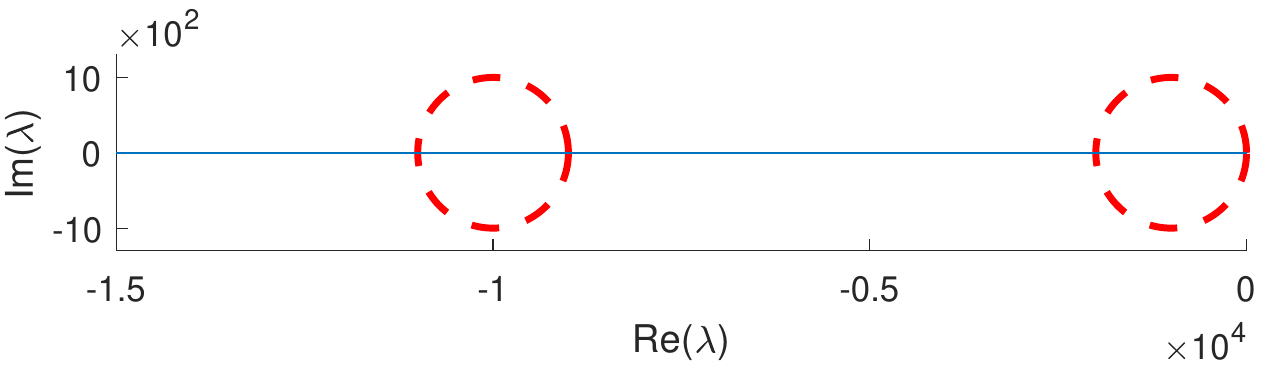}
    \caption{Stability domain \eqref{e:stability_conditionPFE_lambda} of PFE method for $\delta t =10^{-4}$, $\Delta t = 10^{-3}$, $K=1$.}
    \label{fig:stabilitydomain}
\end{figure}

While the accuracy of the inner integrator is not of interest for the overall accuracy of the scheme, the outer integrator can be generalized to a Runge-Kutta method, see \cite{Lafitte2016}. We thus employ a standard $S+1$-stage Runge-Kutta method with parameters $\vect{A} \in \mathbb{R}^{S+1 \times S+1}$, $\vect{c} \in \mathbb{R}^{S+1}$, $\vect{b} \in \mathbb{R}^{S+1}$. The result is a Projective Runge-Kutta scheme (PRK). Each stage $s$ is the result of a PFE iteration. A PRK scheme consists of $S+1$ outer stages, which each include $K+1$ small inner time steps of size $\delta t$ and one subsequent extrapolation step of the remaining time step $c_S \Delta t - (K+1)\delta t$.

The first stage's slope is
\begin{equation}
    s=0: \left\{
                \begin{array}{lcl}
                  \Var_M^{n,k+1} &=& \Var_M^{n,k} + \delta t D_t\left(\Var_M^{n,k}\right), \quad 0 \leq k \leq K \\
                  \vect{k}_1 &=& \frac{\Var_M^{n,K+1}-\Var_M^{n,K}}{\delta t}
                \end{array}
              \right.
\end{equation}
And the other stages are subsequently computed as
\begin{equation}
    2 \leq s \leq S: \left\{
                \begin{array}{lcl}
                  \Var_M^{n+c_s,0} &=& \Var_M^{n,K+1} + \left( c_s \Delta t - (K+1)\delta t \right) \displaystyle\sum_{l=1}^{s-1} \frac{a_{s,l}}{c_s} \vect{k}_l \\
                  \Var_M^{n+c_s,k+1} &=& \Var_M^{n+c_s,k} + \delta t D_t\left(\Var_M^{n+c_s,k}\right), \quad 0 \leq k \leq K \\
                  \vect{k}_s &=& \frac{\Var_M^{n+c_s,K+1}-\Var_M^{n+c_s,K}}{\delta t}
                \end{array}
              \right.
\end{equation}
The new time step is then extrapolated to
\begin{equation}
    \Var_M^{n+1} = \Var_M^{n,K+1} + (\Delta t - (K+1)\delta t) \sum_{s=1}^{S} b_s \vect{k}_s.
\end{equation}
In this paper we employ a third-order PRK3 scheme, which is based on the third-order strong stability preserving Runge-Kutta (SSPRK3) scheme and a second-order PRK2 scheme, which uses the Heun method. All PRK methods used here also have as parameters the inner time step size $\delta t$ and the number of inner time steps $K+1$ in addition to the outer time step size $\Delta t$.

The stability domain of a PRK scheme is similar to the one for the PFE scheme in equation \eqref{e:stability_conditionPFE_lambda}, but not disc-shaped \cite{Lafitte2016}. The stability domain of a PRK scheme contains the PFE stability domain, such that all parameter choices for a stable PFE method also lead to a stable PRK method.

\subsection{Telescopic projective integration}
If there is no clear separation in one cluster of fast modes and one cluster of slow modes, because there are either multiple clusters or there is an extended spectrum of modes (i.e. covering a dense distribution of eigenvalues over a wide range), multiple telescopic levels of PI can be constructed \cite{Melis2019,Melis2016}. While the method can be generalized for arbitrary number of telescopic levels, we focus on one additional intermediate level and describe a Telescopic Projective Forward Euler (TPFE) method with two projective levels for conciseness.

The innermost level performs $K_0+1$ innermost time steps with innermost time step size $\delta t_0$
\begin{equation}
  \Var_M^{n,k_1,k_0+1} = \Var_M^{n,k_1,k_0} + \delta t_0 D_t\left(\Var_M^{n,k_1,k_0}\right), \quad k_0=0,1,\ldots, K_0.
\end{equation}
and extrapolates to the intermediate level
\begin{equation}
  \Var_M^{n,k_1+1,0} = \Var_M^{n,k_1,K_0+1} + \left(\delta t_1 - (K_0+1) \delta t_0 \right) \frac{\Var_M^{n,k_1,K_0+1} - \Var_M^{n,k_1,K_0}}{\delta t_0}, \quad k_1=0,1,\ldots, K_1.
\end{equation}

The intermediate level finally extrapolates to the next time step
\begin{equation}
  \Var_M^{n+1} = \Var_M^{n,K_1+1,0} = \Var_M^{n,K_1,0} + \left(\Delta t - (K_1+1) \delta t_1 \right) \frac{\Var_M^{n,K_1+1,0} - \Var_M^{n,K_1,0}}{\delta t_1}.
\end{equation}

The respective extrapolation sizes $N_0,N_1$ are defined via
\begin{equation}
      N_0 \delta t_0 = \delta t_1 - (K_0+1) \delta t_0, \quad
      N_1 \delta t_1 = \Delta t - (K_1+1) \delta t_1.
    \label{e:extrapolation_factor}
\end{equation}
This TPFE method uses as parameters the innermost time step size $\delta t_0$, the intermediate time step size $\delta t_1$ and the respective number of inner and intermediate steps $K_0+1, K_1+1$ in addition to the outer time step size $\Delta t$.

In comparison to the PFE method in figure \ref{fig:stabilitydomain}, the stability domain of the TPFE method includes one additional domain that can be placed depending on the parameters to achieve a stable integration of models with more than one fast eigenvalue cluster or an extended spectrum of eigenvalues along the negative real axis, see \cite{Melis2016} for more details.

\subsection{Computational Speedup of PI}
\label{sec:speedup}
PI is used to speed up simulations of moment models close to hydrodynamic equilibrium, where the stiffness of the model equation would normally require an extremely small time step size. Assuming that the extrapolation step can be neglected in comparison to the inner integrators, the speedup $S$ with respect to a standard FE method is given by the ratio of the number of total time steps over a unit time interval and can be computed according to \cite{Melis2019} for the different methods.

For a PFE or PRK method the speedup is given by
\begin{equation}
    S_{PFE} = \frac{\Delta t}{(K+1) \cdot \delta t}.
    \label{e:speedupPFE}
\end{equation}

And for a TPFE method assuming constant $K_l=K$
\begin{equation}
    S_{TPFE} = \prod_{l=0}^{L}\frac{N_l+K_l+1}{K_l+1} = \frac{\delta t_1}{\delta t_0} \cdot \frac{\delta t_2}{\delta t_1} \cdot \ldots \cdot \frac{\Delta t}{\delta t_{L-1}} \frac{1}{\left(K+1\right)^L}= \frac{\Delta t}{\left(K+1\right)^L \cdot \delta t_0}.
    \label{e:speedupTPFE}
\end{equation}

Typically, the macroscopic time step is chosen according to a CFL-type constraint as $\Delta t \leq \frac{CFL \cdot \Delta x}{c_{max}}$, where $c_{max}$ is the largest eigenvalue of the moment model. The inner time step size $\delta t_0$ or $\delta t$ depends on the stiffness of the equation and is normally determined by $\delta t_0 = \frac{\nu_{max}}{\tau}$, for maximum collision frequency $\nu_{max}$.

Note that the speedup can be increased by a coarse macroscopic time step $\Delta t$ which is possible with a coarser spatial discretization, i.e. larger $\Delta x$, due to the CFL condition for the slow modes. This makes high-order spatial discretizations necessary to keep the desired spatial accuracy. In this work we thus use up to third-order spatial discretizations.
\section{Spectral Analysis}
\label{sec:LSA}
To overcome the stiffness of the kinetic equation \eqref{e:BTE} caused by the fast relaxing higher-order moments, we first need to characterize the spectral properties of the semi-discrete system \eqref{e:semidiscrete}, compare \cite{Yong1999}, so that we can match the spectrum of the system with the stability domain of the method, e.g., equation \eqref{e:stability_conditionPFE_lambda} and figure \ref{fig:stabilitydomain}. A detailed stability analysis for Discrete Velocity Models (DVM) was carried out in \cite{Melis2016} and similarly in \cite{Lafitte2016,Melis2019}. We will show that a linearized version of the hyperbolic moment model has the same spectrum as the DVM model studied in \cite{Melis2016} and then use the results to obtain the parameters for the PI schemes for a large variety of setups. 

\subsection{Linearized Hermite Spectral Method}
In the moment model, the non-linearity originates from the shifted expansion of the distribution function, see \eqref{e:vars_expansion} and \eqref{e:vars_transformed_velocity}, around a \emph{local} equilibrium Maxwellian. The expansion \eqref{e:vars_expansion} thus depends on the local density $\rho$, velocity $\Vel$ and temperature $\theta$. A linearized model can be derived by using a global Maxwellian for the expansion with vanishing velocity shift \eqref{e:vars_transformed_velocity}. This method is very close to the non-linear moment model because it uses the same Hermite basis and test functions, but it leads to a much simpler, linear model \cite{Fan2019}.

The expansion in 1D then reads
\begin{equation} \label{e:expansionHSM}
    f(t,x,c)=\sum_{\alpha = 0}^M f_{\alpha}(t,x) \mathcal{H}_{\alpha}(c),
\end{equation}
where the weighted Hermite basis functions $\mathcal{H}_{\alpha}$ are defined as
\begin{equation}\label{e:grad-basisfunctionHSM}
    \mathcal{H}_{\alpha}(c) = \frac{1}{\sqrt{2 \pi}} \exp \left( -\frac{c^2}{2} \right) He_{\alpha}(c) \cdot \frac{1}{\sqrt{2^{\alpha} {\alpha}!}}
\end{equation}
and $He_{\alpha}$ is the standard Hermite polynomial of degree ${\alpha}$. The last factor is chosen for normalization of the basis functions.

Similar to the non-linear model, the following constraints hold for the linear model
\begin{equation} \label{e:constraints_Grad_HSM}
    f_{0} = \rho, ~~ f_{1} = \rho u, ~~ f_{2} = \frac{1}{\sqrt{2}} \left( \rho \theta + \rho u^2 - \rho \right).
\end{equation}

The linear moment model is then also derived by projection of equation \eqref{e:BTE} onto Hermite polynomials and can be written as
\begin{equation}\label{e:grad-system_HSM}
     \frac{\partial \vect{f}}{\partial t} + \Vect{A} \frac{\partial \vect{f}}{\partial x} = -\frac{\nu}{\tau}\vect{S}\left( \vect{f} \right),
\end{equation}
using a $\vect{f}= \left(f_0,\ldots,f_M\right)^T$ and constant system matrix $\Vect{A} \in \mathbb{R}^{(M+1) \times (M+1)}$ given by
\begin{equation}
\label{e:grad_A_HSM}
\Vect{A} = \left(
  \begin{array}{ccccc}
     & 1 &   &   &   \\
     1 &  & \sqrt{2}  &   &   \\
     & \sqrt{2} &   &  \ddots &   \\
     &  & \ddots  &   & \sqrt{M}  \\
     &  &   & \sqrt{M}  &
  \end{array}
\right).
\end{equation}

The right-hand side vector $\vect{S}\left( \Var_M \right) \in \mathbb{R}^{M+1}$ for the 1D BGK model \eqref{e:BGK} is given by
\begin{equation}
  \vect{S}_{\alpha} = \int_{\mathbb{R}} \left( f(t,x,c) - f_M(t,x,c) \right) \psi_{\alpha}(c) \,dc, \quad \textrm{ for } \psi_{\alpha}(c) = He_{\alpha}(c) \cdot \frac{1}{\sqrt{2^{\alpha} {\alpha}!}},
\end{equation}
using the ansatz from \eqref{e:expansionHSM} and its form can be computed analytically beforehand. We omit the details of the derivation here for conciseness. We note that the right-hand side terms still have the same relaxation behavior as the non-linear model.

The propagation speeds of the hyperbolic transport part of the system \eqref{e:grad-system_HSM} are the roots of the Hermite polynomial
\begin{equation}\label{e:HSM_EV}
    \lambda_{i} = c_{i}, \textrm{ for } \textrm{He}_{M+1}(c_{i}) = 0,~ i=1, \ldots, M+1,
\end{equation}
which is a linearized version of the full non-linear model \eqref{e:HME_EV}. On the other hand, the following theorem states that the HSM model can also be seen as a linear transformation of a DVM method with non-uniformly placed discrete velocities according to the roots of the Hermite polynomial.

\begin{theorem}\label{th:HSM_CVM}
  The Hermite Spectral Method (HSM) \eqref{e:grad-system_HSM} for the BGK collision operator \eqref{e:BGK} using $M+1$ equations has the same spectrum as the Discrete Velocity Model (DVM) used in \cite{Melis2016}, Theorem 3.1, when the discrete velocities are the roots of the Hermite polynomial of degree $M+1$.
\end{theorem}
\begin{proof}
  We first relate the coefficients of the HSM and DVM models, before using a similarity transformation of the system \eqref{e:grad-system_HSM}.

  The HSM uses ansatz \eqref{e:expansionHSM} and then employs the projection of a distribution function $g$ to the $j$-th Hermite polynomial $He_{j}$
  \begin{equation}\label{e:HSM_projection}
    P_j(g) = \int_{\mathbb{R}} f(c) He_{j}(c) \, dc
  \end{equation}

  Applying the projection to the expanded distribution function \eqref{e:expansionHSM} and using an (exact) Gauss-Hermite quadrature rule with quadrature weights $\omega_k$ and quadrature points $c_k$ as roots of $He_{M+1}$, for $k=0,\ldots, M$, leads to
  \begin{align}
    P_{j}(f) &= \int_{\mathbb{R}} f(t,x,c) He_{j}(c) \, dc \\
     &= \sum_{k=0}^M \omega_k f(t,x,c_k) He_{j}(c_k)\\
     &= \sum_{k=0}^M \omega_k \sum_{\alpha = 0}^M f_{\alpha}(t,x) \mathcal{H}_{\alpha}(c_k) He_{j}(c_k)\\
     &= f_{j}(t,x),
  \end{align}
  due to orthonormality and exactness of the Gauss-Hermite quadrature rule.

  Similarly, a DVM method using the same $M+1$ discrete velocities $v_{\alpha}=c_{\alpha}$ can be written using an expansion in Dirac functions $\delta\left(c-v_{\alpha}\right)$ as
  \begin{equation} \label{e:expansionDVM}
    f(t,x,c)=\sum_{\alpha = 0}^M \widetilde{f_{\alpha}}(t,x) \delta\left(c-v_{\alpha}\right),
  \end{equation}
  with the corresponding projection $\widetilde{P_j}$ that leads to
  \begin{align}
    \widetilde{P_j}(f) &= \int_{\mathbb{R}} f(t,x,c) \delta\left(c-v_j\right) \, dc \\
     &= \int_{\mathbb{R}} \sum_{\alpha = 0}^M \widetilde{f_{\alpha}}(t,x) \delta\left(c-v_{\alpha}\right) \delta\left(c-v_j\right) \, dc \\
     &= \widetilde{f_{j}}(t,x).
  \end{align}

  A relation between both sets of coefficients $f_{\alpha}$ and $\widetilde{f_{\alpha}}$ can be derived by matching the moments of the respective expansions for the distribution function $f(t,x,c)$. The moments are computed with Hermite test functions
  \begin{align}
    \int_{\mathbb{R}} f(t,x,c) He_{j}(c) \, dc &= \int_{\mathbb{R}} f(t,x,c) He_{j}(c) \, dc \\
    \int_{\mathbb{R}} \sum_{\alpha = 0}^M f_{\alpha}(t,x) \mathcal{H}_{\alpha}(c) He_{j}(c) \, dc &= \int_{\mathbb{R}} \sum_{\alpha = 0}^M \widetilde{f_{\alpha}}(t,x) \delta\left(c-v_{\alpha}\right) He_{j}(c) \, dc \\
     f_{j}(t,x) &= \sum_{\alpha = 0}^M \widetilde{f_{\alpha}}(t,x) He_{j}(v_{\alpha})  \\
     \Rightarrow \vect{f} &= \Vect{B} \cdot \widetilde{\vect{f}},
  \end{align}
  where the entries of the transformation matrix $B \in \mathbb{R}^{(M+1)\times(M+1))}$ are point evaluations of the Hermite functions at the $v_k$, i.e. $\vect{B}_{i,j} = He_{i}\left(v_j\right)$.

  Using the relation of the different sets of coefficients, we continue from the HSM system \eqref{e:grad-system_HSM}, with BGK right-hand side from \eqref{e:BGK}

  \begin{equation}\label{e:grad-system_HSM_proj}
     \frac{\partial \vect{f}}{\partial t} + \Vect{A} \frac{\partial \vect{f}}{\partial x} = -\frac{\nu}{\tau}\left( \vect{P}\left(f_M\right) - \vect{f} \right).
  \end{equation}
  The entries of the projection $\vect{P}\left( f_M \right)$ of the Maxwellian on the right-hand side can be computed using the same quadrature rule as before to find the relation to the DVM model's projection $\widetilde{\vect{P}}\left( f_M \right)$
  \begin{align}
    P_j\left( f_M \right) &= \int_{\mathbb{R}} f_M(c) He_{j}(c) \, dc \\
     &= \sum_{k=0}^M \omega_k f_M(c_k) He_{j}(c_k) \\
    \Rightarrow \vect{P}\left( f_M \right) &= \Vect{B} \widetilde{\vect{P}}\left( f_M \right).
  \end{align}

  By multiplication of \eqref{e:grad-system_HSM_proj} with the inverse of the constant transformation matrix $\Vect{B}^{-1}$ we get
  \begin{align}
    \frac{\partial \Vect{B}^{-1} \vect{f}}{\partial t} + \Vect{B}^{-1} \Vect{A} \frac{\partial  \vect{f}}{\partial x} &= -\frac{\nu}{\tau}\left( \Vect{B}^{-1} \vect{P}\left(f_M\right) - \Vect{B}^{-1} \vect{f} \right) \\
    \frac{\partial \Vect{B}^{-1} \Vect{B} \widetilde{\vect{f}}}{\partial t} + \Vect{B}^{-1} \Vect{A} \Vect{B} \frac{\partial  \widetilde{\vect{f}}}{\partial x} &= -\frac{\nu}{\tau}\left( \Vect{B}^{-1} \Vect{B} \widetilde{\vect{P}}\left( f_M \right) - \Vect{B}^{-1} \Vect{B} \widetilde{\vect{f}} \right) \\
    \Rightarrow \frac{\partial \widetilde{\vect{f}}}{\partial t} + \widetilde{\Vect{A}} \frac{\partial \widetilde{\vect{f}}}{\partial x} &= -\frac{\nu}{\tau}\left(  \widetilde{\vect{P}}\left( f_M \right) - \widetilde{\vect{f}} \right),
  \end{align}
  resulting in the DVM system with system matrix $\widetilde{\Vect{A}} = \Vect{B}^{-1} \Vect{A} \Vect{B}$, which is a similarity transformation of the HSM system matrix $\Vect{A}$.

  The DVM system and the HSM system for the BGK model are thus similar with the same spectral properties.
\end{proof}

This allows to use the results for the DVM schemes presented in \cite{Melis2019} in this linearized version of the full non-linear moment model. As the investigation of the full Boltzmann collision operator is more involved, we restrict ourselves to the BGK operator here. However, a similar result is expected to hold for the Boltzmann operator and the linearized moment model.

In \cite{Melis2019} the spectrum of the BGK operator is characterized by a slow eigenvalue cluster and one or more fast eigenvalue clusters. After spatial discretization, the spectrum of the respective semi-discrete system can be evaluated. The main parameters to characterize the spectrum are the collision frequency $\nu$ and the relaxation time $\tau$ occurring on the right-hand side of the collision term \eqref{e:grad-system_HSM}. Following \cite{Melis2019} and theorem \ref{th:HSM_CVM}, the spectrum $\mathcal{S}$ of system \eqref{e:grad-system_HSM} after spatial discretization can be formally characterized by
\begin{equation}
\label{e:HSM_spectrum}
  \mathcal{S} \subset \mathcal{D}\left( -\frac{\nu}{\tau} , r_{D_x} \right) \cup \{ \lambda^{(1)} \} ,
\end{equation}
where the radius $r_{D_x}$ depends on the velocity space and spatial discretization and the dominant eigenvalues $\lambda^{(1)}$ correspond to the slow eigenvalues of the macroscopic variables. We clearly see a scale separation and want to exemplify the results for the HSM model numerically in the following.

For the numerical computation of the spectrum, we consider equation \eqref{e:grad-system_HSM} and write its semi-discrete version after discretization in space similar to the non-linear model \eqref{e:semidiscrete} as
\begin{equation}\label{e:semi_discrete_LSA}
     \frac{\partial \vect{f}}{\partial t} = \vect{D_{x}}(\vect{f}, \tau, \nu),
\end{equation}
where the term $\vect{D_{x}}(\vect{f}, \tau, \nu)$ denotes the spatial discretization of the transport and collision terms on a spatial grid. In this linear stability analysis, we then linearize the right hand side in $\vect{f}$ around equilibrium $\vect{f_0}$
\begin{equation}\label{e:LSA_linearization}
     \vect{D_{x}}(\vect{f}, \tau, \nu) \approx \vect{D_{x}}(\vect{f_0}, \tau, \nu) + \frac{\partial \vect{D_{x}}}{\partial \vect{f}}\left( \vect{f} - \vect{f_0}\right),
\end{equation}
where the Jacobian of the spatial discretization with respect to the solution is the matrix $\frac{\partial \vect{D_{x}}}{\partial \vect{f}}$, that determines the linear stability of the semi-discrete system \eqref{e:semi_discrete_LSA}. Matching the spectrum of this matrix with the stability domain of the time integration scheme later is the crucial step to achieve a stable time stepping scheme.

\begin{remark}
    We assume that the linearization error for the HSM model is small because of two reasons. Firstly, the HSM system matrix \eqref{e:grad_A_HSM} is constant. Secondly, the BGK collision operator $\vect{S}$ models relaxation of higher order moments towards the equilibrium. Close to equilibrium, this leads to a diagonal matrix, see also \eqref{e:BGK_term}. In addition, the spatial discretization is a simple, linear combination of neighboring values on the grid, at least for the first order spatial schemes. The linearization in \eqref{e:LSA_linearization} is therefore reasonable. This is in agreement with the analysis in \cite{Melis2016} for DVM models.

    Transferring the results of the linear stability analysis to the non-linear model is justified, because we are only interested in solutions of the non-linear model close to equilibrium. In equilibrium, the distribution function degenerates to a Maxwellian and all non-equilibrium variables vanish. For larger deviations from non-equilibrium, the linear stability analysis might no longer be sufficient. The interested reader is referred to the equilibrium stability analysis for a class of non-linear moment models in \cite{zhao2017}.

    The parameters for the PI method are chosen based on the linear stability analysis. The numerical results in the next section of this paper indicate that this procedure works for the test cases presented. A more detailed study of the spectral properties of the non-linear system is beyond the scope of this paper and might be considered as future work.
\end{remark}

The computations of the Jacobian matrix $\frac{\partial \vect{D_{x}}}{\partial \vect{f}}$ and the corresponding eigenvalues are performed numerically using finite differences in the software \cite{Koellermeier2020b}. As the eigenvalues depend largely on $\nu$ and $\tau$, we will distinguish three different examples for $\nu$ and $\tau$ to investigate the spectrum and allow for a proper choice of the PI parameters thereafter.
All examples are performed for a spatial discretization of the domain $[-2,2]$ using a constant $\Delta x = 0.01$ using the HSM with $M=4$. This leads to a semi-discrete system with $400\cdot 5 = 2000$ variables. However, the results are qualitatively the same for different spatial discretizations and models.

\subsubsection{BGK with constant relaxation time}
\label{sec:LSA_const}
Figure \ref{fig:SpectrumNu1} shows the numerical eigenvalue spectrum of the HSM method with $M=4$ for constant collision frequency $\nu=1$ and varying relaxation times. Starting with $\tau=1$ in figure \ref{fig:HSM5FENu1Kn1}, which corresponds to the kinetic regime, we do only see one cluster of eigenvalues as the microscopic and macroscopic scale are of the same order. The same holds for $\tau = 10^{-2}$ (not shown), where the scales are still about the same. Further in the transitional regime for $\tau = 10^{-3}$ there is a separation for the first time. The microscopic modes relax much faster than the macroscopic modes, which are independent of the relaxation time. In the hydrodynamic regime, see figure \ref{fig:HSM5FENu1Kn0p0001}, there is a larger separation and a significant speedup can be expected from using a PI method in this cases. We note that there is only one fast cluster in all cases, which will be different in the next test case.

\begin{figure}[htb!]
    \centering
    \begin{subfigures}
    \subfloat[Eigenvalue spectrum for $\nu=1$, $\tau = 1$. \label{fig:HSM5FENu1Kn1}
    ]{\includegraphics[width=0.85\linewidth]{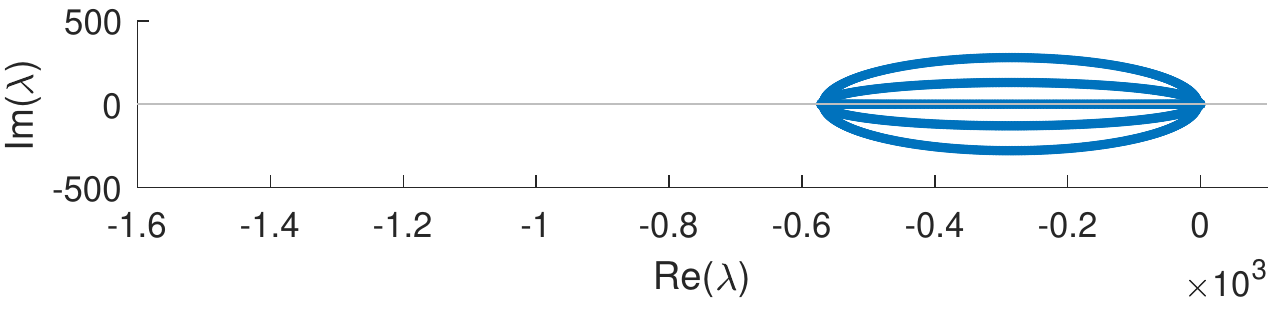}}\\
    \subfloat[Eigenvalue spectrum for $\nu=1$, $\tau = 10^{-3}$. \label{fig:HSM5FENu1Kn0p001}
    ]{\includegraphics[width=0.85\linewidth]{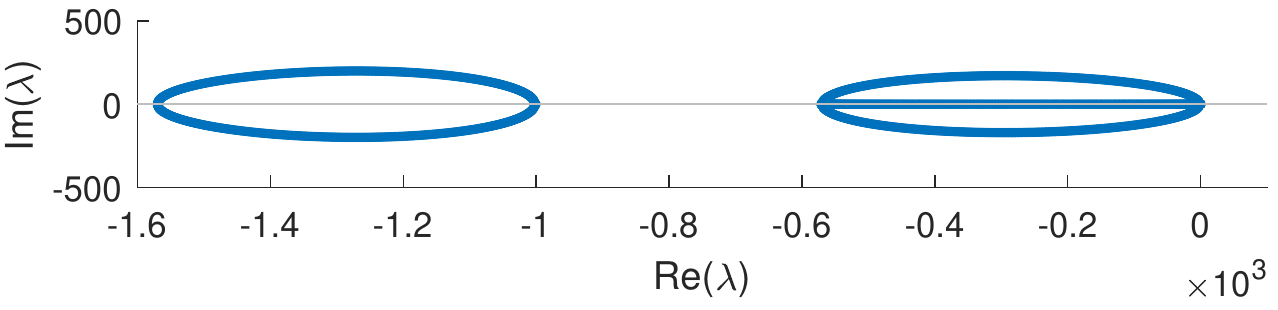}}\\
    \subfloat[Eigenvalue spectrum for $\nu=1$, $\tau = 10^{-4}$. \label{fig:HSM5FENu1Kn0p0001}
    ]{\includegraphics[width=0.85\linewidth]{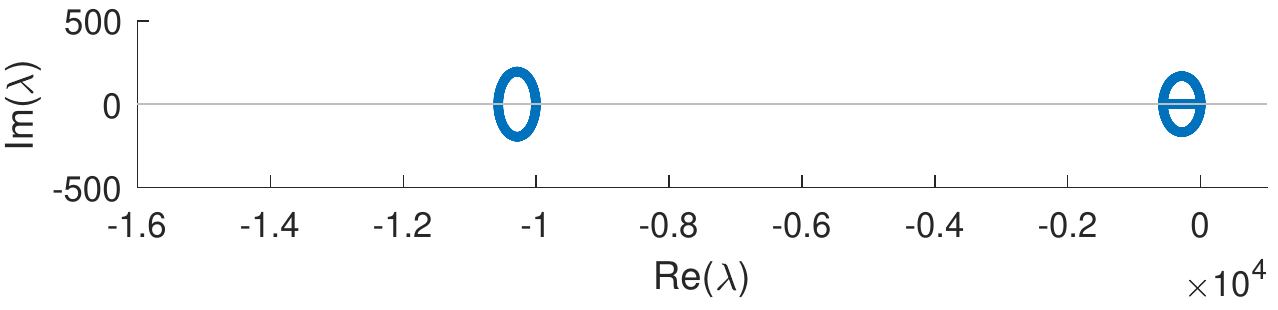}}\\
    \end{subfigures}
    \caption{Increasing spectral gap in eigenvalue spectra of HSM4 for constant collision frequency $\nu=1$ and varying $\tau$ is ideally suited for the application of projective integration.}
    \label{fig:SpectrumNu1}
\end{figure}

\subsubsection{BGK with piecewise constant relaxation time}
\label{sec:LSA_piecewise}
For a piecewise constant collision frequency time $\nu \in \{0.1,1\}$, e.g., by changing the collision frequency the spatial domain, the results are shown in figure \ref{fig:SpectrumNuPiecewise10}. Figure \ref{fig:HSM5FENuPiecewise10Kn0p001} shows no additional eigenvalue cluster for $\tau = 10^{-3}$ because the cluster corresponding to the low-collisional part coincides with the macroscopic slow cluster. However, there is an intermediate cluster for both $\tau = 10^{-4}$ in figure \ref{fig:HSM5FENuPiecewise10Kn0p0001} which will be much more pronounced for larger $\tau = 10^{-6}$ (not shown). 
In those cases, a standard PFE method would not be stable and an additional projective integrator needs to be used. This can efficiently be realized by the TPFE method with an intermediate integrator tailored to the intermediate cluster.

\begin{figure}[htb!]
    \centering
    \begin{subfigures}
    \subfloat[Eigenvalue spectrum for $\nu \in \{0.1,1\}$, $\tau = 10^{-3}$. \label{fig:HSM5FENuPiecewise10Kn0p001}
    ]{\includegraphics[width=0.85\linewidth]{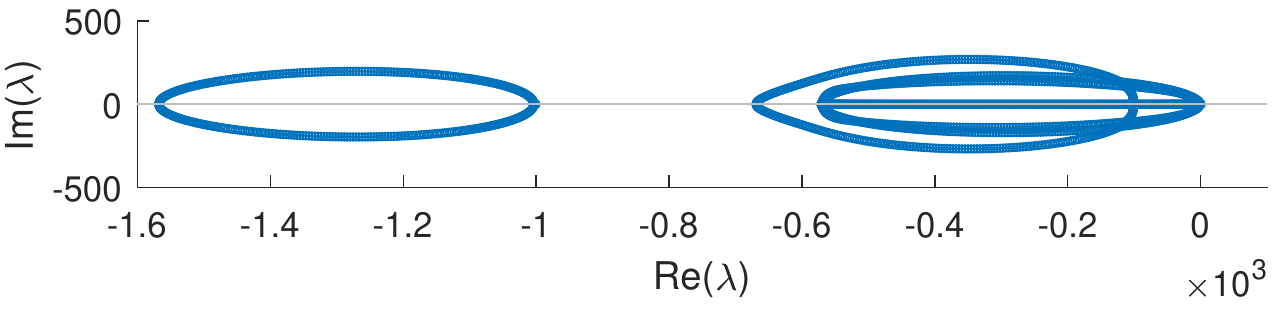}}\\
    \subfloat[Eigenvalue spectrum for $\nu \in \{0.1,1\}$, $\tau = 10^{-4}$. \label{fig:HSM5FENuPiecewise10Kn0p0001}
    ]{\includegraphics[width=0.85\linewidth]{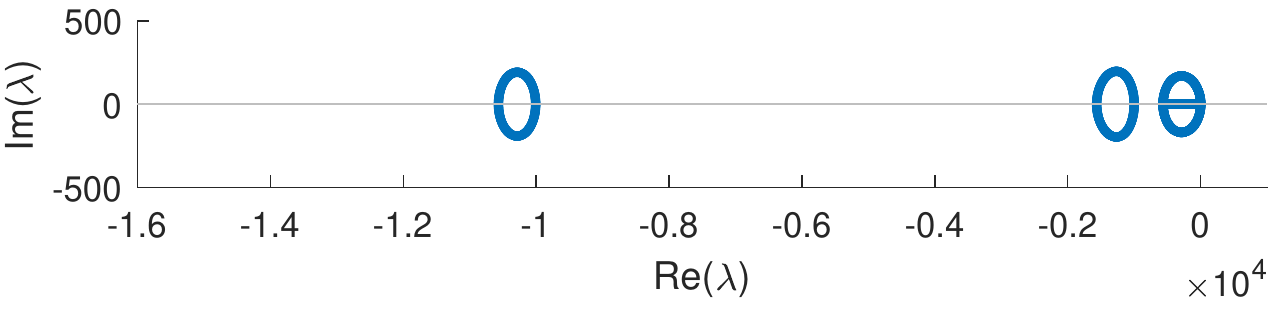}}\\
    \end{subfigures}
    \caption{Additional intermediate cluster in eigenvalue spectra of HSM4 for piecewise constant collision frequency $\nu \in \{0.1,1\}$ and varying $\tau$ is ideally suited for the application of a two-level telescopic projective integration.}
    \label{fig:SpectrumNuPiecewise10}
\end{figure}

\subsubsection{BGK with space-dependent relaxation time}
\label{sec:LSA_cont}
For a space-dependent collision frequency $\nu = \rho(x)$ with $\rho \in [1,7]$ according to the shock tube test case later, we get the extended spectrum shown in figure \ref{fig:SpectrumNuRho71}. There is still one slow macroscopic cluster. However, the microscopic cluster is spread out along the negative axis in all cases. For $\tau = 10^{-3}$ in figure \ref{fig:HSM5FENuRho71Kn0p001}, the values are not exactly on the negative real axis as the relaxation is not enough to damp the imaginary parts. In all cases, neither a standard PFE method nor a TPFE method with the parameter settings derived from the previous test case are stable. Instead a TPFE with a connected stability region needs to be used. We will outline the construction of this method according to \cite{Melis2019} in the following section.

\begin{figure}[htb!]
    \centering
    \begin{subfigures}
    \subfloat[Eigenvalue spectrum for space-dependent $\nu$ and $\tau = 10^{-3}$. \label{fig:HSM5FENuRho71Kn0p001}
    ]{\includegraphics[width=0.85\linewidth]{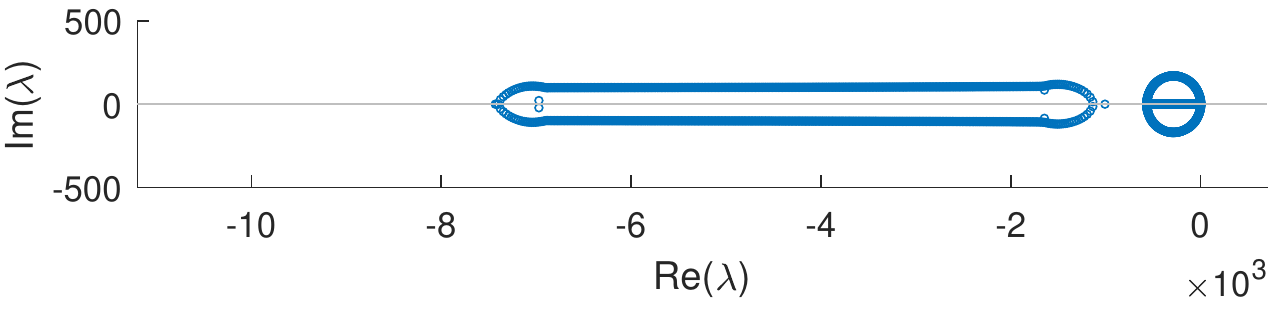}}\\
    \subfloat[Eigenvalue spectrum for space-dependent $\nu$ and $\tau = 10^{-4}$. \label{fig:HSM5FENuRho71Kn0p0001}
    ]{\includegraphics[width=0.85\linewidth]{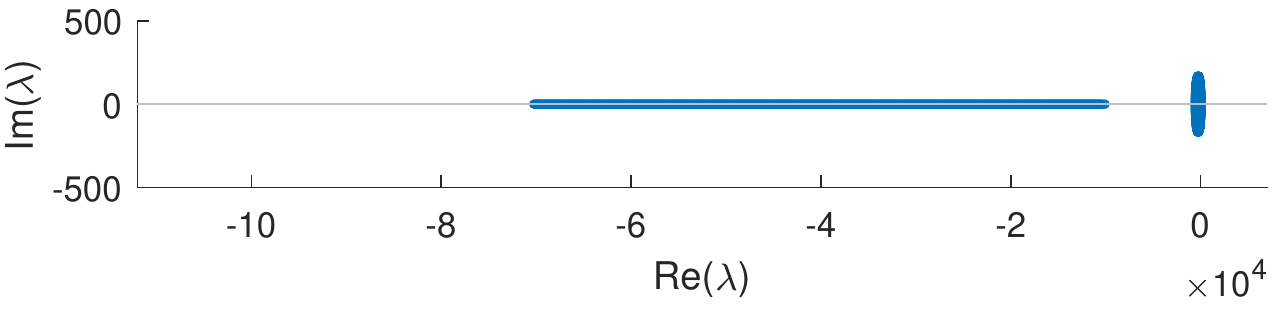}}\\
    \end{subfigures}
    \caption{Extended fast eigenvalue spectra of HSM4 for space-dependent collision frequency $\nu = \rho(x) \in [1,7]$ and varying $\tau$ requires a connected stability region of a TPI method.}
    \label{fig:SpectrumNuRho71}
\end{figure}

\subsection{Non-linear Hyperbolic Moment model and Boltzmann collision operator}
\label{sec:LSA_BTE}
The results of the linear stability analysis obviously depend on the linearization. However, the right hand side relaxation terms are only linear in the non-equilibrium variables for the BGK equation. For the Boltzmann collision operator, we refer to the results in \cite{Melis2019}, leading to a extended spectrum. In the case where the collision frequency $\nu$ depends on the density $\rho$, we assume a maximum principle for the density $\rho$. This means that the range of $\rho$ can be known a-priori and the properties of the spectrum can be determined.

We can thus apply the observations of the spectral properties from the previous test cases also for the non-linear moment model and the Boltzmann collision operator and use this to determine the parameters of the PI methods in the following section.

\subsection{PI parameter choice}
\label{sec:parameters}
Now we choose all involved parameters so that a stable time stepping method for the semi-discrete system \eqref{e:semidiscrete} is obtained. This requires matching of the spectrum of the semi-discrete system with the stability domain of the method, see section \ref{sec:Num}. It is essential for this to know the position of the fast eigenvalue clusters, which has been investigated and clarified by the stability analysis. For constant and piecewise constant collision frequency $\nu$, the fast clusters are always located at the position corresponding to $-\frac{\nu}{\tau}$ on the real axis. For the space-dependent $\nu=\rho(x)$, the spectrum of fast modes extends over the domain $\left[-\frac{\rho_{max}}{\tau},-\frac{\rho_{min}}{\tau}\right]$, assuming that only intermediate values in the interval $\left[\rho_{min},\rho_{max}\right]$ are obtained by the pressure.

Based on the stability analysis in the previous sections, the appropriate numerical schemes and their parameters can be chosen while taking into account the respective stability properties. We follow the suggestions in \cite{Melis2019,Melis2016} where a similar studies were performed for DVM methods. In general, we can distinguish four cases:

\begin{itemize}
  \item[1.] No clear scale separation, compare figures \ref{fig:HSM5FENu1Kn1}, \ref{fig:HSM5FENu1Kn0p001}, \ref{fig:HSM5FENuPiecewise10Kn0p001}. All modes can be covered by the standard macroscopic time step using a CFL-type time step size. The use of PI is not necessary. A standard FE scheme will be used.
  \item[2.] Scale separation with one cluster of fast modes that requires PI, compare figures \ref{fig:HSM5FENu1Kn0p0001}. In that case, we choose a standard PI method such as PFE or PRK3 using $\delta t$ according to the position of the fast eigenvalue cluster. In the aforementioned cases, this leads to $\delta t = \tau$. According to \cite{Melis2017}, a small number of inner iterations is sufficient. In our numerical tests, we use $K=1$ for cases with only one fast cluster. This leads to a fast, but stable integration scheme.
  \item[3.] Scale separation with more than one cluster of fast modes that requires TPI, compare \ref{fig:HSM5FENuPiecewise10Kn0p0001}. The appropriate parameters of the TPI method can be derived in the following way:
      \begin{itemize}
        \item[(a)] The time step sizes $\delta t_i$ are determined depending on the positions of the respective fast clusters according to $\delta t_i = \frac{\nu_i}{\tau}$.
        \item[(b)] According to \cite{Melis2019}, choosing $K=1$ is sufficient for the application cases.
        \item[(c)] The respective extrapolation factor can be computed directly according to equation \eqref{e:extrapolation_factor}.
      \end{itemize}
  \item[4.] Continuous spectrum extending outside of the stability region of the macroscopic time step requiring an A-stable integrator, see figures \ref{fig:HSM5FENuRho71Kn0p001}, \ref{fig:HSM5FENuRho71Kn0p0001}:
      \begin{itemize}
        \item[4.1] A small range of the spectrum, i.e. $\frac{\Delta t}{\delta t_0}< 27$, for CFL-type time step size $\Delta t$ and microscopic time step size $\delta t_0$ according to the fastest eigenvalues in the system obtained from the spectrum. In this case, one level of PI is enough and only the number of inner time steps needs to be increased. The maximum extrapolation factor $N+K+1=\frac{\Delta t}{\delta t_0}$ for which a connected stability region is obtained is given in table \ref{tab:M}.

            As an example, a range of $\frac{\Delta t}{\delta t_0} \approx 10$ leads to $K=3$ inner iterations to obtain a connected spectrum.
        \item[4.2] The spectrum extends over a wider range $\frac{\Delta t}{\delta t_0}$, so that at least one additional layer of PI is necessary leading to a an actual TPI scheme. The parameters are then chosen according to
            \begin{itemize}
                \item[(a)] The innermost time step size $\delta t_0$ is determined based on the positions of the fastest mode.
                \item[(b)] A number of inner time steps $K_i$ on each level is fixed. Here we always use a constant $K_i=K$. The outer time step size $\Delta t$ is fixed according to a CFL-type condition.
                \item[(c)] The maximum factor $N$ used for the extrapolation is found in table \ref{tab:M}. This maximum choice of $N$ ensures that the stability regions are connected.
                \item[(d)] The minimum number of levels is computed using
                    \begin{equation}
                        L = \frac{\log(\Delta t) + \log\left(\frac{1}{\delta t_0}\right)}{\log(N+K+1)}
                    \end{equation}
                \item[(e)] The intermediate time step sizes can be computed according to \eqref{e:extrapolation_factor}, e.g., as
                    \begin{equation}
                        \delta t_1 = (N+K+1)\cdot \delta t_0
                    \end{equation}
                    The outer factors $N$ might need to be slightly adapted to be consistent with the outermost time step size $\Delta t$, for more details see \cite{Melis2019}.
            \end{itemize}
      \end{itemize}
\end{itemize}

\begin{table}[H]
    \centering
    \caption{Maximum extrapolation factors $N+K+1$ for connected stability region depending on $K$ according to \cite{Melis2019}.}
    \label{tab:M}
    \begin{tabular}{c|c|c|c|c|c|c|c}
      $K$ & 1 & 2 & 3 & 4 & 5 & 6 & 7 \\ \hline
      $N+K+1$ & 4 & 6 & 10.66 & 13.32 & 18.21 & 21.24 & 26.21 \\ \hline
      $N$ & 2 & 3 & 6.66 & 8.32 & 12.21 & 14.24 & 18.21
    \end{tabular}
\end{table}
\section{Numerical experiments}
\label{sec:NumEx}
All test cases are computed either with the non-linear QBME \eqref{e:vars_system1D} and \eqref{e:vars_system} based on a moment method expansion of the distribution function, or with the linearized HSM \eqref{e:grad-system_HSM} derived with the help of a linearization around a global Maxwellian. For implementation details used in all examples of this section we refer to the implementation \cite{Koellermeier2020b}.

\subsection{Shock tube problem}
For the first application test we consider a 1D shock tube, a standard benchmark problem in rarefied gases, see \cite{Au2001,Cai2013,Koellermeier2017}. The shock tube features a strong propagating shock wave. Close to the shock the solution will be in non-equilibrium if the relaxation time $\tau$ is large. However, for small relaxation time, the solution will quickly relax to the equilibrium Maxwellian and in the limit it can be derived easily by the well known Euler equations \eqref{e:cons_mass}-\eqref{e:cons_energy}. In this regime the kinetic equation becomes stiff and it is difficult to solve. We are thus interested in a speedup of moment models for simulations close to equilibrium.

At $t=0$, the gas is in exact equilibrium, with the density, velocity, and temperature given by
\begin{equation}
    \left(\rho, \vel, \theta \right) = \left\{
  \begin{array}{cl}
    \left(7,0,1 \right) & \textrm{if } x < 0 \\
    \left(1,0,1 \right) & \textrm{if } x > 0 \\
  \end{array}
\right.,
    \label{e:shock_tube_problem}
\end{equation}
modeling a jump in density at the discontinuity at $x=0$.

The computational domain is $[-2,2]$. The simulations run until $t_{\textrm{\tiny{END}}} = 0.3$ and the constant macroscopic time step is $\Delta t = 3.85 \cdot 10^{-4}$ corresponding to a CFL number of $0.5$ on a spatial grid discretized with $1000$ cells. Note that we use less cells as in \cite{Koellermeier2017} due to the higher-order spatial discretization. Both moment models HSM and QBME use $M=9$.

\subsubsection{Constant collision frequency \texorpdfstring{$\nu=1$}{}}
We first consider the case of a constant collision frequency $\nu=1$, in which the spectrum has a clear spectral gap for small relaxation time, as analyzed in \ref{sec:LSA_const}. Before starting the simulations, the correct methods and parameters need to be chosen. We distinguish five different cases by value of the Knudsen number:
\begin{itemize}
  \item[1.] $\tau = 10^{-1}$: kinetic regime. According to figure \ref{fig:SpectrumNu1}, there is no scale separation and the fast modes can accurately be captured by the CFL-type macroscopic time step. We thus employ the standard FE scheme using $\Delta t = 3.85 \cdot 10^{-4}$ corresponding to a CFL number of $0.5$.
  \item[2.] $\tau = 10^{-2}$: transitional regime. According to figure \ref{fig:SpectrumNu1}, there is still no scale separation and we can use the same settings as for case 1.
  \item[3.] $\tau = 10^{-3}$: transitional regime. According to figure \ref{fig:SpectrumNu1}, the scales have separated and we see a very small spectral gap. However, the fast scale is of the order the time step size $\Delta t = 3.85 \cdot 10^{-4}$. This means that we are not yet in a stiff situation, where the fast modes require a special treatment. We use the FE scheme with the above settings.
  \item[4.] $\tau = 10^{-4}$: transitional regime. Figure \ref{fig:SpectrumNu1} shows a clear scale separation and the a standard FE method would be unstable. We thus employ the PFE method. For the stable integration of the fast cluster, we use inner step size $\delta t = \tau = 10^{-4}$ and $K=1$.
  \item[5.] $\tau \leq 10^{-5}$: hydrodynamic regime. The clear scale separation in figure \ref{fig:SpectrumNu1} grows and we take this into account by choosing the proper PFE with inner step size $\delta t = \tau$ and $K=1$.
\end{itemize}
We note that the choices for the methods clearly follow the stability analysis of the previous section and do not require any iteration or try and error. All simulations run stable and the results can be compared in figure \ref{fig:shocktubeNu1}, which shows the results for the different Knudsen numbers $\tau$ depending on the model (linear HSM or non-linear QBME) and depending on the order of the spatial and temporal discretization.

When comparing the left and right column of figure \ref{fig:shocktubeNu1}, we can clearly see that a higher-order discretization leads to a sharper profile, accurately resolving the limiting Euler solution. The first order scheme yields more diffusion, damps the shocks, and does not predict the shock front accurately. Due to the use of the higher-order scheme, larger spatial discretizations are possible and allow for large $\Delta t$ according to the CFL number. However, this makes the use of PI even more necessary, as the fast scales require a small time step.

Both the HSM and the QBME model approach the hydrodynamic limit for decreasing $\tau$.
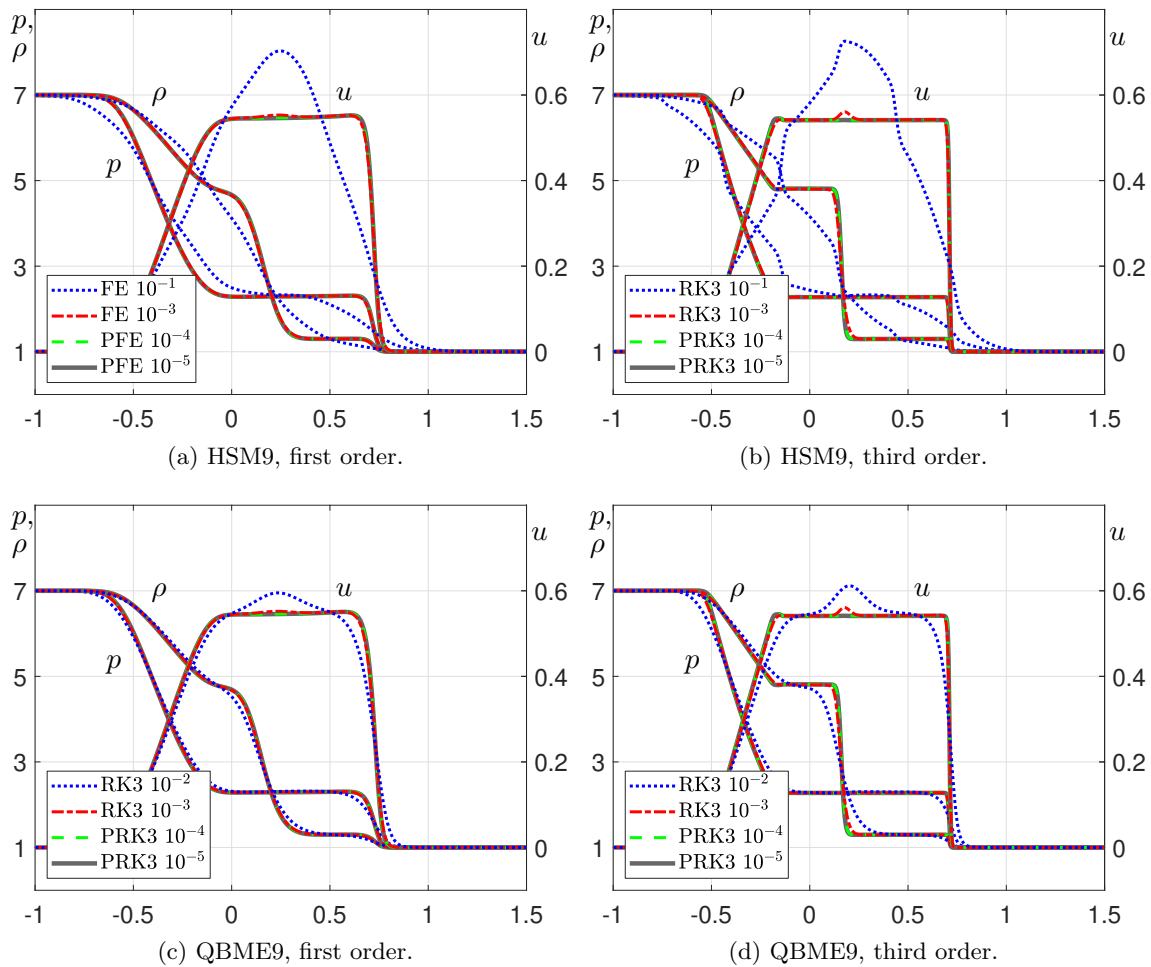
\begin{figure}[htb!]
    \centering
    \begin{subfigures}
    \subfloat[HSM9, first order. \label{fig:shocktubeHSM10FORCE1}
    ]{\begin{overpic}[width=0.49\textwidth]{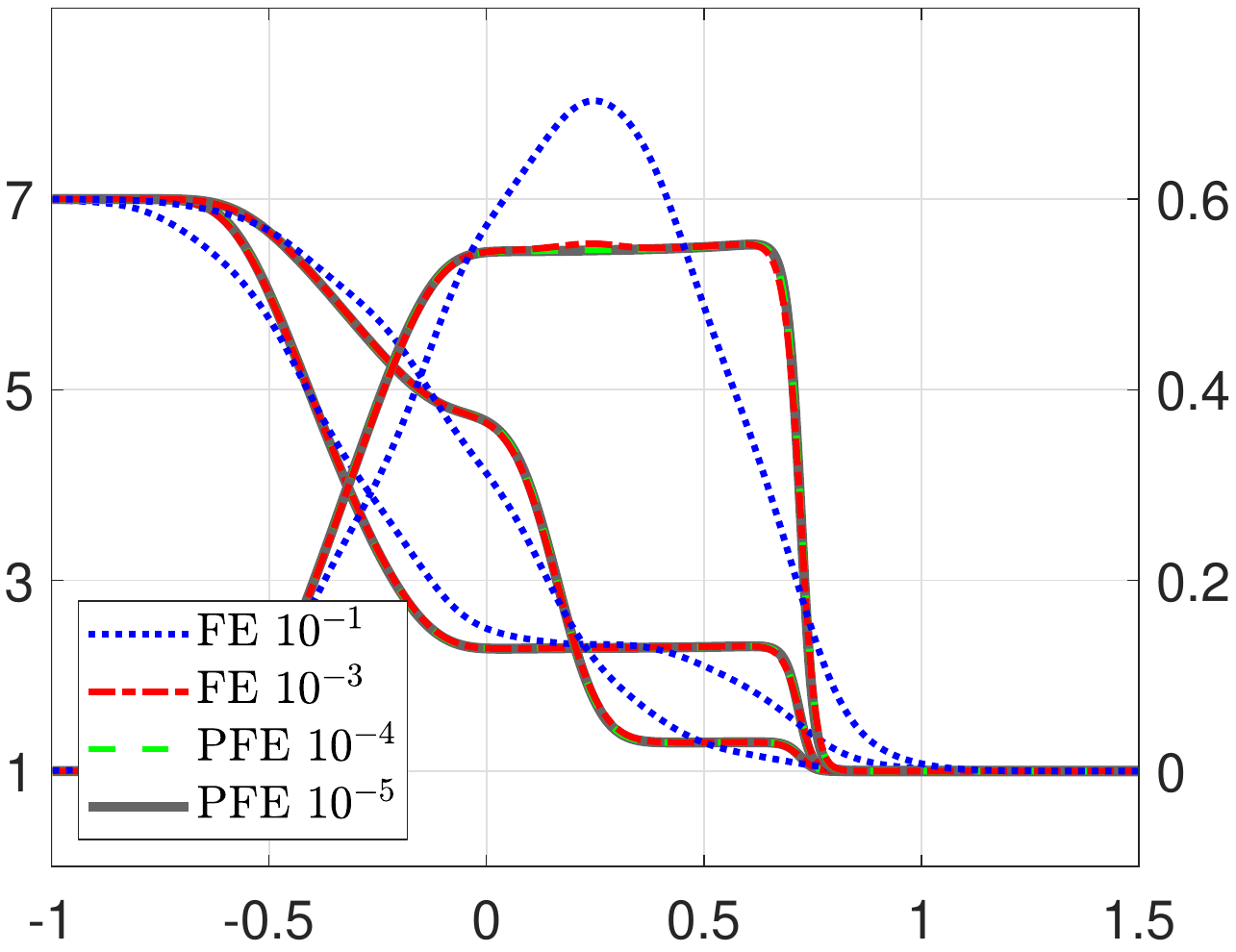}
        \put(0,73){$p,$}
        \put(0,68){$\rho$}
        \put(93,70){$\vel$}
        \put(17,47){$p$}
        \put(25,60){$\rho$}
        \put(58,60){$\vel$}
    \end{overpic}}~
    \subfloat[HSM9, third order. \label{fig:shocktubeHSM10FORCE3}
    ]{\begin{overpic}[width=0.49\textwidth]{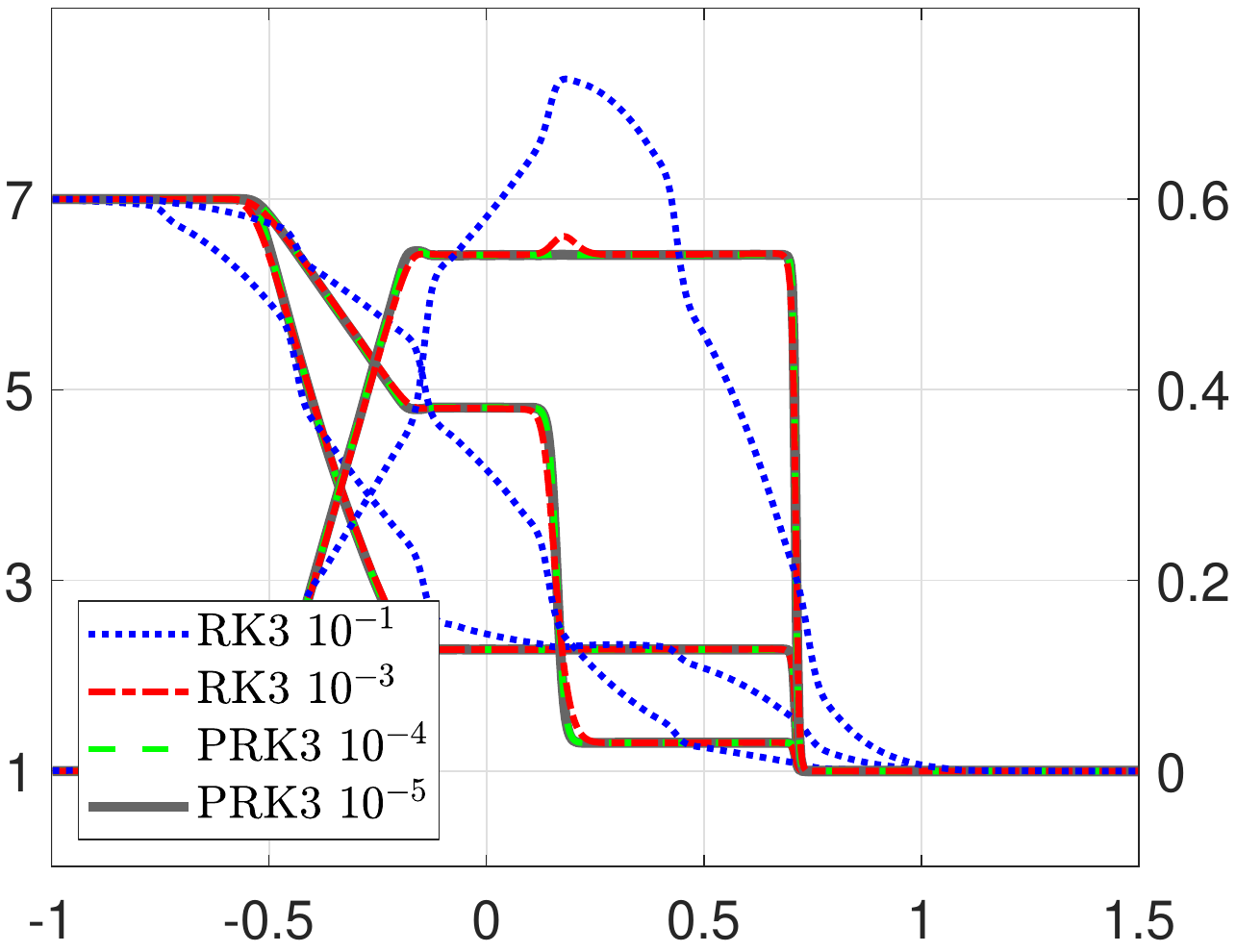}
        \put(0,73){$p,$}
        \put(0,68){$\rho$}
        \put(93,70){$\vel$}
        \put(17,47){$p$}
        \put(25,60){$\rho$}
        \put(58,60){$\vel$}
    \end{overpic}}\\
    \subfloat[QBME9, first order. \label{fig:shocktubeQBME10FORCE1}
    ]{\begin{overpic}[width=0.49\textwidth]{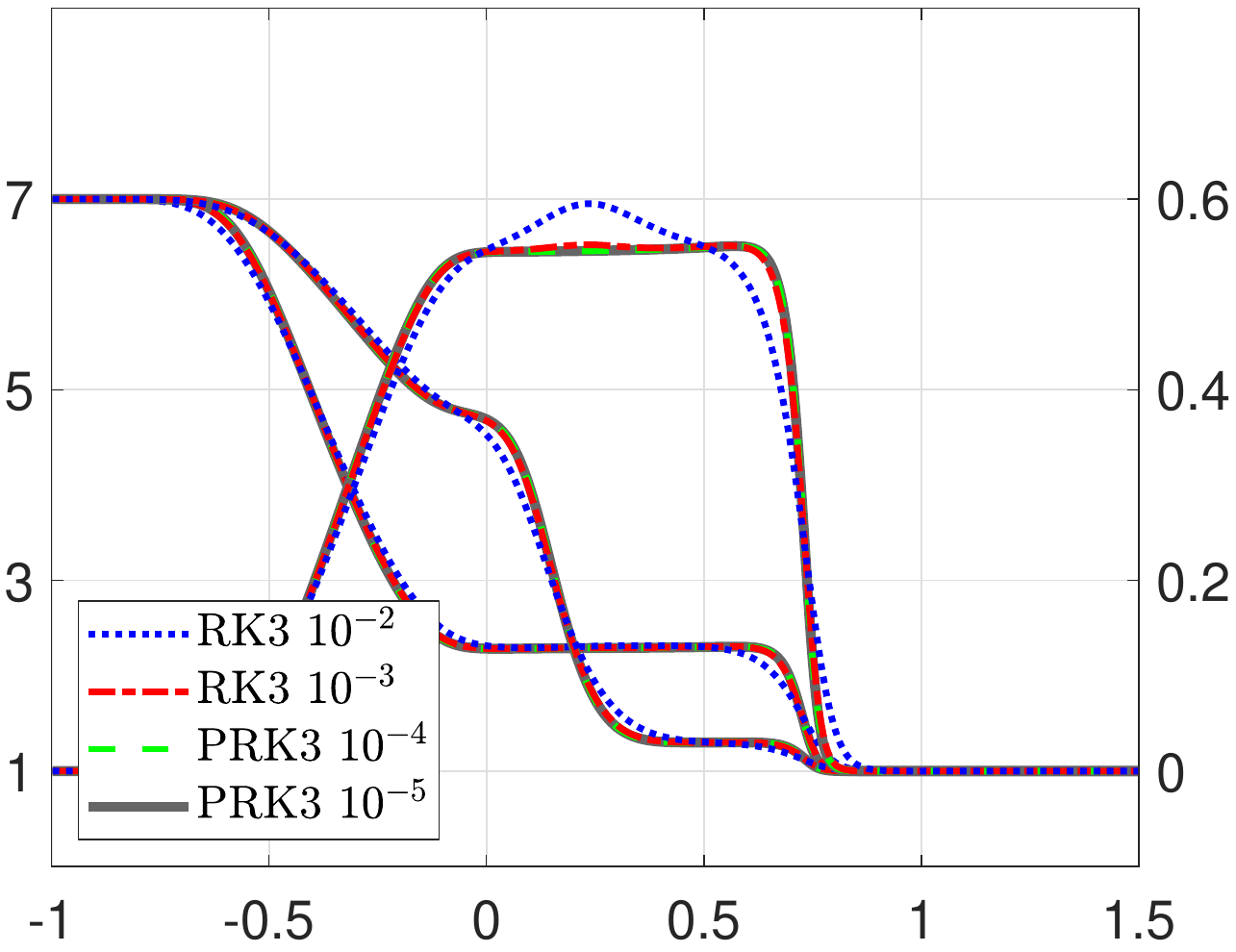}
        \put(0,73){$p,$}
        \put(0,68){$\rho$}
        \put(93,70){$\vel$}
        \put(17,47){$p$}
        \put(25,60){$\rho$}
        \put(58,60){$\vel$}
    \end{overpic}}~
    \subfloat[QBME9, third order. \label{fig:shocktubeQBME10FORCE3}
    ]{\begin{overpic}[width=0.49\textwidth]{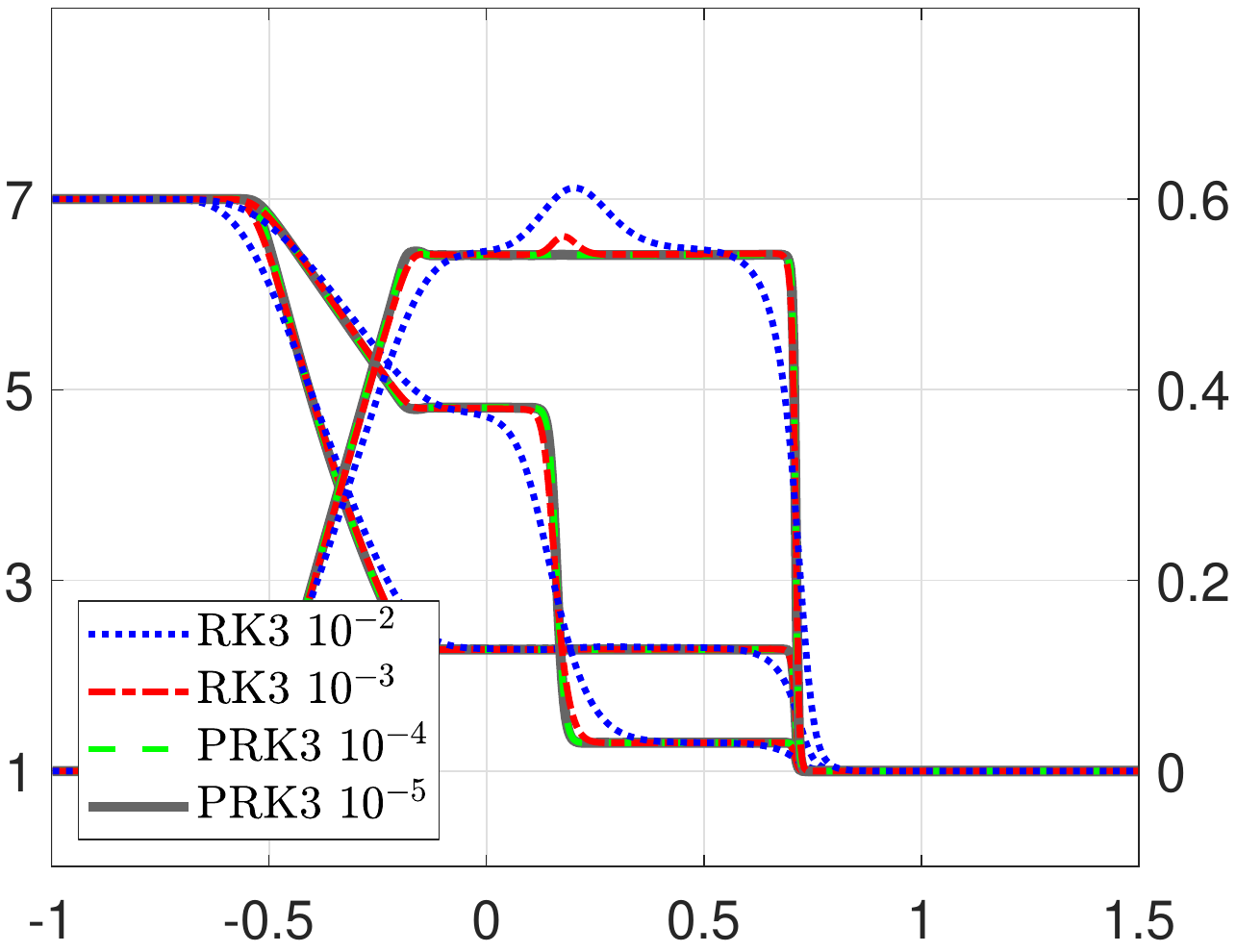}
        \put(0,73){$p,$}
        \put(0,68){$\rho$}
        \put(93,70){$\vel$}
        \put(17,47){$p$}
        \put(25,60){$\rho$}
        \put(58,60){$\vel$}
    \end{overpic}}
    \end{subfigures}
    \caption{Shock tube for constant collision frequency $\nu=1$ and varying $\tau$.}
    \label{fig:shocktubeNu1}
\end{figure}

Table \ref{tab:parameter_PFE} exemplarily indicates the stable parameter settings used for the QBME model, $\nu=1$ and relaxation time $\tau=10^{-5}$. As expected, the chosen inner time step size has to be chosen as $\delta t = \tau$ in this simple test case to prevent instabilities (shown in red). This indicates that the stable PI parameters for the non-linear QBME models agree very well with the prediction of the linear stability analysis in this test case.
\begin{table}[H]
    \centering
    \caption{Stability of different parameter settings for PFE. QBME model, $\nu=1, \tau=10^{-5}$. Base parameters $K=1$, $\delta=1 \cdot 10^{-5}$. Parameters predicted by linear stability analysis indicated by gray column. Instable simulation indicated by \color{red}{red} numbers.}
    \label{tab:parameter_PFE}
    \begin{tabular}{c||c|c| >{\columncolor{light-gray}}c|c|c}\hline
      $\delta t / 10^{-5}$    & $\color{red}{1.5} $ & $\color{red}{1.1}$ & $1$ & $\color{red}{0.9}$ & $\color{red}{0.5}$  \\ \hline
    \end{tabular}
\end{table}

\subsubsection{Space-dependent collision frequency \texorpdfstring{$\nu=\rho(x)$}{}}
When choosing a space-dependent collision frequency, the linear stability analysis in section \ref{sec:LSA_cont} revealed an extended eigenvalue spectrum that needs to be taken into account to obtain stability, see section \ref{sec:parameters}. Note that we use the same jump in density from $\rho_L=7$ to $\rho_R=1$ as was used to analyze the stability, so that the results from there are directly transferrable. In this test case we thus need to make the following adjustments to the numerical method (see also \cite{Melis2019}):
\begin{itemize}
  \item[1.] $\tau = 10^{-2}$: transitional regime. Performing the same stability analysis as in \ref{sec:LSA_cont} for this $\tau$, we observe that there is no scale separation yet and the "fast" modes lie well within the slower modes. We can thus use the FE scheme with macroscopic $\Delta t$ according to the CFL number.
  \item[3.] $\tau = 10^{-3}$: transitional regime. Figure \ref{fig:SpectrumNuRho71} shows a beginning separation. The derivation in section \ref{sec:parameters} shows that a single level PFE method is still stable and we use $\delta t=1.4\cdot 10^{-4}$ with $K=1$.
  \item[4.] $\tau = 10^{-4}$: hydrodynamic regime. Due to the stronger separation, also shown in figure \ref{fig:SpectrumNu1}, we need to choose more inner time steps to achieve an A-stable method that has a connected stability region. According to \ref{sec:parameters}, we choose $\delta t = \tau / 7 = 1.4 \cdot 10^{-5}$ and $K=6$.
  \item[5.] $\tau \leq 10^{-5}$: hydrodynamic regime. For this test case, the derivation in section \ref{sec:parameters} shows that it is necessary to employ a telescopic method. Following this derivation we use a TPFE method with $K=6$ on both inner levels, $\delta t_0=1.4\cdot 10^{-6}$, and $\delta t_1=3.0 \cdot 10^{-5}$.
\end{itemize}
Note again that all parameter choices can be directly obtained by means of the stability analysis \ref{sec:LSA}, section \ref{sec:parameters}, and the properties of the PI method \ref{sec:Num}.

The results in figure \ref{fig:shocktubeNuRho} show a clear convergence of
the QBME model towards the hydrodynamic equilibrium. The same holds for the HSM model (not shown). Similar as for the previous test case, we see that there is a significant gain in accuracy when using a higher-order spatial discretization as shown in the right column, where the third-order FORCE scheme is used, in comparison to the left column, where the first-order FORCE scheme was employed.

\begin{figure}[htb!]
    \centering
    \begin{subfigures}
    \subfloat[QBME9, first order. \label{fig:shocktuberhoQBME10FORCE1}
    ]{\begin{overpic}[width=0.49\textwidth]{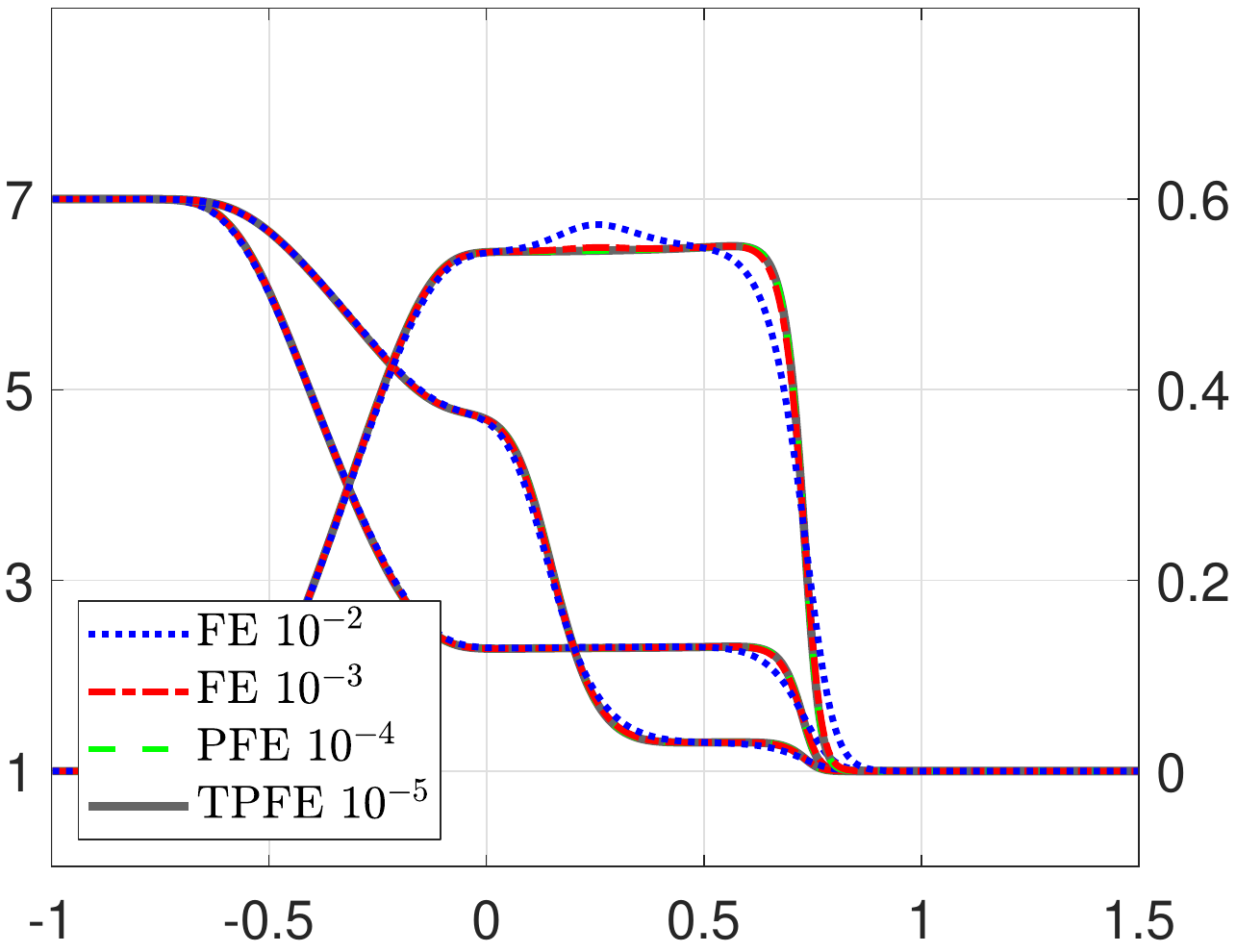}
        \put(0,73){$p,$}
        \put(0,68){$\rho$}
        \put(93,70){$\vel$}
        \put(17,47){$p$}
        \put(25,60){$\rho$}
        \put(58,60){$\vel$}
    \end{overpic}}~
    \subfloat[QBME9, third order. \label{fig:shocktuberhoQBME10FORCE3}
    ]{\begin{overpic}[width=0.49\textwidth]{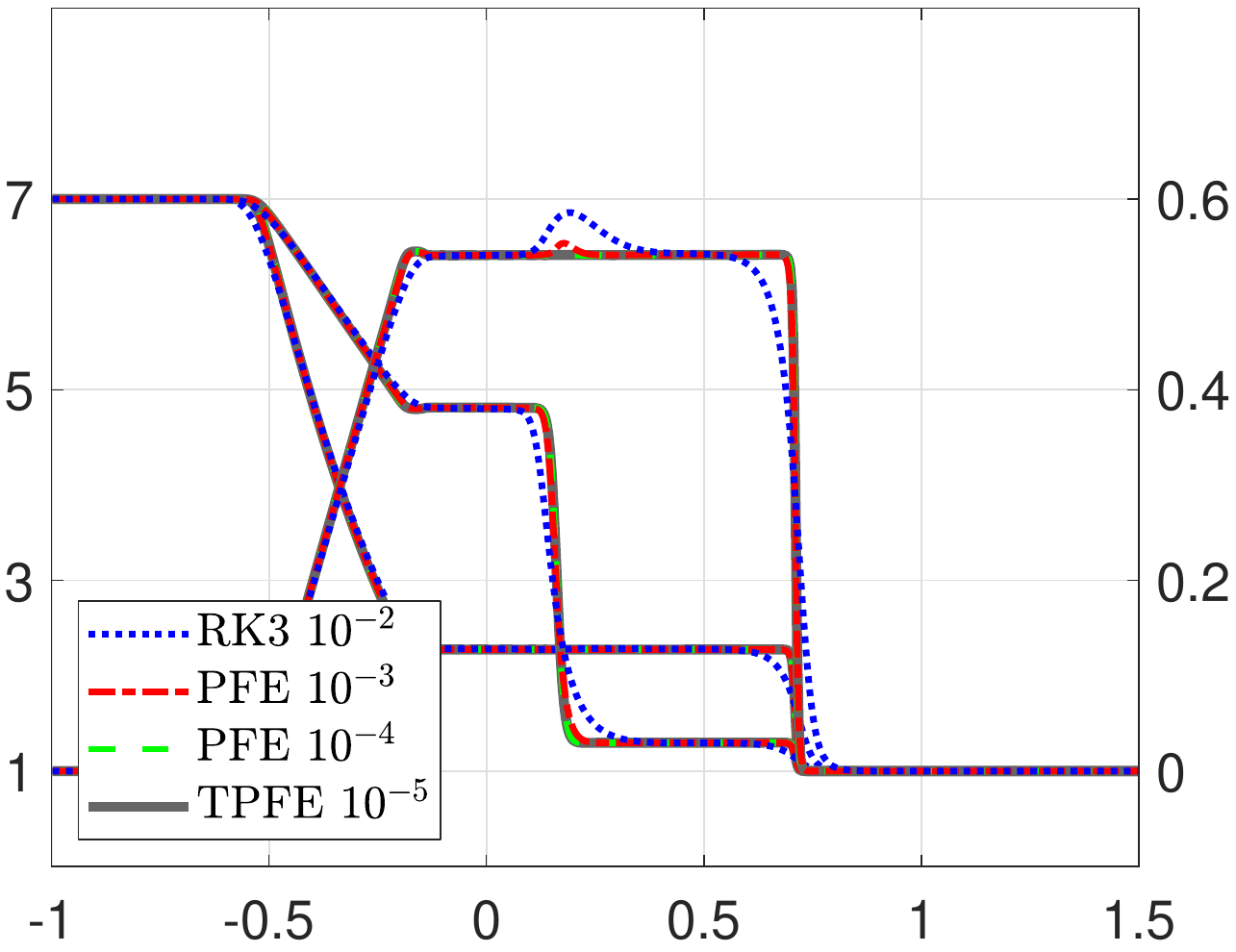}
        \put(0,73){$p,$}
        \put(0,68){$\rho$}
        \put(93,70){$\vel$}
        \put(17,47){$p$}
        \put(25,60){$\rho$}
        \put(58,60){$\vel$}
    \end{overpic}}
    \end{subfigures}
    \caption{Shock tube for space-dependent collision frequency $\nu=\rho(x)$ and varying $\tau$.}
    \label{fig:shocktubeNuRho}
\end{figure}
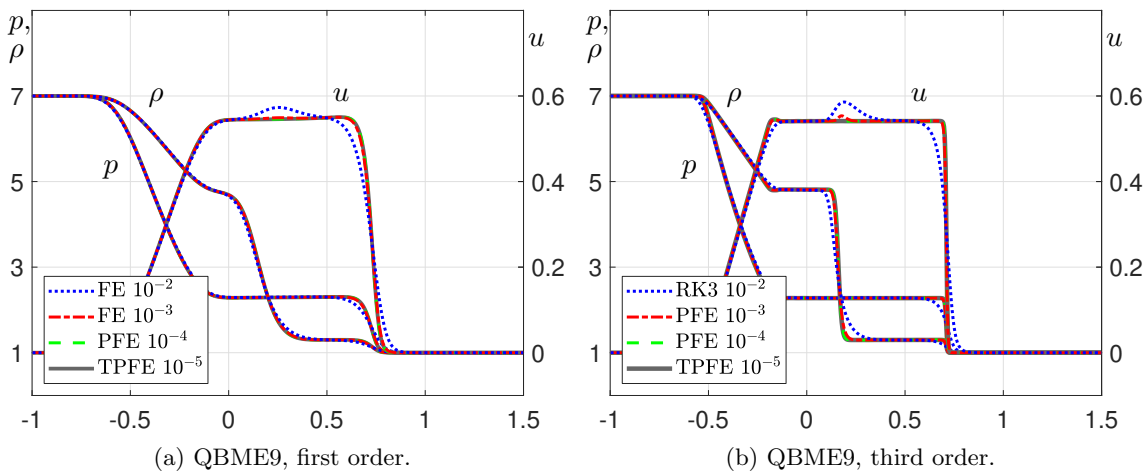


Table \ref{tab:parameter_TPFE} exemplarily indicates the stable parameter settings of the TPFE method used for the QBME model, $\nu=\rho(x)$ and relaxation time $\tau=10^{-5}$. The predicted number of inner time steps $K=6$ is indeed the minimum stable value. A further reduction of steps leads to instability problems (shown in red). For the inner time step size $\delta_0$, not only the predicted value $\delta_0=1.4 \cdot 10^{-6}$ is stable, but also values $2\cdot 10^{-6} \geq \delta_0 \geq 1.3 \cdot 10^{-6}$. For values outside this region it can be assumed that the stability region splits up into two domains leading to instability. Also for the intermediate time step size $\delta t_1$, larger and smaller values are possible. This might be attributed to the smaller extrapolation size on the intermediate level. Table \ref{tab:parameter_TPFE} indicates that the parameters chosen with the help of the linear stability analysis are not the only stable choices for the TPFE method in this non-linear QBME test case. However, the choice $K=6$ minimizes the computational cost in this case and the other parameters $\delta_0,\delta_1$ lie well within the set of stable parameters.
\begin{table}[H]
    \centering
    \caption{Stability of different parameter settings for TPFE. QBME model, $\nu=\rho(x), \tau=10^{-5}$. Each line changes only one parameter from the chosen parameters predicted by linear stability analysis $K=6$, $\delta_0=1.4 \cdot 10^{-6}$, $\delta_1=3 \cdot 10^{-5}$ indicated by gray column. Instable simulation indicated by \color{red}{red} numbers.}
    \label{tab:parameter_TPFE}
    \begin{tabular}{c||c|c| >{\columncolor{light-gray}}c|c|c}\hline
      $K$                       & $8$ & $7$ & $6$ & $\color{red}{5}$ & $\color{red}{4}$  \\ \hline
      $\delta t_0 / 10^{-6}$    & $\color{red}{2.5}$ & $2$ & $1.4$ & $1.3$ & $\color{red}{1.2}$  \\ \hline
      $\delta t_1 / 10^{-5}$    & $\color{red}{5}$ & $4$ & $3$ & $2$ & $1$  \\ \hline
    \end{tabular}
\end{table}

\subsubsection{Model comparison}
Based on the previous tests, we have a closer look at the model differences in the limit of smaller $\tau$ in figure \ref{fig:shocktubeNuRhoModel}. Each graph includes the different models for the same relaxation time $\tau$. Figure \ref{fig:shocktubeModelsKn0p0001FORCE3} shows that all models eventually converge to the same hydrodynamic limit and even for $\tau=10^{-3}$ there are no visible differences. For $\tau=10^{-2}$, the models with constant collision frequency $\nu=1$ differ slightly from the space-dependent collision frequency $\nu=\rho(x)$ due to the large density jump. In the kinetic regime for $\tau=10^{-1}$, the linearized HSM model begins to show deviations from the non-linear QBME model, too. We conclude that besides the kinetic region, the model differences are most prominent in the transitional regime, whereas the hydrodynamic regime will be simulated accurately by either model.

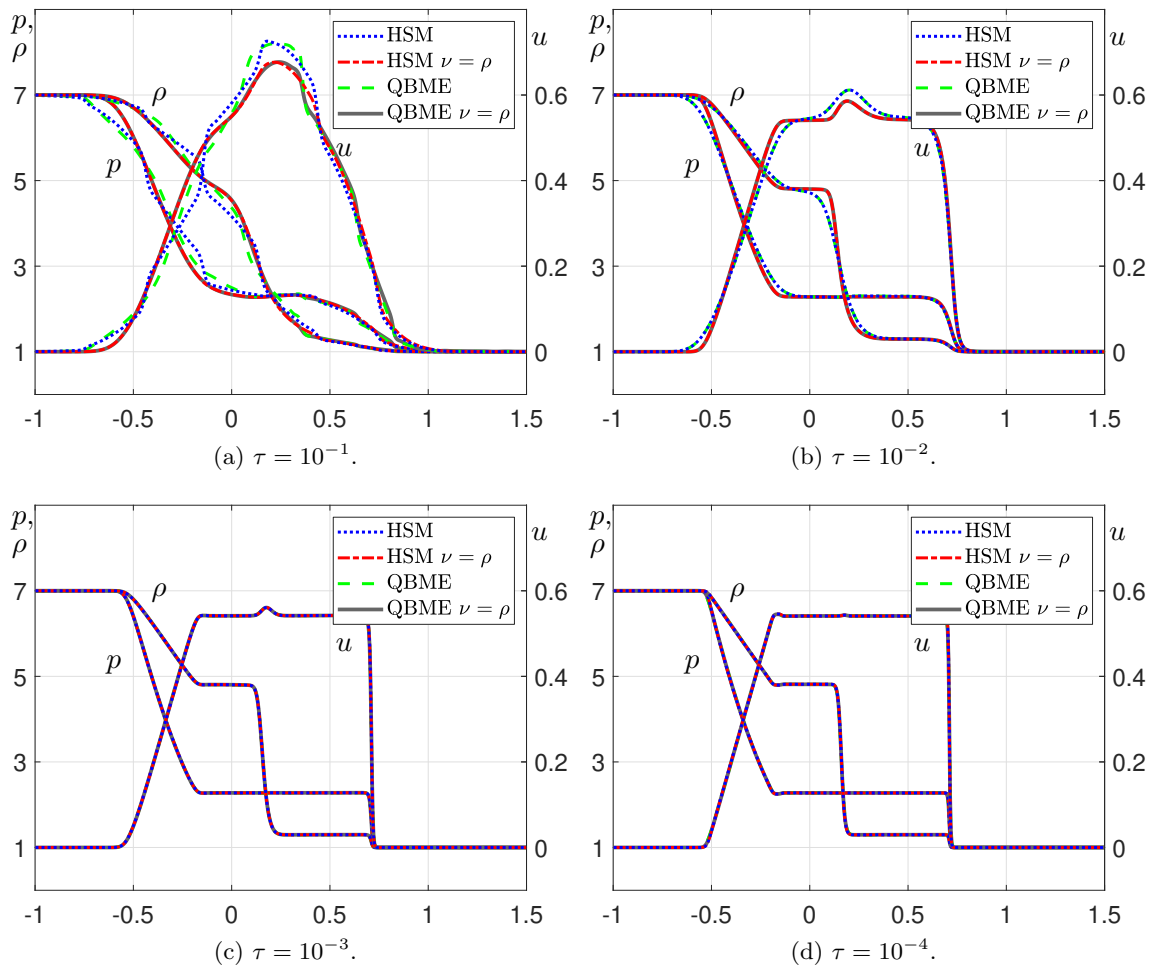
\begin{figure}[htb!]
    \centering
    \begin{subfigures}
    \subfloat[$\tau=10^{-1}$. \label{fig:shocktubeModelsKn0p1FORCE3}
    ]{\begin{overpic}[width=0.49\textwidth]{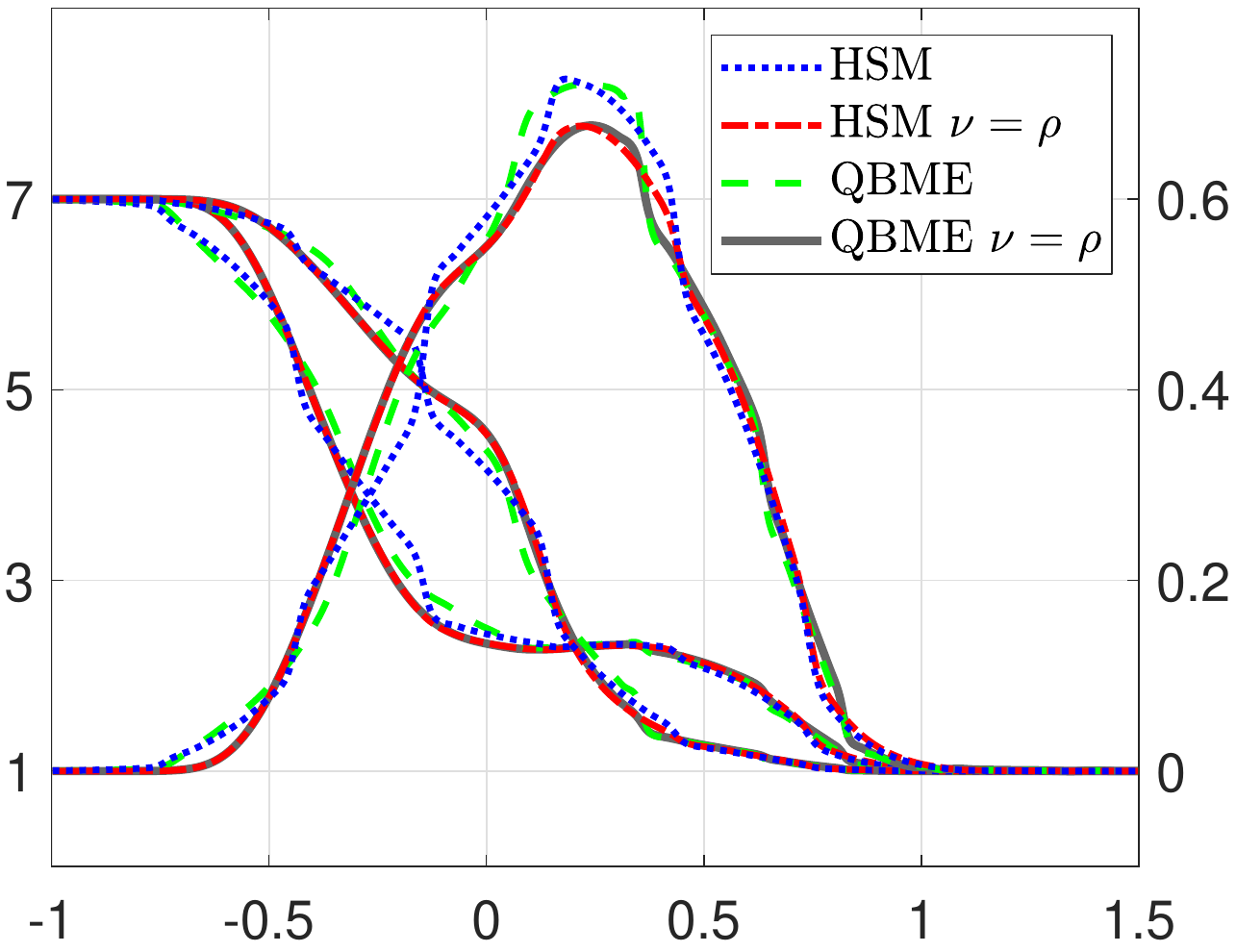}
        \put(0,73){$p,$}
        \put(0,68){$\rho$}
        \put(93,70){$\vel$}
        \put(17,47){$p$}
        \put(25,60){$\rho$}
        \put(58,50){$\vel$}
    \end{overpic}}~
    \subfloat[$\tau=10^{-2}$. \label{fig:shocktubeModelsKn0p01FORCE3}
    ]{\begin{overpic}[width=0.49\textwidth]{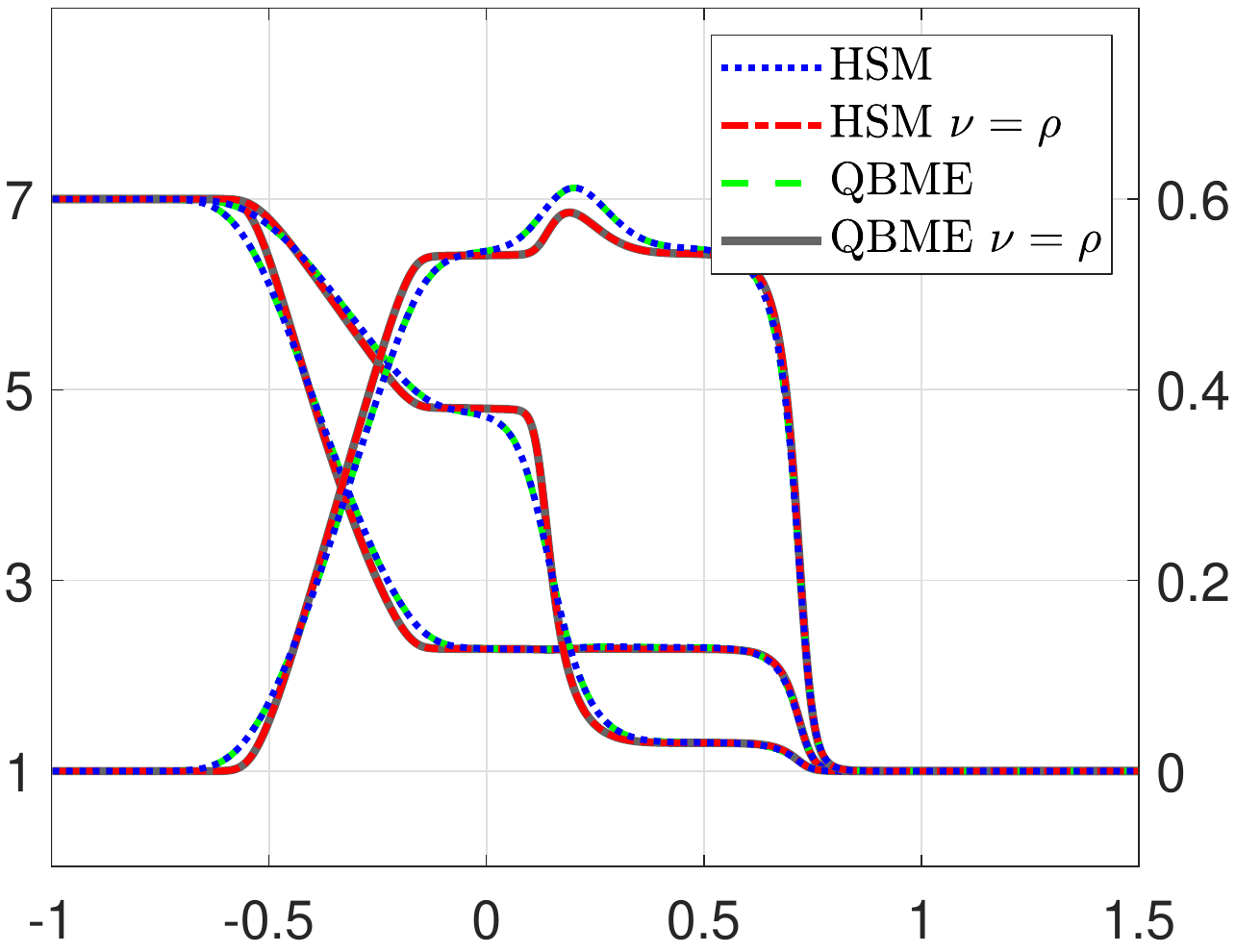}
        \put(0,73){$p,$}
        \put(0,68){$\rho$}
        \put(93,70){$\vel$}
        \put(17,47){$p$}
        \put(25,60){$\rho$}
        \put(58,50){$\vel$}
    \end{overpic}}\\
    \subfloat[$\tau=10^{-3}$. \label{fig:shocktubeModelsKn0p001FORCE3}
    ]{\begin{overpic}[width=0.49\textwidth]{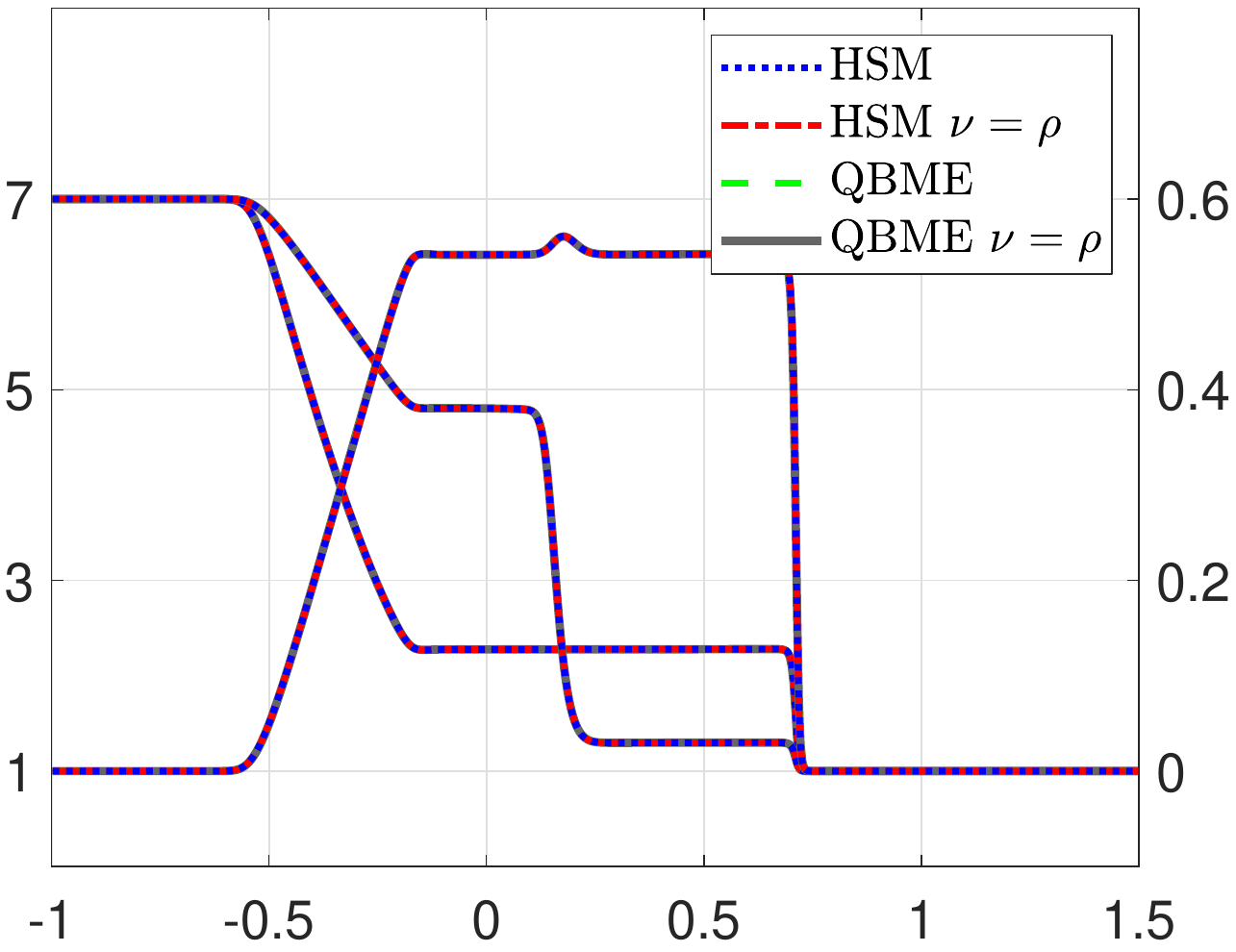}
        \put(0,73){$p,$}
        \put(0,68){$\rho$}
        \put(93,70){$\vel$}
        \put(17,47){$p$}
        \put(25,60){$\rho$}
        \put(58,50){$\vel$}
    \end{overpic}}~
    \subfloat[$\tau=10^{-4}$. \label{fig:shocktubeModelsKn0p0001FORCE3}
    ]{\begin{overpic}[width=0.49\textwidth]{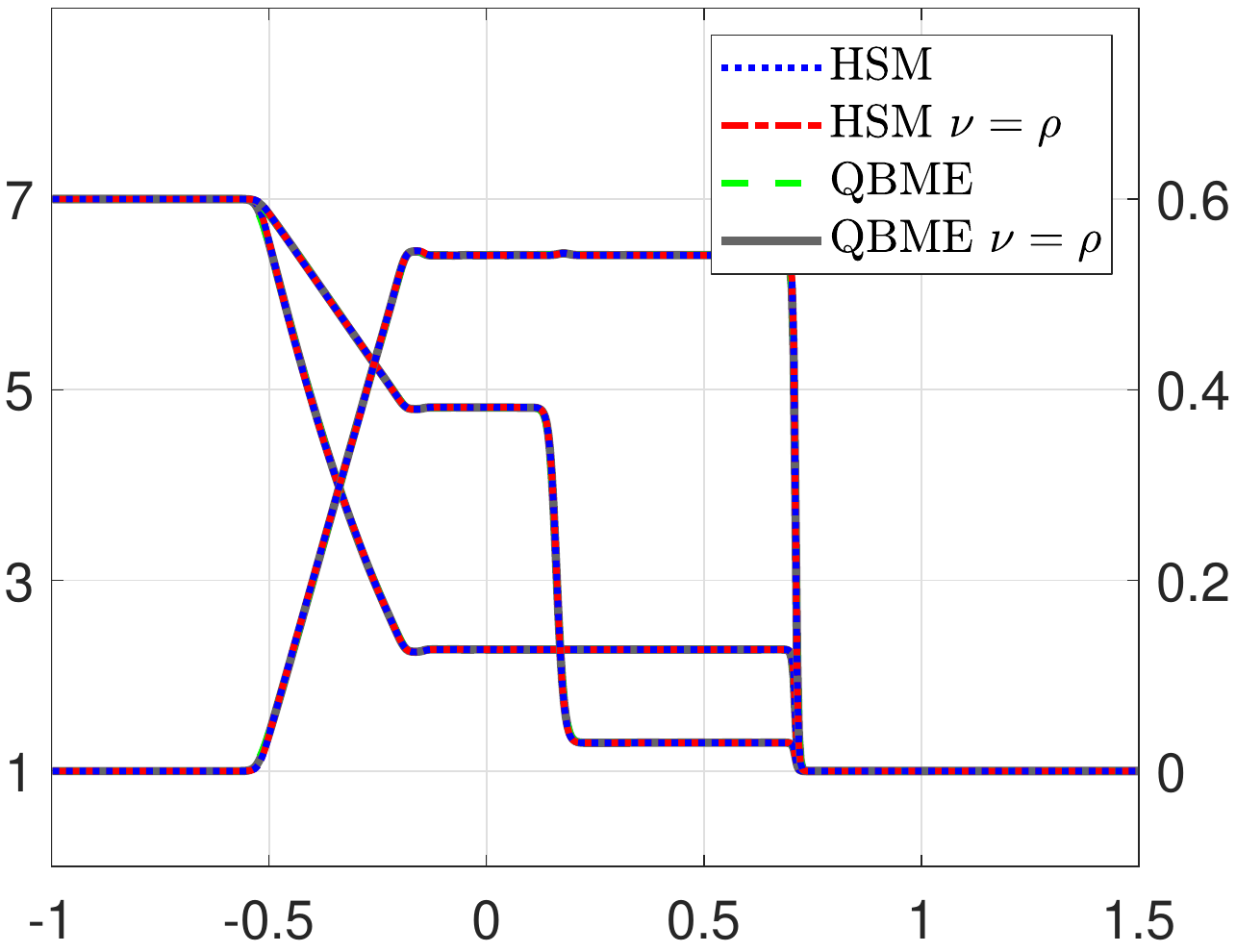}
        \put(0,73){$p,$}
        \put(0,68){$\rho$}
        \put(93,70){$\vel$}
        \put(17,47){$p$}
        \put(25,60){$\rho$}
        \put(58,50){$\vel$}
    \end{overpic}}
    \end{subfigures}
    \caption{Model comparison for shock tube using space-dependent collision frequency $\nu=\rho(x)$ and varying $\tau$.}
    \label{fig:shocktubeNuRhoModel}
\end{figure}

\subsection{Two-beam problem}
The two beam test case was used in \cite{Koellermeier2017} and \cite{Schaerer2015} to investigate different moment models using the 1D BGK equation with constant collision frequency time $\nu=1$. Here we first perform the standard test and then do a variation with piecewise constant collision frequency time.

The initial Riemann data for the left-hand side and the right-hand side of the domain, respectively, is given by
\begin{equation}
    \vect{u}_M^L = \left( 1,0.5,1,0,\ldots,0\right)^T, \quad \quad \vect{u}_M^R = \left( 1,-0.5,1,0,\ldots,0\right)^T,
    \label{e:2beam_IC}
\end{equation}
modeling two colliding Maxwellian distributed particle beams. This test case is especially challenging as it is difficult to represent the analytical solution using a polynomial expansion. In the free streaming case $\textrm{Kn} = \infty$ the analytical solution is a sum of two Maxwellians according to \cite{Schaerer2015}.

The numerical tests are performed on the computational domain $[-10,10]$, discretized using $500$ points and the end time is $t_{\textrm{END}}=0.1$ using a constant macroscopic time step according to a CFL number of $0.5$ for all tests. This results in the same macroscopic time step size $\Delta t = 3.85 \cdot 10^{-4}$ as in the shock tube test case.

The moment models again use $M=9$. A discrete velocity method can be used as reference solution. In \cite{Schaerer2015} a DVM solution was computed using 2000 cells in physical space and 600 variables for the discretization of the microscopic velocity space. Note that the DVM method is computationally much more expensive in comparison to the lower-dimensional moment models described in section \ref{sec:models}. In \cite{Koellermeier2017}, extensive comparisons of the moment models with those reference solutions were made for the rarefied regime. It was obtained that the relative error of the heat flux was only about $4\%$ for a relaxation time of $\tau=0.1$. For a smaller relaxation time as used in all our tests of the current work, the error is naturally even smaller. We thus assume that the model error of the moment model can be neglected and do not show a comparison with the DVM method for these small values of the relaxation time. For more details on the accuracy of moment models for the two-beam model, we refer to \cite{Koellermeier2017}.

\subsubsection{Constant collision frequency \texorpdfstring{$\nu=1$}{}}
For this symmetric test case, we use a constant $\nu=1$ and plot only the left part of the spatial domain. The parameter choices for the different relaxation times $\tau$ can directly be carried over from the previous test case. Figure \ref{fig:2beamNu1QBME} shows the results of the QBME model for first order and third order spatial discretization for pressure $p$ and heat flux $Q$, computed as
\begin{equation}
\label{e:normalized_heat_fluxQBME}
    Q = \frac{6 f_3}{\rho \sqrt{\theta}^3}
\end{equation}
for the QBME model.

Similar to the shock tube test case, we see that there is a significant gain in accuracy when using a third-order spatial discretization. This leads to a possibly coarser distribution of cells and an overall gain in computational efficiency. However, this makes it necessary to use PI earlier, as the macroscopic time step $\Delta t$ is larger. Here we use PI for the cases $\tau=10^{-4}$ and $\tau=10^{-6}$ using $K=1$ and PFE or PRK3, respectively. Even though the graphs for the pressure $p$ in the left column of figure \ref{fig:2beamNu1QBME} seems to be already converged for $\tau=10^{-3}$, we can clearly see in the respective figures for the heat flux $Q$ that there is still a non-equilibrium heat flux present. However, for $\tau=10^{-4}$ and $\tau=10^{-6}$, the models have almost completely converged. This test case shows that PI is indeed necessary to obtain a converged solution for the heat flux.

\begin{figure}[htb!]
    \centering
    \begin{subfigures}
    \subfloat[$p$, first order. \label{fig:2beamQBME10FORCE1p}
    ]{\includegraphics[width=0.5\linewidth]{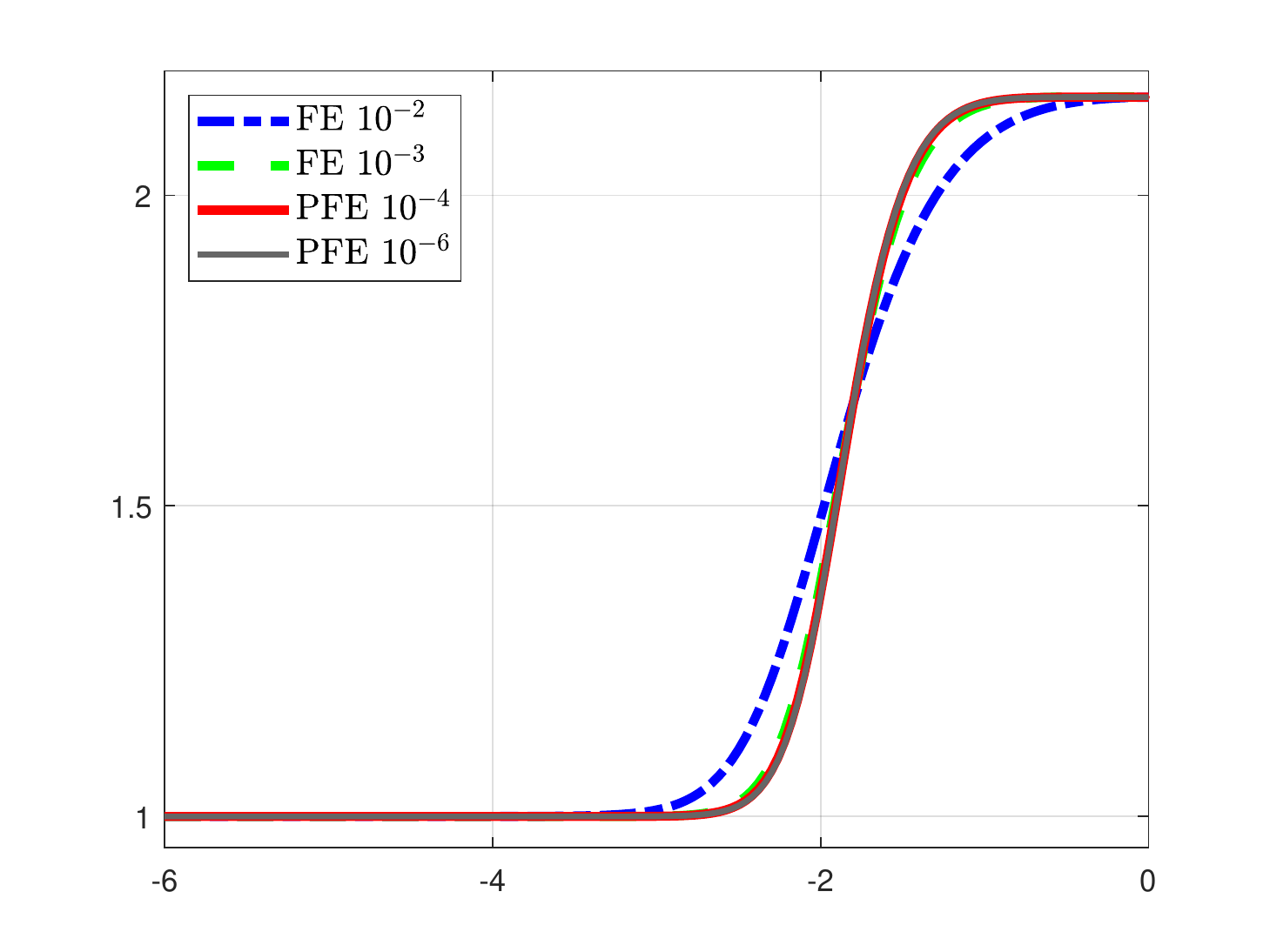}}
    \subfloat[$Q$, first order. \label{fig:2beamQBME10FORCE1Q}
    ]{\includegraphics[width=0.5\linewidth]{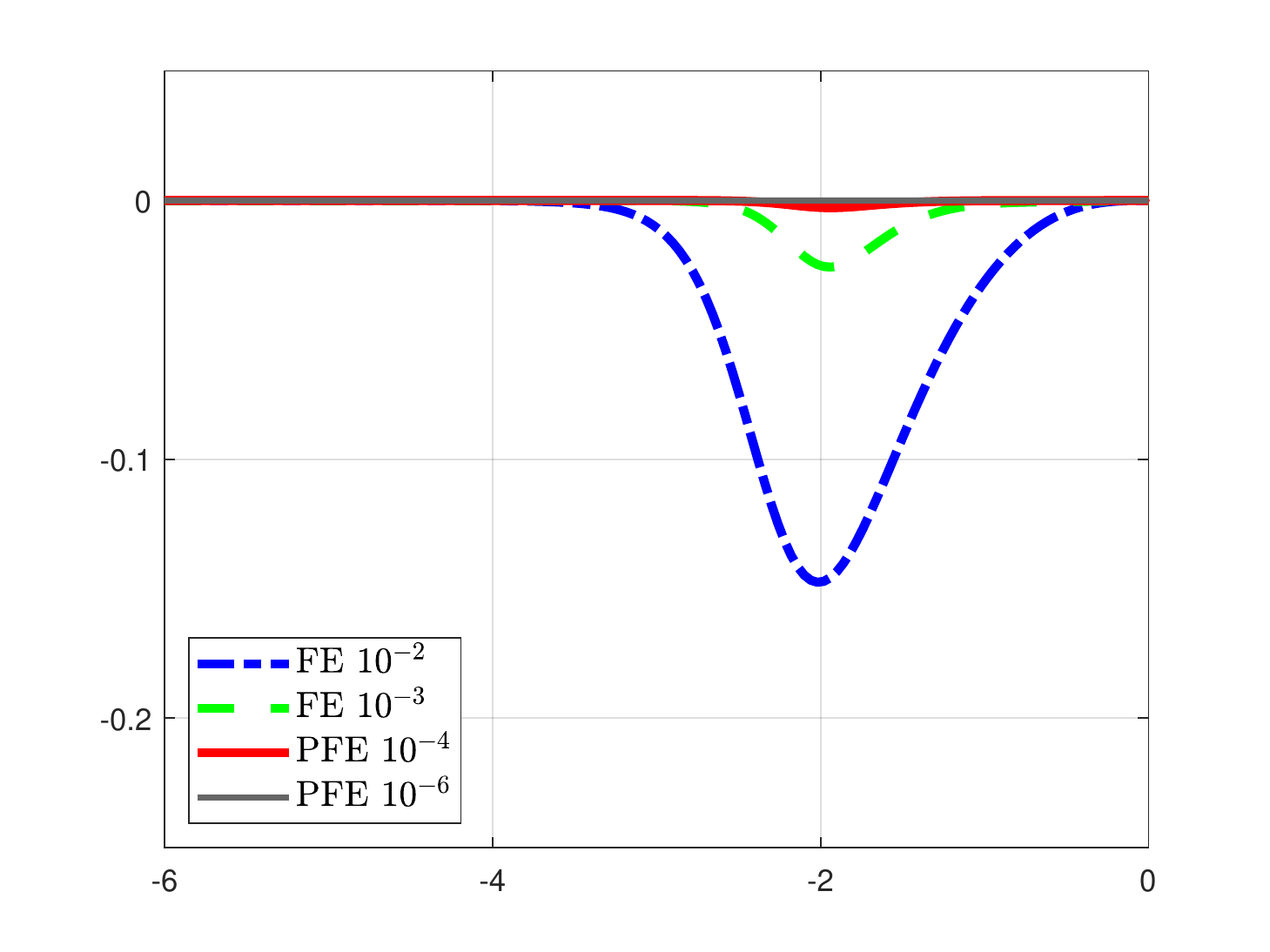}}\\
    \subfloat[$p$, third order. \label{fig:2beamQBME10FORCE3p}
    ]{\includegraphics[width=0.5\linewidth]{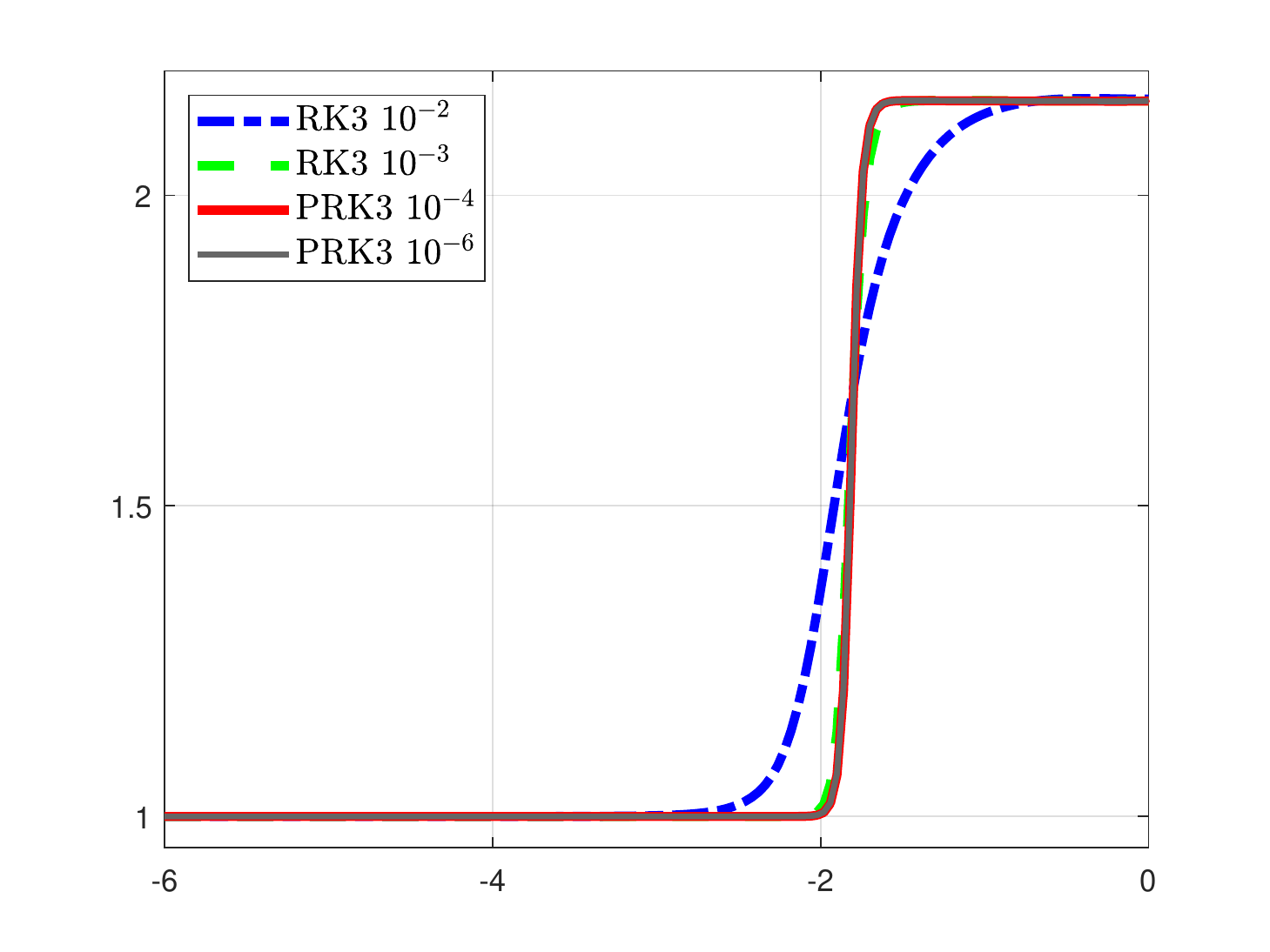}}
    \subfloat[$Q$, third order. \label{fig:2beamQBME10FORCE3Q}
    ]{\includegraphics[width=0.5\linewidth]{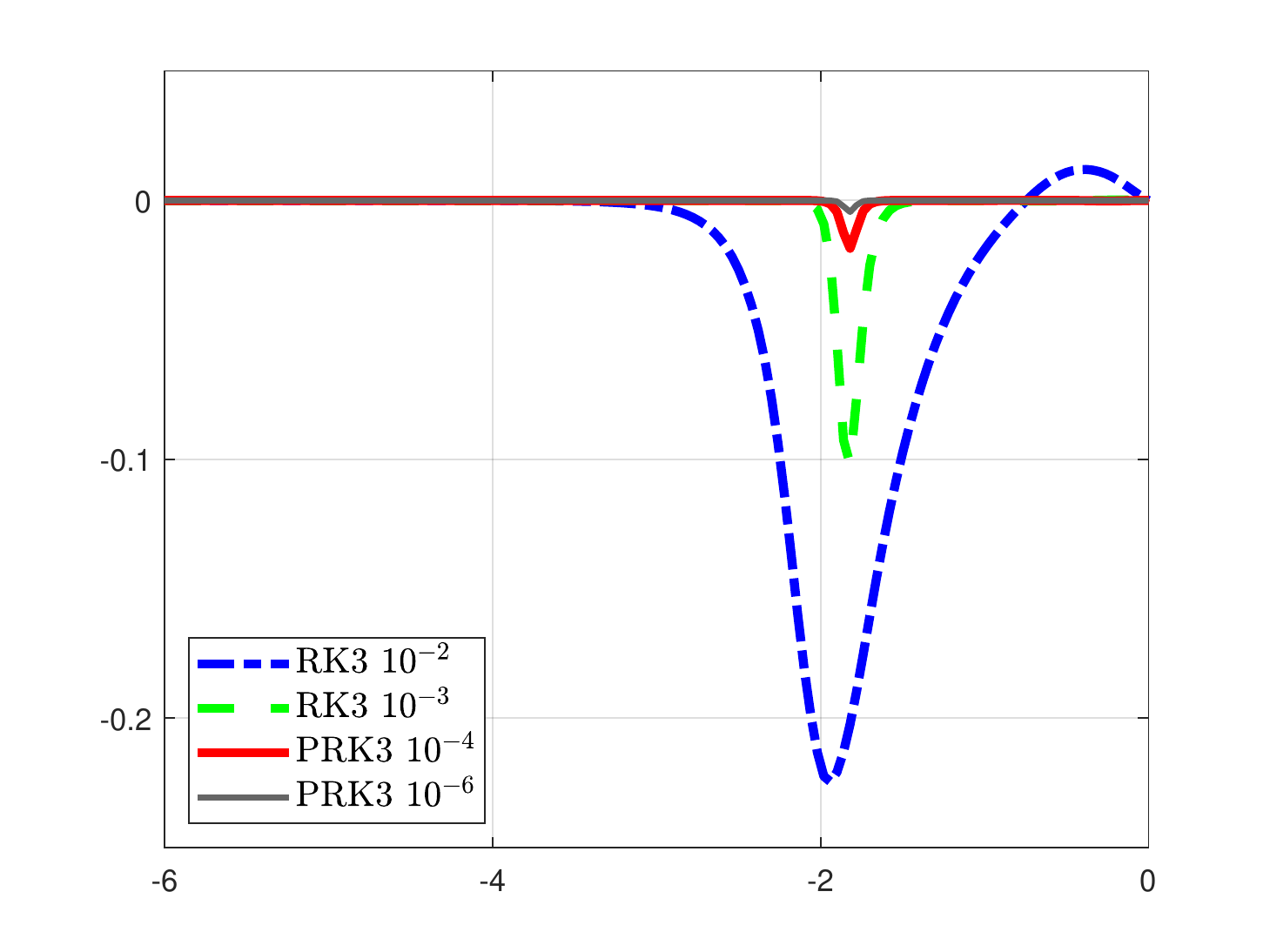}}
    \end{subfigures}
    \caption{Two-beam test for QBME, constant collision frequency $\nu=1$, and varying $\tau$.}
    \label{fig:2beamNu1QBME}
\end{figure}

\subsubsection{Piecewise constant collision frequency \texorpdfstring{$\nu \in \{0.1,1\}$}{}}
Next, we use the QBME model and a piecewise constant collision frequency that has different values $\nu(x<0) = 0.01$ in the left part and $\nu(x>0) = 1$ in the right part of the computational domain. This means that we will have two fast clusters of eigenvalues given by the respective modes in the domain.
This case can be fully described by the stability analysis in section \ref{sec:LSA_piecewise} and the parameters are chosen as follows:
\begin{itemize}
  \item[1.] $\tau = 10^{-2}$: both sides of the domain can be integrated with the standard FE scheme.
  \item[2.] $\tau = 10^{-3}$: A beginning scale separation can be seen according to figure \ref{fig:HSM5FENuPiecewise10Kn0p001} but both fast clusters are still in the region of stability for the time step size $\Delta t$. We can use the FE scheme.
  \item[3.] $\tau = 10^{-4}$: Figure \ref{fig:HSM5FENuPiecewise10Kn0p0001} shows that the fastest scale separated from the remaining two clusters. According to the derivation in section \ref{sec:parameters}, we use the PFE method with $K=1$ and inner time step size $\delta t = 10^{-4}$.
  \item[4.] $\tau \leq 10^{-6}$: Now the clearly separated intermediate cluster requires an additional layer of telescopic PI. The parameter choice is discussed in section \ref{sec:parameters} and we thus choose the TPFE method with $K=1$, $\delta t_0=1\cdot 10^{-6}$, and $\delta t_1=1 \cdot 10^{-4}$. Note, how each level's time step size guarantees the stable integration of one separated cluster.
\end{itemize}

The results shown in figure \ref{fig:2beamNuPiecewise} show a stable solution even for the very small relaxation times $\tau$, despite the large spectral gaps featuring an additional intermediate cluster. The left side of the domain relaxes to the equilibrium solution only for larger relaxation times $\tau$ as the collision frequency $\nu$ is 100 times smaller in this part of the domain. Despite the different propagation speeds due to their different hyperbolic wave structure, the HSM and QBME models give very similar results and we omit a more detailed investigation. The results confirm the observations of the shock tube test case.

\begin{figure}[htb!]
    \centering
    \begin{subfigures}
    \subfloat[QBME9, $p$. \label{fig:2beamPiecewiseQBME10FORCE3p}
    ]{\includegraphics[width=0.5\linewidth]{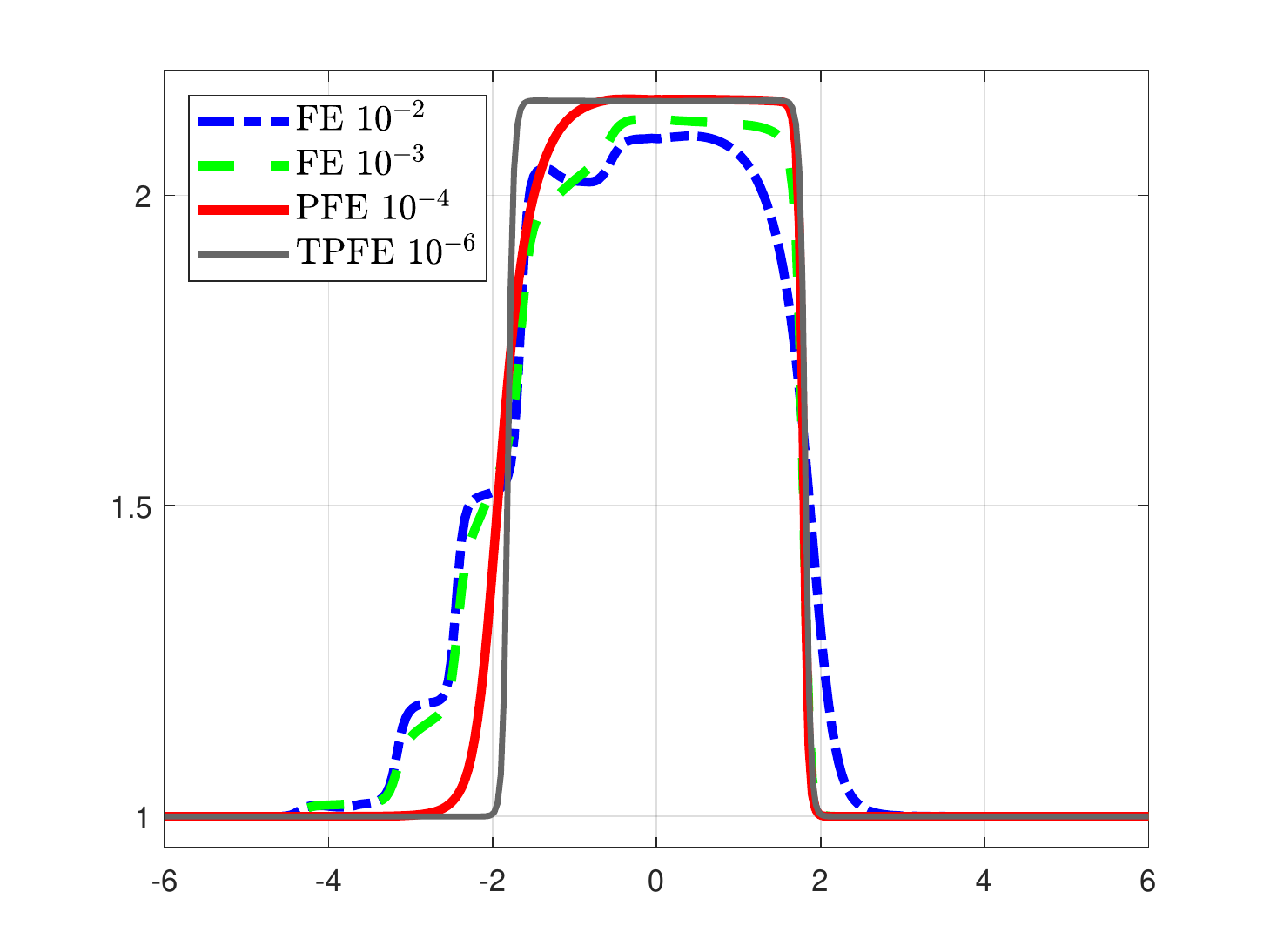}}
    \subfloat[QBME9, $Q$. \label{fig:2beamPiecewiseQBME10FORCE3Q}
    ]{\includegraphics[width=0.5\linewidth]{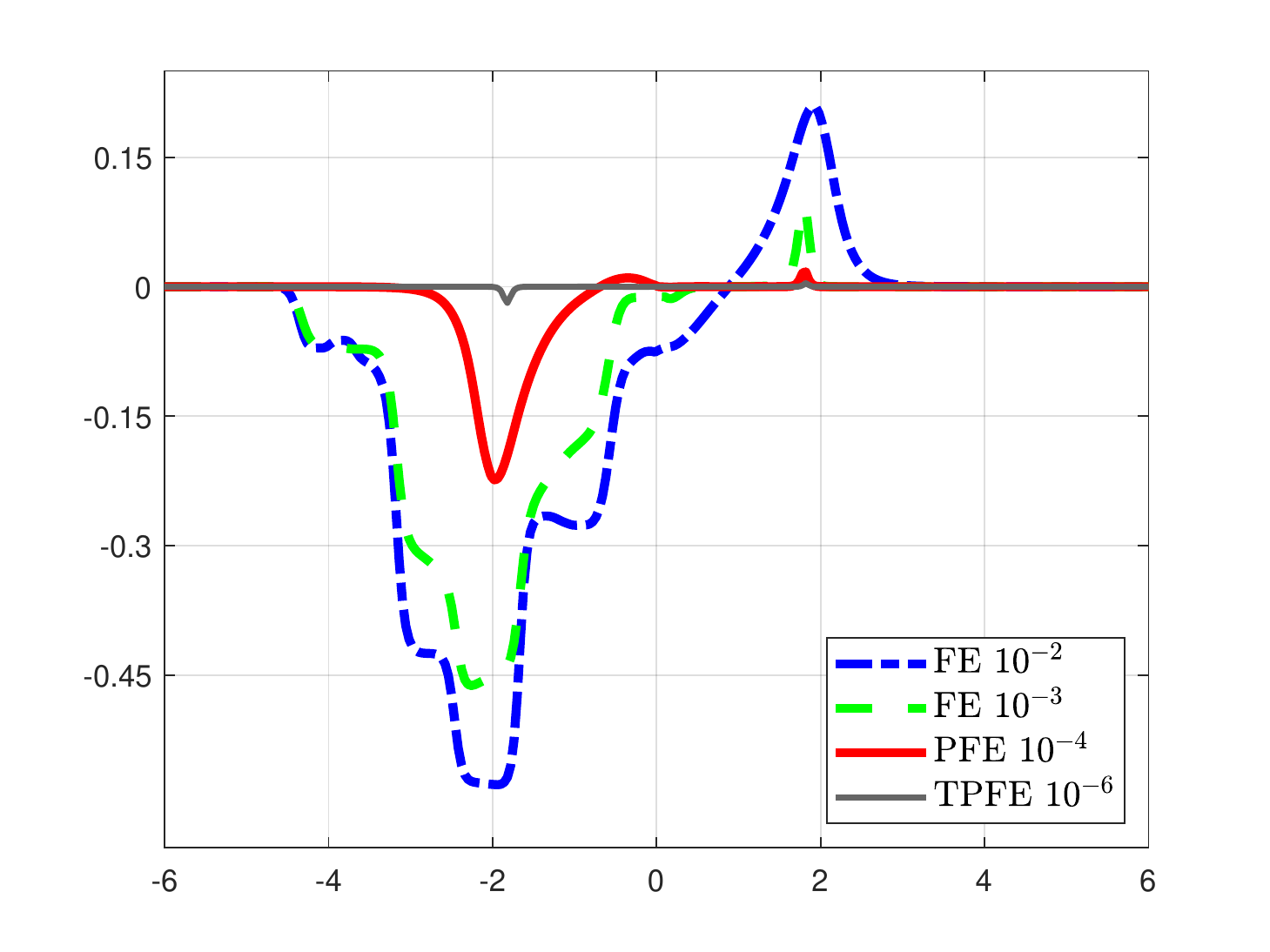}}
    \end{subfigures}
    \caption{Two-beam test for QBME, piecewise constant collision frequency $\nu \in \{0.01,1\}$, third order, and varying $\tau$.}
    \label{fig:2beamNuPiecewise}
\end{figure}

\subsection{Forward facing step}
For a full 2D test case, we present a rarefied supersonic flow over a forward facing step. It has been studied, among others in \cite{Bogolepov1983,Stueer1999} and for the hyperbolic moment models in \cite{Koellermeier2017}.
A flow with Mach number $\textrm{Ma} = 3$ is used at the inlet of a rectangular domain that has a step close to the inlet to generate shock waves.
The two-dimensional domain is composed of an inlet section and a subsequent forward facing step of $20\%$ the height of the inlet section. The domain is shown in figure  \ref{fig:forwardFacingStep}.
\begin{figure}[htb!]
    \centering
    \includegraphics[width=0.6\textwidth]{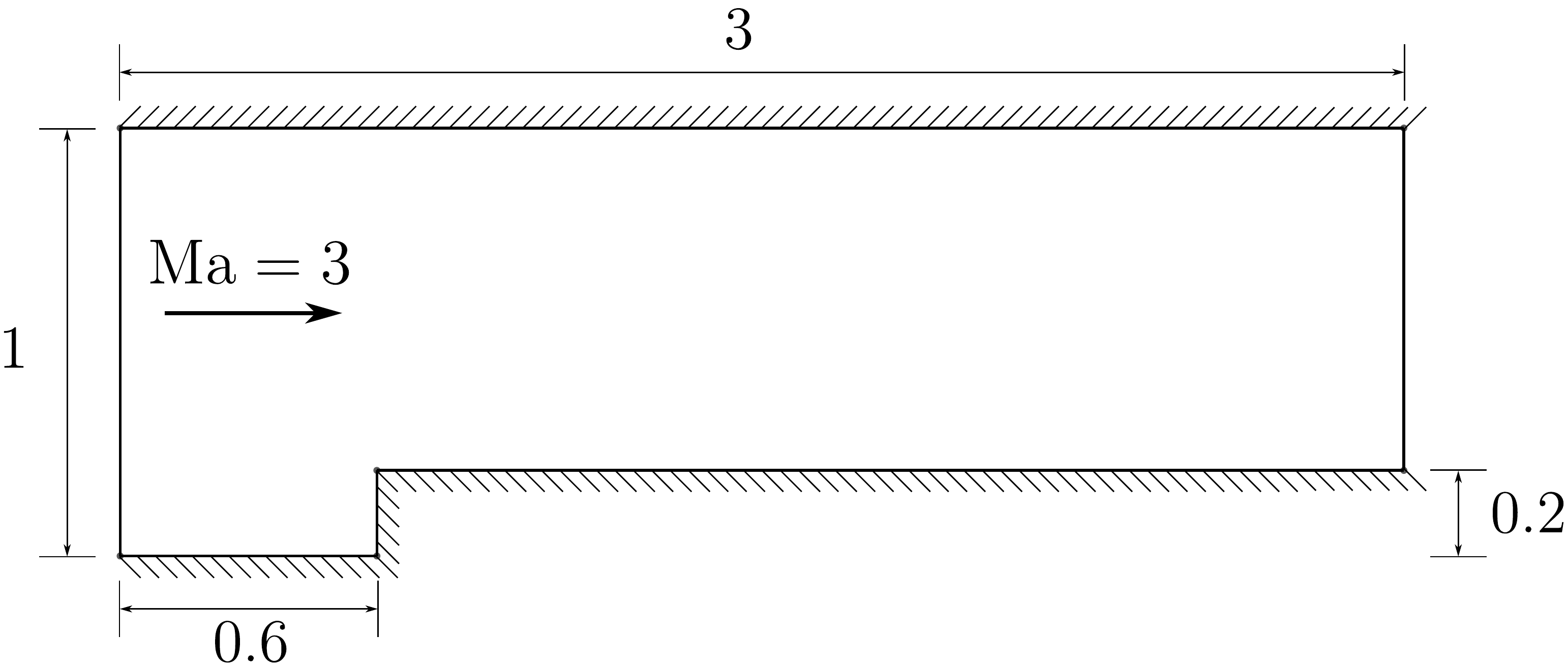}
    \caption{Computational domain for the forward facing step test case, taken from \cite{Koellermeier2018}.}
    \label{fig:forwardFacingStep}
\end{figure}

The flow exhibits a shock, separation and reattachment as well as reflection of the shock wave at the boundaries of the domain. The variety of flow phenomena makes this test case an interesting application future developments of the PI method as there are large parts of purely equilibrium flow allowing for an adaptive use of applying projective integration.

The computational grid is composed of $31,951$ unstructured quadrilateral grid cells and each cell extends over about $\Delta x \approx 0.01$ in one direction. We use the two-dimensional QBME model with $M=3$, see \cite{Koellermeier2018} for the explicit form of the equations. The propagation speeds of the transport part can be evaluated from \eqref{e:HME_EV} to derive the time step size. The macroscopic time step size according to a CFL number of $0.5$ is then $\Delta t = 0.001$. We compute until $t_{\textrm{end}}=6$ and plot the scalar pressure $p$, see \cite{Koellermeier2017} for details.

\subsubsection{BGK collision operator with collision frequency \texorpdfstring{$\nu = \rho$}{}}
We first test the BGK collision operator with space-dependent collision frequency $\nu = \rho(x)$ as described in section \ref{sec:BGK} and analyzed in section \ref{sec:LSA_cont}.

The test case yields a stationary solution and from tests in the hydrodynamic regime as well as in the kinetic regime, we can identify the following range for the density $\rho \in [\rho_{min},\rho_{max}]=[1,10]$. For different relaxation times, we then choose the following time stepping method for a stable integration of all modes:
\begin{itemize}
  \item[1.] $\tau = 10^{-1}$: kinetic regime. There is no scale separation and the standard FE scheme with $\Delta t = 0.001$ is sufficient to yield a stable method.
  \item[3.] $\tau = 10^{-2}$: transitional regime. Due to the large $\Delta t=0.001$ and the non-linear collision frequency $\nu \in [1,10]$, this test case is already beyond the stability region of the standard FE method. We use a PFE method with $K=1$ and $\delta t = 2.5\cdot 10^{-4}$ to obtain stability.
  \item[4.] $\tau = 10^{-3}$: transitional regime. The extended spectrum requires a PFE method with connected stability region and thus the use of a larger $K$. Following the derivation in section \ref{sec:parameters}, we get $\delta t=10^{-4}$ and $K=3$ for an extrapolation factor of $N\approx10$.
  \item[5.] $\tau = 10^{-4}$: hydrodynamic regime. Using section \ref{sec:parameters}, one level of PI is no longer enough. The parameters obtained by the outlined steps yields one additional telescopic level for $K=3$ and $\delta t_0=10^{-5}$, $\delta t_1=10^{-4}$.
\end{itemize}
The chosen parameters ensure stability for all cases and yield the results shown in figure \ref{fig:ffsBGK}. The results can be compared to the forward facing step simulations in \cite{Koellermeier2018}, where a different right-hand side treatment was used. We see a clear agreement of both methods. The result for large relaxation time $\tau=10^{-1}$ in figure \ref{fig:0p01QBMEM3Kn0p1FEt1} shows a clearly smoother shock profile, while the shock becomes more and more pronounced with decreasing relaxation time $\tau$. Despite the different numerical methods, we do not see evidence that the macroscopic solution is spoiled by additional diffusion or wrong propagation speeds, which shows that the PI approach leads to a consistent solution towards the hydrodynamic limit.

\begin{figure}[htbp!]
    \centering
    \begin{subfigures}
    \subfloat[BGK $\tau=10^{-1}$, FE. \label{fig:0p01QBMEM3Kn0p1FEt1}
    ]{\begin{overpic}[width=0.95\textwidth]{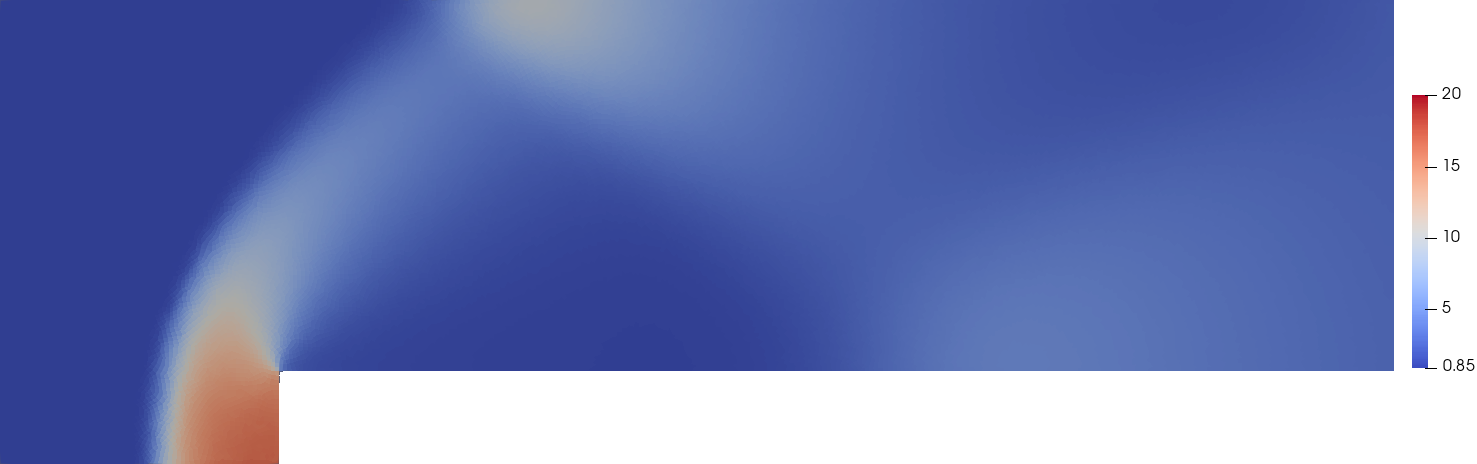}
        \put(96,27){$p$}
    \end{overpic}}\\
    \subfloat[BGK $\tau=10^{-2}$, PFE. \label{fig:0p01QBMEM3Kn0p01PFEK1deltaT0p00025t1}
    ]{\begin{overpic}[width=0.95\textwidth]{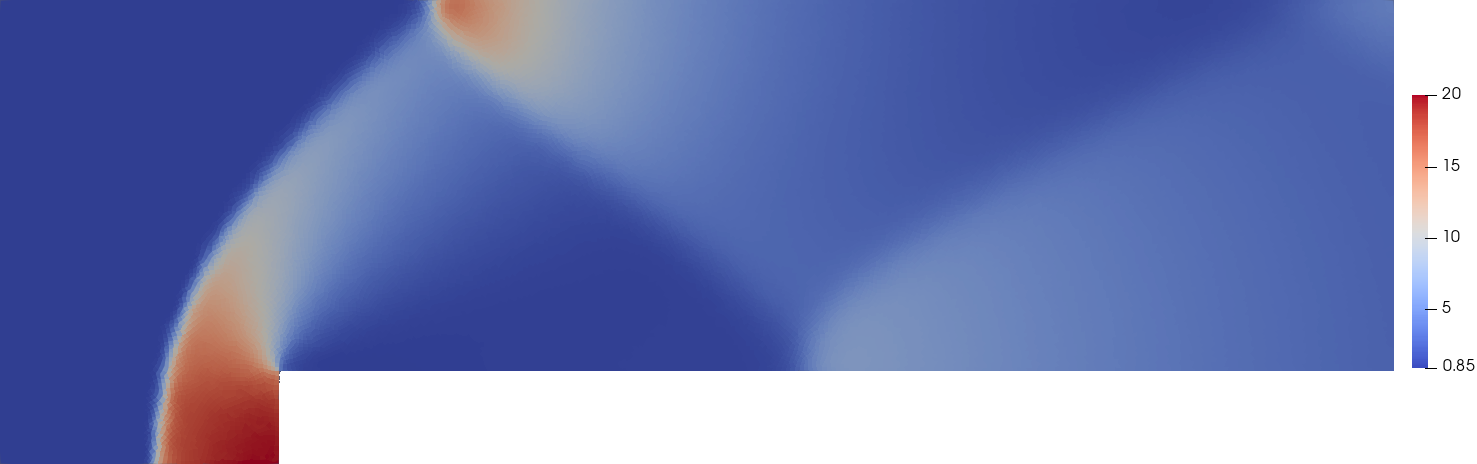}
        \put(96,27){$p$}
    \end{overpic}}\\
    \subfloat[BGK $\tau=10^{-3}$, PFE. \label{fig:0p01QBMEM3Kn0p001PFEK3deltaT0p0001t1}
    ]{\begin{overpic}[width=0.95\textwidth]{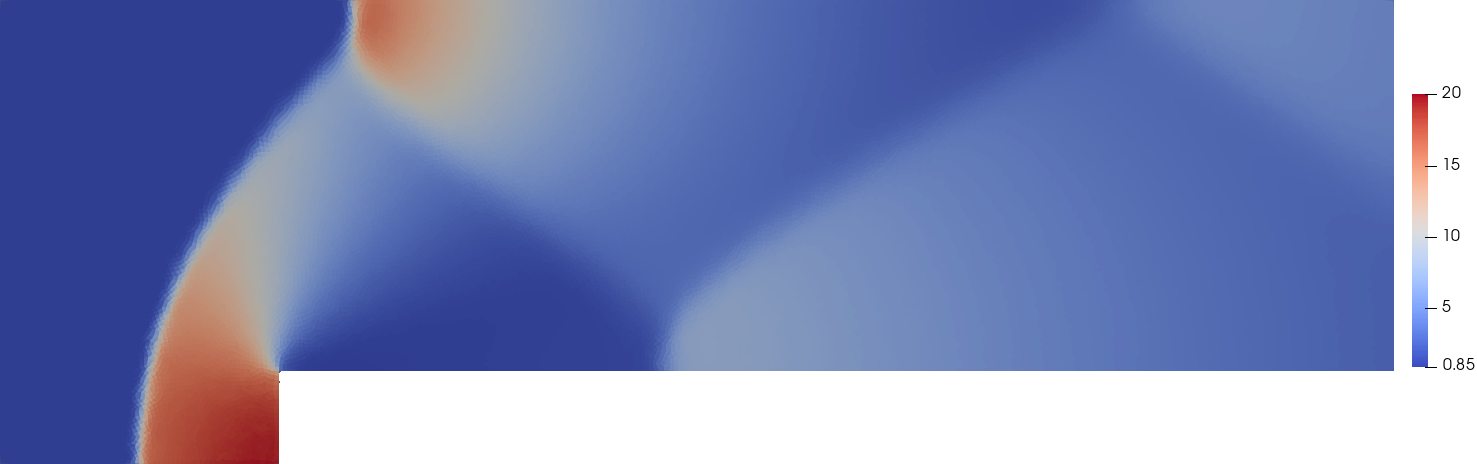}
        \put(96,27){$p$}
    \end{overpic}}\\
    \subfloat[BGK $\tau=10^{-4}$, TPFE. \label{fig:0p01QBMEM3Kn0p0001TPFEK3IntdeltaT0p0001InnerdeltaT0p00001t1}
    ]{\begin{overpic}[width=0.95\textwidth]{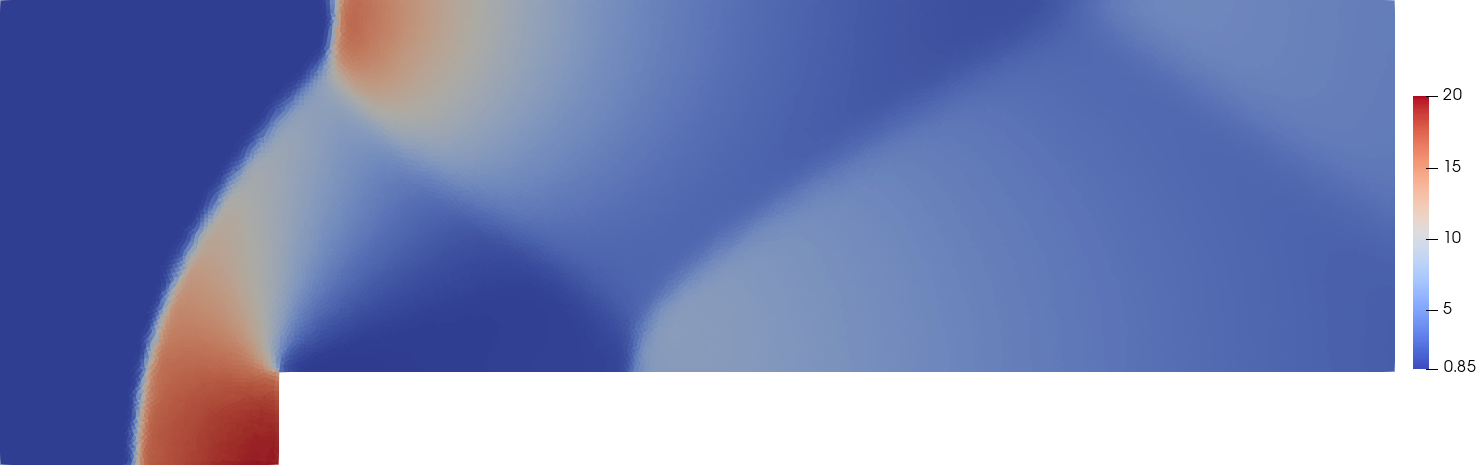}
        \put(96,27){$p$}
    \end{overpic}}
    \end{subfigures}
    \caption{Forward facing step for QBME, BGK collision operator, and varying $\tau$.}
    \label{fig:ffsBGK}
\end{figure}

\subsubsection{Boltzmann collision operator}
For the Boltzmann collision operator \eqref{e:Boltzmann_collision} on the right-hand side of the kinetic equation \eqref{e:BTE}, the spectrum is also extended and exhibits approximately the same maximum and minimum values for the fastest modes. This is why it is possible to choose the same parameters as chosen for the BGK case for the integration of the semi-discrete system. We use $b_0 = \frac{1}{2 \pi}$ for the collision kernel to obtain a collision frequency $\nu = \rho$ in equation \eqref{e:collision_frequency}. The collision frequency then resembles the collision frequency of the BGK operator and the same PI parameters can be used.

The results for the different relaxation times are shown in figure \ref{fig:ffsBTE}. We see a very good agreement of the hydrodynamic case $\tau=10^{-4}$ in figure \ref{fig:0p01BTEQBMEM3Kn0p0001TPFEK3IntdeltaT0p0001InnerdeltaT0p00001t1} compared to the BGK case in figure \ref{fig:0p01QBMEM3Kn0p0001TPFEK3IntdeltaT0p0001InnerdeltaT0p00001t1}. This means that the hydrodynamic limit is computed correctly by the chosen TPFE method. In the transitional and kinetic regime, the position of the shock is slightly further downstream, which is due to the differences between the Boltzmann collision operator and the BGK operator. However, the chosen time stepping methods are able to obtain a stable solution nevertheless.
\begin{figure}[htbp!]
    \centering
    \begin{subfigures}
    \subfloat[Boltzmann $\tau=10^{-1}$, FE. \label{fig:0p01BTEQBMEM3Kn0p1FEt1_pi}
    ]{\begin{overpic}[width=0.95\textwidth]{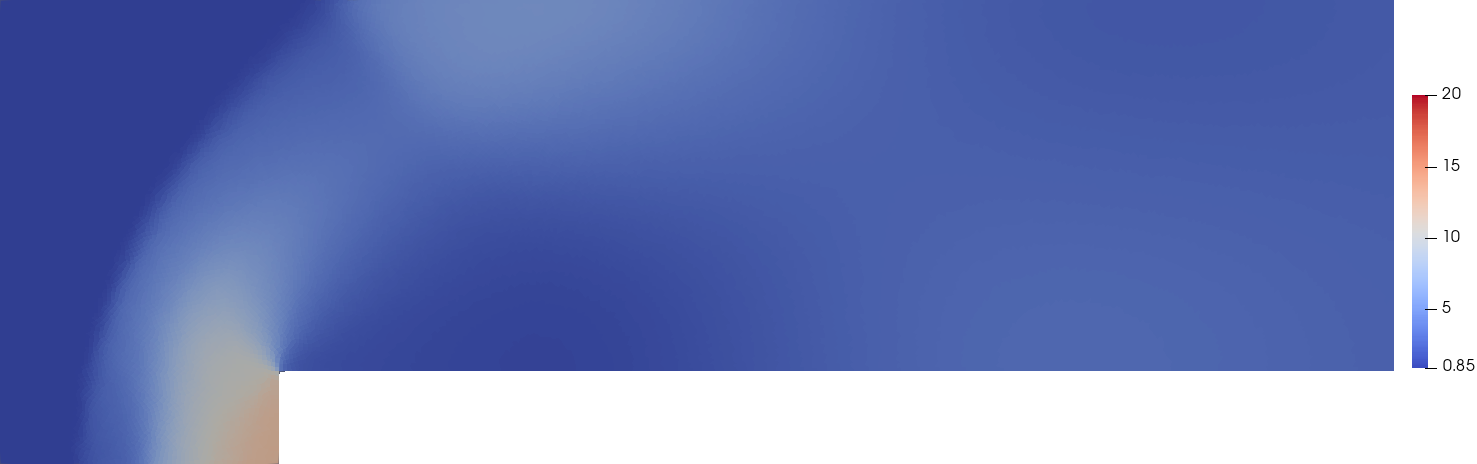}
        \put(96,27){$p$}
    \end{overpic}}\\
    \subfloat[Boltzmann $\tau=10^{-2}$, PFE. \label{fig:0p01BTEQBMEM3Kn0p01PFEK1deltaT0p00025t1}
    ]{\begin{overpic}[width=0.95\textwidth]{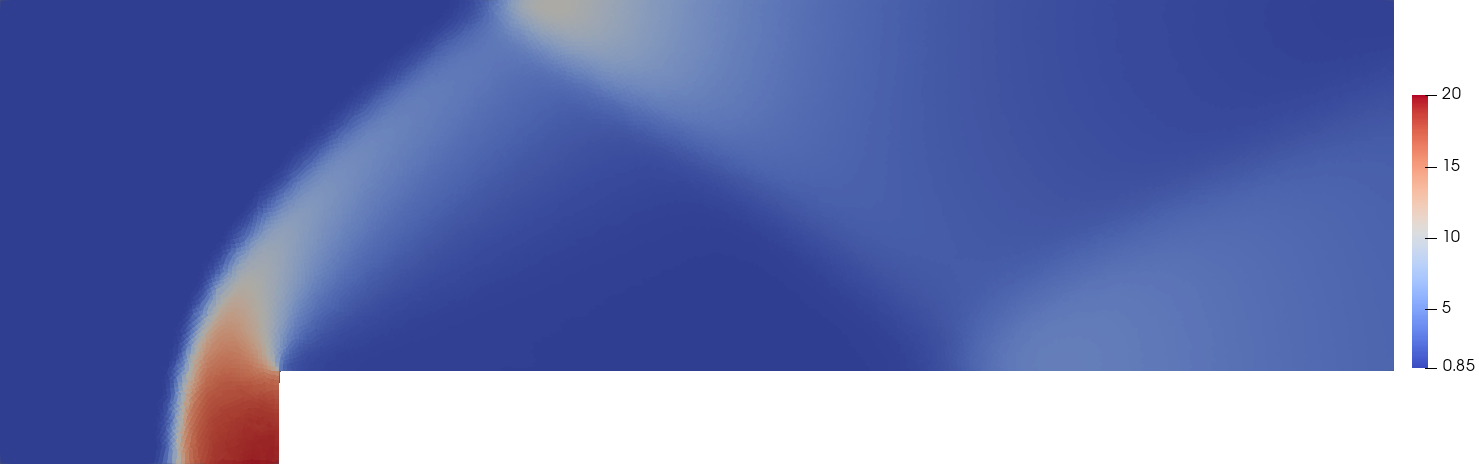}
        \put(96,27){$p$}
    \end{overpic}}\\
    \subfloat[Boltzmann $\tau=10^{-3}$, PFE. \label{fig:0p01BTEQBMEM3Kn0p001PFEK3deltaT0p0001t1}
    ]{\begin{overpic}[width=0.95\textwidth]{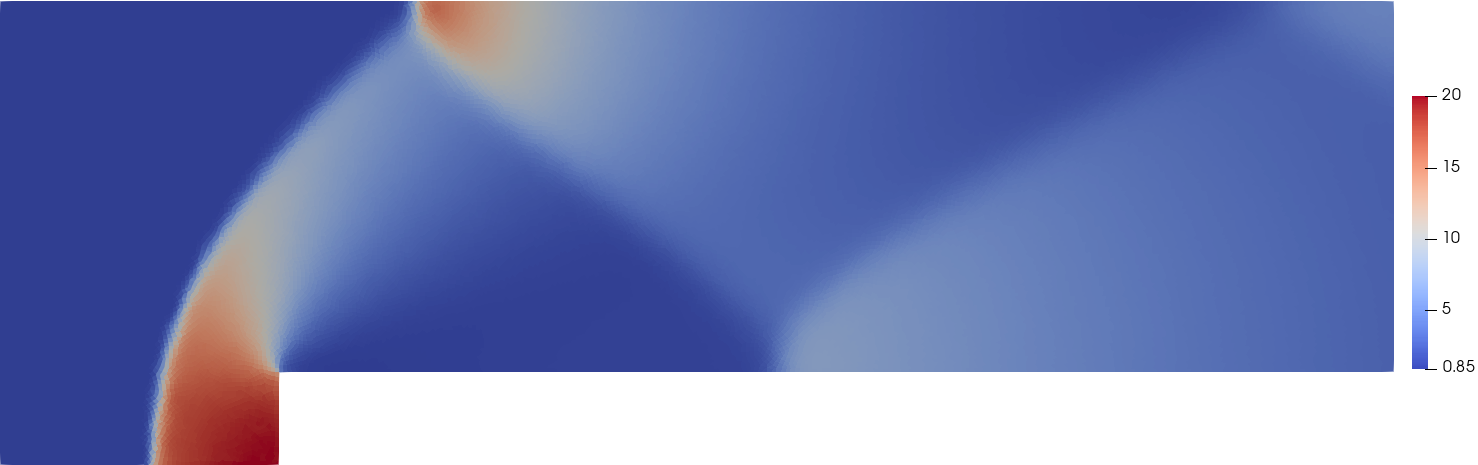}
        \put(96,27){$p$}
    \end{overpic}}\\
    \subfloat[Boltzmann $\tau=10^{-4}$, TPFE. \label{fig:0p01BTEQBMEM3Kn0p0001TPFEK3IntdeltaT0p0001InnerdeltaT0p00001t1}
    ]{\begin{overpic}[width=0.95\textwidth]{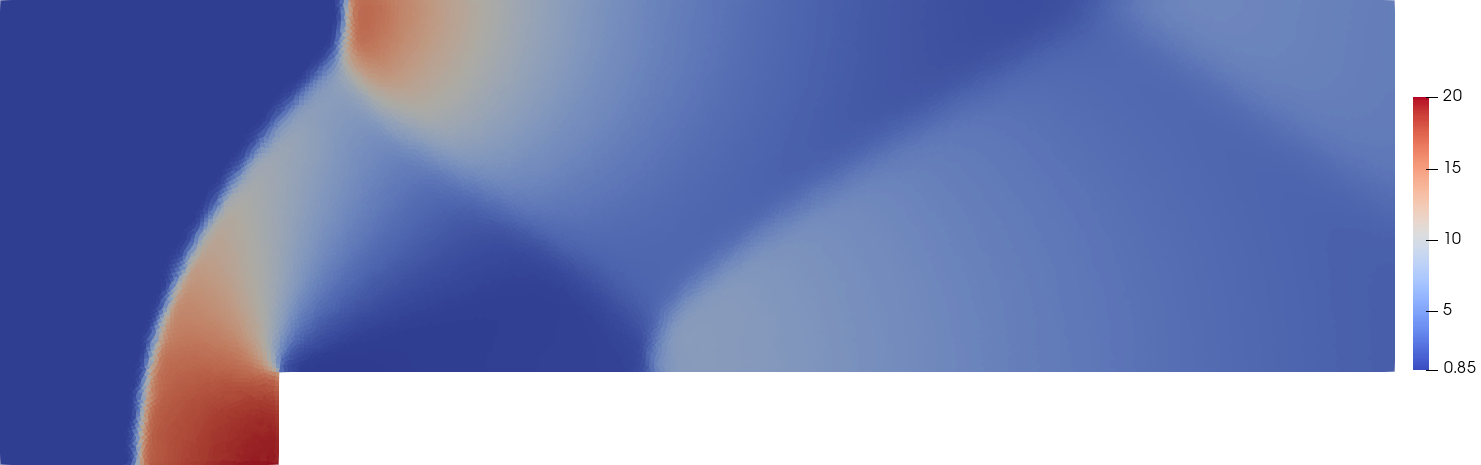}
        \put(96,27){$p$}
    \end{overpic}}
    \end{subfigures}
    \caption{Forward facing step for QBME, Boltzmann collision operator, and varying $\tau$.}
    \label{fig:ffsBTE}
\end{figure}

\subsection{Speedup computation}
The speedup of the different computations shown in this section can be computed with the formulas given in section \ref{sec:speedup} for the PI or TPI, respectively. Table \ref{tab:speedup} shows the speedup of all methods used throughout this paper. No speedup needs to be obtained, where the standard CFL-type macroscopic time step is sufficient to achieve stability. However, a significant speedup can be achieved for values of the relaxation time $\tau$ closer to the hydrodynamic regime. The speedup increases with the number of time steps a naive FE method would require, as these scale linearly with the relaxation time $\tau$, while the PI method uses a constant number of time steps and the number of time steps for the TPI method only depends mildly on relaxation time because the number of levels increases slowly, see section \ref{sec:parameters}. Note that the values in {\color{gray}{gray}} correspond to simulations that are not shown in the figures for conciseness. The parameters of those methods can nevertheless be obtained by a straightforward application of the selection process in section \ref{sec:parameters} followed by the respective speedup computation.
\begin{table}[H]
    \centering
    \caption{Speedup of (T)PI schemes in comparison to standard FE scheme.}
    \label{tab:speedup}
    \begin{tabular}{c||c|c|c|c|c}
      relaxation time $\tau$    & $10^{-2}$ & $10^{-3}$ & $10^{-4}$ & $10^{-5}$ & $10^{-6}$ \\ \hline
      shock tube $\nu = 1$      & 1 & 1 & 1.925 & 19.25 & 192.25 \\ \hline
      shock tube $\nu = \rho$   & 2 & 1.375 & 3.93 & 5.61 & \color{gray}{8.02} \\ \hline
      two-beam $\nu = 1$        & 1 & 1 & 1.925 & \color{gray}{19.25} & 192.25 \\ \hline
      two-beam $\nu = \nu_i$    & 1 & 1 & 1.925 & \color{gray}{9.625} & 96.25 \\ \hline
      forward facing step       & 2 & 2.5 & 6.25 & \color{gray}{15.63} & \color{gray}{39.06} \\
    \end{tabular}
\end{table}
The significant speedup of PI schemes towards the hydrodynamic limit justifies the investigation and use of this explicit time integration method for moment models.
\section{Conclusion}
\label{sec:conclusion}
In this paper, we carried out the first application of explicit PI schemes for hyperbolic moment models out to overcome the stiffness of the right-hand side collision operator in the transitional and hydrodynamic regime.

After introduction of the model equations and both the full Boltzmann collision operator and the simplified BGK collision operator, a linear stability analysis of the linearized Hermite Spectral Method allowed for a detailed understanding of the relation between the fast microscopic scales and the relaxation time as well as the collision frequency of the model. Based on the analysis for constant, piecewise constant, and space-dependent collision frequency, we described the choice of a stable time integrator and outlined explicit steps to choose the parameters in all occurring cases.

In numerical simulations of a 1D shock tube, a 1D two-beam problem, and a 2D forward facing step test case the stability of the algorithms and the convergence of the moment models towards the hydrodynamic limit could be demonstrated. We showed results for high-order spatial discretization in combination with PFE, PRK and TPFE methods with different parameter ranges. Especially the case of an extended eigenvalue spectrum was covered by constructing an A-stable two-stage TPFE method for the moment model.

The combination of moment model and PI methods achieves large accelerations in runtime of up to almost $200$ in comparison to a standard explicit Euler scheme. In addition, fewer variables than for a standard DVM method are necessary, thus combining an efficient model with a high-fidelity solution method.

We could show that the full non-linear model and the non-linear relaxation time show small differences in the relaxation scheme, but converge to the same hydrodynamic limit.

The work in this paper opens up possibilities for many further advancements: In most applications including those in this paper the non-equilibrium is confined to a small portion of the computational domain. This makes an adaptive selection of the time stepping method desirable to allow for an efficient simulation in each respective region. Furthermore, the extension towards a time-adaptive PI method seems promising to speed up time-accurate simulations. After the successful application of PI for moment models, other acceleration methods like the micro-Macro splitting in \cite{Debrabrant2017} can be applied. Lastly, investigation of other collision terms or types of equations might lead to faster simulations for other application cases, too.

\section*{Acknowledgement}
This project has received funding from the European Union’s Horizon 2020 research and innovation programme under the Marie Sklodowska-Curie grant agreement no. 888596.
The first author is a postdoctoral fellow in fundamental research of the Research Foundation - Flanders (FWO), funded by FWO grant no. 0880.212.840.

\appendix
\section{2D QBME model equations}
\label{app:2D_QBME}

The terms of the 2D QBME as derived in \cite{Koellermeier2014a} and written in explicit form first in \cite{Koellermeier2018} are given by

\begin{equation}\label{e4:QBME_x_f_3}
\Vect{A}_x = \quad\quad\quad\quad\quad\quad\quad\quad\quad\quad\quad\quad\quad\quad\quad\quad\quad\quad\quad\quad\quad\quad\quad\quad\quad\quad\quad\quad\quad\quad\quad
\end{equation}
\tiny{
\begin{equation*}
    \setlength{\arraycolsep}{1pt}
    \left(
    \begin{array}{cccccccccc}
     u_x & \rho  & 0 & 0 & 0 & 0 & 0 & 0 & 0 & 0\\
     0 & u_x & 0 & \frac{2}{\rho } & 0 & 0 & 0 & 0 & 0 & 0 \\
     0 & 0 & u_x & 0 & \frac{1}{\rho } & 0 & 0 & 0 & 0 & 0 \\
     \frac{15 f_{3,0}}{2 \rho } & \frac{3 p_1}{2} & 0 & u_x-\frac{15 f_{3,0}}{\tilde{p}} & 0 & -\frac{15 f_{3,0}}{\tilde{p}} & 3 & 0 & 0 & 0 \\
     \frac{5 f_{2,1}}{\rho } & 2 f_{1,1} & p_1 & -\frac{10 f_{2,1}}{\tilde{p}} & u_x & -\frac{10 f_{2,1}}{\tilde{p}} & 0 & 2 & 0 & 0 \\
     \frac{5 f_{1,2}}{2 \rho } & \frac{p_2}{2} & f_{1,1} & -\frac{5 f_{1,2}}{\tilde{p}} & 0 & u_x-\frac{5 f_{1,2}}{\tilde{p}} & 0 & 0 & 1 & 0 \\
     \frac{\tilde{p}{}^3-100 \rho  f_{3,0} \tilde{q}_1}{-8 \tilde{p} \rho ^2} & \frac{5 p_1 f_{3,0}}{\tilde{p}} & \frac{5 f_{1,1}
       f_{3,0}}{\tilde{p}} & \frac{p_2}{\rho }-\frac{25 f_{3,0} \tilde{q}_1}{\tilde{p}^2} & 0 & -\frac{25 f_{3,0} \tilde{q}_1}{\tilde{p}^2} & \frac{15 f_{3,0}}{\tilde{p}}+u_x & 0 & \frac{5 f_{3,0}}{\tilde{p}} & 0 \\
     \frac{25 f_{2,1} \tilde{q}_1}{2 \tilde{p} \rho } & \frac{5 p_1 f_{2,1}}{\tilde{p}} & \frac{5 f_{1,1} f_{2,1}}{\tilde{p}} & -\frac{25 f_{2,1}
       \tilde{q}_1}{\tilde{p}^2}-\frac{2 f_{1,1}}{\rho } & \frac{p_1+3 p_2}{4 \rho } & -\frac{25 f_{2,1} \tilde{q}_1}{\tilde{p}^2} & \frac{15 f_{2,1}}{\tilde{p}} & u_x & \frac{5 f_{2,1}}{\tilde{p}} & 0 \\
     \frac{\tilde{p}^3-100 \rho  f_{1,2} \tilde{q}_1}{-8 \tilde{p} \rho ^2} & \frac{5 p_1 f_{1,2}}{\tilde{p}} & \frac{5 f_{1,1}
       f_{1,2}}{\tilde{p}} & \frac{p_1-p_2}{2 \rho }-\frac{25 f_{1,2} \tilde{q}_1}{\tilde{p}^2} & -\frac{f_{1,1}}{\rho } & \frac{\tilde{p}}{2 \rho }-\frac{25 f_{1,2} \tilde{q}_1}{\tilde{p}^2} & \frac{15 f_{1,2}}{\tilde{p}} & 0 & \frac{5 f_{1,2}}{\tilde{p}}+u_x & 0 \\
     \frac{25 f_{0,3} \tilde{q}_1}{2 \tilde{p} \rho } & \frac{5 p_1 f_{0,3}}{\tilde{p}} & \frac{5 f_{0,3} f_{1,1}}{\tilde{p}} & -\frac{25 f_{0,3}
       \tilde{q}_1}{\tilde{p}^2} & \frac{p_1-p_2}{4 \rho } & -\frac{25 f_{0,3} \tilde{q}_1}{\tilde{p}^2} & \frac{15 f_{0,3}}{\tilde{p}} & 0 & \frac{5 f_{0,3}}{\tilde{p}} & u_x \\
    \end{array}
    \right)
    \setlength{\arraycolsep}{6pt}
\end{equation*}
}
\normalsize
and
\begin{equation}\label{e4:QBME_y_f_3}
\Vect{A}_y = \quad\quad\quad\quad\quad\quad\quad\quad\quad\quad\quad\quad\quad\quad\quad\quad\quad\quad\quad\quad\quad\quad\quad\quad\quad\quad\quad\quad\quad\quad
\end{equation}
\tiny{
\begin{equation*}
    \setlength{\arraycolsep}{1pt}
    \left(
    \begin{array}{cccccccccc}
     u_y & 0 & \rho  & 0 & 0 & 0 & 0 & 0 & 0 & 0 \\
     0 & u_y & 0 & 0 & \frac{1}{\rho } & 0 & 0 & 0 & 0 & 0 \\
     0 & 0 & u_y & 0 & 0 & \frac{2}{\rho } & 0 & 0 & 0 & 0 \\
     \frac{5 f_{2,1}}{2 \rho } & f_{1,1} & \frac{p_1}{2} & u_y-\frac{5 f_{2,1}}{\tilde{p}} & 0 & -\frac{5 f_{2,1}}{\tilde{p}} & 0 & 1 & 0 & 0 \\
     \frac{5 f_{1,2}}{\rho } & p_2 & 2 f_{1,1} & -\frac{10 f_{1,2}}{\tilde{p}} & u_y & -\frac{10 f_{1,2}}{\tilde{p}} & 0 & 0 & 2 & 0 \\
     \frac{15 f_{0,3}}{2 \rho } & 0 & \frac{3 p_2}{2} & -\frac{15 f_{0,3}}{\tilde{p}} & 0 & u_y-\frac{15 f_{0,3}}{\tilde{p}} & 0 & 0 & 0 & 3 \\
     \frac{25 \tilde{q}_2 f_{3,0}}{2 \tilde{p} \rho } & \frac{5 f_{1,1} f_{3,0}}{\tilde{p}} & \frac{5 p_2 f_{3,0}}{\tilde{p}} & -\frac{25 \tilde{q}_2 f_{3,0}}{\tilde{p}^2} & \frac{p_2-p_1}{4 \rho } & -\frac{25 \tilde{q}_2 f_{3,0}}{\tilde{p}^2} & u_y & \frac{5 f_{3,0}}{\tilde{p}} & 0 & \frac{15 f_{3,0}}{\tilde{p}} \\
     -\frac{\tilde{p}^3-100 \rho  f_{2,1} \tilde{q}_2}{8 \tilde{p} \rho ^2} & \frac{5 f_{1,1} f_{2,1}}{\tilde{p}} & \frac{5 p_2
       f_{2,1}}{\tilde{p}} & \frac{\tilde{p}}{2 \rho }-\frac{25 f_{2,1} \tilde{q}_2}{\tilde{p}^2} & -\frac{f_{1,1}}{\rho } & \frac{p_2-p_1}{2 \rho }-\frac{25 f_{2,1} \tilde{q}_2}{\tilde{p}^2} & 0 & \frac{5 f_{2,1}}{\tilde{p}}+u_y & 0 & \frac{15 f_{2,1}}{\tilde{p}} \\
     \frac{25 f_{1,2} \tilde{q}_2}{2 \tilde{p} \rho } & \frac{5 f_{1,1} f_{1,2}}{\tilde{p}} & \frac{5 p_2 f_{1,2}}{\tilde{p}} & -\frac{25 f_{1,2} \tilde{q}_2}{\tilde{p}^2} & \frac{3 \tilde{p}}{4 \rho } & -\frac{25 f_{1,2} \tilde{q}_2}{\tilde{p}^2}-\frac{2 f_{1,1}}{\rho } & 0 & \frac{5 f_{1,2}}{\tilde{p}} & u_y & \frac{15 f_{1,2}}{\tilde{p}} \\
     -\frac{\tilde{p}^3-100 \rho  f_{0,3} \tilde{q}_2}{8 \tilde{p} \rho ^2} & \frac{5 f_{0,3} f_{1,1}}{\tilde{p}} & \frac{5 p_2
       f_{0,3}}{\tilde{p}} & -\frac{25 f_{0,3} \tilde{q}_2}{\tilde{p}^2} & 0 & \frac{p_1}{\rho }-\frac{25 f_{0,3} \tilde{q}_2}{\tilde{p}^2} & 0 & \frac{5 f_{0,3}}{\tilde{p}} & 0 & \frac{15 f_{0,3}}{\tilde{p}}+u_y \\
    \end{array}
    \right)
    \setlength{\arraycolsep}{6pt}
\end{equation*}
}
\normalsize

for $\tilde{p} = p_1 + p_2$, $\tilde{q}_1 = 3 f_{3,0}+f_{1,2}$ and $\tilde{q}_2 = 3 f_{0,3}+f_{2,1}$.

The right-hand side collision term $\Vect{S}$ for the BGK model \cite{Bhatnagar1954} reads
\begin{equation}
\label{e4:vars_collision_full}
    \Vect{S}(\Var_{\textrm{f},3}) = - \frac{1}{\tau} \left( 0, 0, 0, f_{2,0}, f_{1,1}, f_{0,2}, f_{3,0}, f_{2,1}, f_{1,2}, f_{0,3}\right)^T.
\end{equation}

\bibliographystyle{plain}
\bibliography{PI_paper}

\end{document}